\documentclass[letterpaper]{article}
\pdfoutput=1
\usepackage{amsmath,amsthm,amssymb,enumerate}
\usepackage{dsfont}
\usepackage{color}
\usepackage{graphicx}
\usepackage{epsfig}
\usepackage{fancyhdr}
\usepackage{subfigure}
\usepackage{hyperref}
\usepackage{graphicx}        % standard LaTeX graphics tool
\usepackage{multicol}        % used for the two-column index
\usepackage[bottom]{footmisc}% places footnotes at page bottom

\begin{document}

\def\sgn{\mathrm{sign}}
\def\E{\mathds{E}}
\def\P{\mathds{P}}
\def\R{\mathds{R}}
\def\C{\mathds{C}}
\def\N{\mathds{N}}
\def\s{\mathfrak{s}}
\def\m{\mathfrak{m}}
\def\I{\mathcal{I}}
\def\M{\mathcal{M}}
\def\U{\mathcal{U}}
\def\T{\mathcal{T}}
\def\r{\mathrm{r}}
\def\d{\mathrm{d}}
\def\G{\mathcal{G}}
\def\F{\mathsf{F}}
\def\X{\mathsf{X}}
\def\ds{\displaystyle}
\def\indfunc{\mathds{1}}

\numberwithin{equation}{section}

 \pagestyle{fancy}
 \fancyhead{}
 \fancyhead[C]{G.~Campolieti and R.~Makarov, Dual Stochastic Transformations of Solvable Diffusions}
 \fancyhead[R]{\thepage}
 \fancyfoot{}

\newtheorem{theorem}{Theorem} %[section]
\newtheorem{corollary}{Corollary}
\newtheorem{proposition}{Proposition}
\newtheorem{lemma}{Lemma}

%%%%%%%%%%%%%%%%%%%%%%%%%%%%%%%%%%%%%%%%%%%%%%%%%%%%%%%%%%%%%%%%%%%%%

%\articletype{paper}

\title{Dual Stochastic Transformations of Solvable Diffusions}

\author{Giuseppe~Campolieti and Roman N. Makarov\\
Mathematics Department, Wilfrid Laurier University\\
75 University Avenue West, Waterloo, Ontario, Canada\\
E-mails: \texttt{gcampoli@wlu.ca} and \texttt{rmakarov@wlu.ca} }

\date{January 24, 2014}

\maketitle

\begin{abstract}
 We present new extensions to a method for constructing several families of solvable one-dimen\-sional time-homogeneous diffusions whose transition densities are obtainable in analytically closed-form. Our approach is based on a dual application of the so-called diffusion canonical transformation method
 that combines smooth monotonic mappings and measure changes via Doob-h transforms.
 This gives rise to new multi-parameter solvable diffusions that are generally divided into two main classes; the first is specified by having affine (linear)
 drift with various resulting nonlinear diffusion coefficient functions, while the second class
 allows for several specifications of a (generally nonlinear) diffusion coefficient with resulting nonlinear drift function. The theory is applicable to diffusions with either singular and/or non-singular endpoints.
 As part of the results in this paper, we also present a complete boundary classification and martingale characterization of the newly developed diffusion families.\\[2mm]
%The first class of models, having linear drift and nonlinear (state-dependent) volatility functions, is useful for equity derivative pricing in finance, while the second class of diffusions contains new models that are mean-reverting and which are applicable in areas such as interest rate modeling. As specific examples of the first class of affine drift models, we present explicit results for three new families of models that arise from the squared Bessel process (the Bessel family), the CIR process (the confluent hypergeometric family), and the Ornstein-Uhlenbeck diffusion (the OU family).  For the second class of nonlinear drift models, we give examples of solvable subfamilies called the Bessel family of mean-reverting diffusions and derive some closed-form integral formulas for conditional expectations of functionals that can be used to price bonds and bond options.
\end{abstract}

%%%%%%%%%%%%%%%%%%%%%%%%%%%%%%%%%%%%%%%%%%%%%%%%%%%%%%%%%%%%%%%%%%%

\section*{Introduction and Main Results}
A solvable continuous-time stochastic process can be basically defined as a process for which transition probability density functions are
obtainable in analytically closed-form. Such solvability permits us to precisely simulate paths of the process from its exact
sample distribution and also to readily compute certain mathematical expectations. For solvable families of diffusion processes, solvability
implies the existence of analytically closed-form spectral expansions for transition densities of the regular processes subject to
appropriate boundary conditions. For certain classes of diffusion models, the spectral expansions can be readily derived in closed-form.
For these same diffusion processes, the standard spectral methods show that analytical tractability also extends beyond
transition densities. In particular, closed-form expressions exist for other
fundamental quantities such as, for example, first-hitting time densities (or distributions) as well as joint probability densities for
various extrema of the process, etc.

The set of diffusion processes that are, on the one hand, tractable
and applicable for mathematical modeling and, on the other hand,
exactly solvable in closed-form is not so vast. This known set of classical diffusions includes mostly linear
diffusion processes or those whose drift and/or diffusion coefficients have a power or quadratic polynomial nonlinearity (see
\cite{BS02} and \cite{Linetsky2004a} for a comprehensive review of
such diffusions; see also \cite{Carr06,HL07}). An important goal is hence to extend solvability to other families that have useful applications.
There are two main tools that allow us to construct new solvable diffusion processes.
The first is related to a measure change on a chosen underlying diffusion and the second
involves a change of variable or smooth monotonic mapping (the It\^{o} formula).
In recent years, a new approach that combines
special measure changes, i.e. time-homogeneous
Doob-h transforms, together with special types of nonlinear smooth monotonic mapping transformations was
introduced for uncovering new families of exactly solvable
\emph{driftless} diffusion models
~\cite{ACCL,AC2005,AK2009,CM06,Kuznetsov}. These models exhibit nonlinear diffusion coefficients with
multiple adjustable parameters and have seen some useful applications in financial derivative pricing \cite{CM06,CM08,CMV11}.
The method has been coined as ``diffusion
canonical transformation'', wherein the solvability of a diffusion
process, say $(F_t)_{t\ge 0}$, is essentially reduced to that of a simpler underlying diffusion $(X_t)_{t\ge 0}$.

This paper provides the first formal extension of the diffusion canonical transformation method to include a substantially
larger {\it dual class} of monotonic mappings and thereby constructs two new main classes of solvable diffusions $(F_t)_{t\ge 0}$. Throughout,
these processes are also called $F$-diffusions. In contrast to the previous related papers,
the time-homogeneous Doob-h transform is now constructed more generally to include diffusions
with any type of singular and/or non-singular endpoints.
Hence, we also provide a complete boundary classification for all possible families of $F$-diffusions.
The first main class of $F$-diffusions consists of families satisfying a
time-homogeneous stochastic differential equation (SDE) of the form $dF_t=(a+bF_t)dt + \sigma(F_t)dW_t$
with an affine (linear) drift and multi-parameter nonlinear diffusion coefficient function.
We therefore note that the affine drift models presented in this paper significantly extend and include those studied in \cite{CM06}.
As part of our new results,
we present three explicitly solvable families of such $F$-diffusions, named {\it Bessel, confluent hypergeometric, and Ornstein-Uhlenbeck} families.
These processes arise via the diffusion canonical transformation method by respectively
choosing a squared Bessel (SQB) process, CIR (squared radial Ornstein-Uhlenbeck) process, and Ornstein-Uhlenbeck process as underlying diffusions.
The three new families include (recover) all the corresponding {\it driftless} $F$-diffusions obtained previously (e.g. see \cite{CM06})
as special subfamilies. Moreover, the new affine models inherit some of the important salient properties of their driftless counterparts.
One immediate application of such diffusions is asset pricing in finance (when $a=0$ and $b$ is a constant such as a risk-free interest rate).
In \cite{CM06}, we showed that these three families generate local volatility profiles with varied pronounced smiles
and skews (see Figures~\ref{fig1} and ~\ref{fig2}). Three particular subfamilies named here as the Bessel-$\mathsf{K}$,
confluent-$\mathcal{U}$, and OU models are of particular importance since, for each of them, there exists a risk-neutral probability measure
such that the discounted asset price process $e^{-bt}F_t$ is a martingale. As in the driftless case, these models are very amenable for pricing
many standard financial derivatives since the transition densities (state price densities)
for the asset or stock price (i.e. $S_t = F_t$) are given in closed form.
Clearly, the pricing of standard European options is reduced to the evaluation of a definite integral (e.g. see \cite{ACCL,CM06}).
As well, these solvable models admit explicit closed-form spectral expansions for the transition densities with imposed killing at arbitrary levels, for
the first hitting time densities, and for joint densities of the extrema and the price process. Hence, efficient pricing formulas of standard exotic options,
such as barrier and lookback options, are also available \cite{Campolieti2008}.
Discretely-monitored path-dependent options can be evaluated by using a path integral approach, as was done with previously related state-dependent
volatility models \cite{CM07,CM08}. Moreover, subfamilies of diffusions belonging to the Bessel and confluent hypergeometric
families admit absorption at zero asset price, so they can naturally be used in derivatives pricing under credit (default) risk.

The second main class of solvable models presented here consists of diffusions with a nonlinear
drift and with specification of a generally nonlinear diffusion coefficient.
In particular, within this second class of diffusions we find some explicitly solvable diffusion families with a {\it nonlinear mean-reverting drift}.
Mean-reverting models have useful applications in modeling interest rates.
Traditional single-factor interest rate models only consider linear mean reversion,
since such solvable models have analytically tractable solutions.
As an example of an alternative one-factor nonlinear mean-reverting solvable model, a new family of $F$-diffusions generated from the SQB
Bessel process is introduced in this paper. For a particular subfamily of such processes, we use
the closed-form transition probability densities for the Doob-h transformed processes and
the fact that an underlying bridge process and its Doob-h transformed bridge process have equivalent probability laws, and hence derive some
closed-form integral formulas for conditional expectations of functionals involving the discount factor of the process and the process terminal value.
The formulas are applicable to standard bond and bond option pricing.

To summarize this introduction, we point out how the rest of the paper is organized. In Section~\ref{sect1}, all of the necessary ingredients
for constructing the newly solvable dual classes of $F$-diffusions are presented.
%This begins with some standard classical results and assumptions for an underlying solvable $X$-diffusion, followed by the Doob-h transform, the construction of the Green functions and the basic form for the transition densities for all families of transformed $X^{(\rho)}$-diffusions.
A useful Lemma~\ref{Lemma_Xrho} for
the boundary classification of the families of transformed $X^{(\rho)}$-diffusions and hence
$F$-diffusions is also given. Subsection~\ref{subsect1.4} concludes this section with the basis of the dual
smooth monotonic $\F$ mapping transformations for
generating the two main classes of solvable $F$-diffusions $\{F_t=\F(X_t^{(\rho)}),\,t\geq 0\}$.
Section~\ref{sect2} presents three explicit $X^{(\rho)}$-diffusions,
i.e., Bessel, confluent hypergeometric, and Ornstein-Uhlenbeck families. For each,
we give analytical expressions for various transition densities and also derive the boundary classification.
In subfamilies where an endpoint is attainable (e.g. subfamilies (i) and (iii) of the
Bessel and confluent $X^{(\rho)}$-diffusions), we also derive in analytically closed-form the density
for the first hitting time to the endpoint. The boundary classification and first hitting time densities
then follow automatically for the $F$-diffusions.
Section~\ref{sect3} presents the construction of the mappings for generating the affine drift $F$-diffusions.
In Subsection~\ref{subsect3.3}, we analyze whether or not $F$-diffusions with linear drift ``preserve'' the drift rate, i.e.
whether $\frac{d}{dT}\mathsf{E}[F_T \mid F_t]=\mathsf{E}[a + bF_T\mid F_t]$, $t\leq T,$ holds and thereby
present a theorem that gives easy-to-implement limit conditions for verifying this property.
This can be viewed as a generalization of the martingale property for \emph{driftless} processes.
%where $\mathsf{E}[F_T\mid F_t]=F_t$ holds for $0\leq t\leq T,$ and, therefore, $\frac{d}{dT}\mathsf{E}[F_T\mid F_t]=0$ holds.
Thus, for the special case with $a=0$ we are able to prove whether a discounted process $(e^{-bt}F_t)$ is a martingale.
Section~\ref{subsect4.1} presents the three main families of affine drift $F$-diffusions with their explicit multi-parameter nonlinear
volatility specifications.
%The $F$-diffusions with $a=0$ are characterized in terms of whether or not $(e^{-bt}F_t)$ is a martingale.
We single out three subfamilies with this martingale property and two subfamilies in which $(e^{-bt}F_t)$ is a strict
supermartingale. In Section~\ref{subsect4.2} we discuss all possible monotonic maps that lead to nonlinear $F$-diffusions
with affine drift $a + bF$, $b\ne 0$. Section~\ref{sect5} presents Bessel families of nonlinear mean-reverting diffusions that are
obtained from subfamilies of the squared Bessel $X^{(\rho)}$-process by applying a power- or exponential-type mapping function $\F$.
Lemmas~\ref{lemma_SQBpwr} and \ref{lemma_SQBexp} give necessary conditions for the mean-reversion. The asymptotic behaviour of the drift
and diffusion coefficients is analyzed. Moreover, we present a model that admits a closed-form expression for the expectation of a discount
factor which can be used for bond and bond option pricing. In Appendix~\ref{sect_a1}, we derive
new asymptotic properties of Wronskians of fundamental solutions used to construct solvable
diffusions from the three main families considered here.
Such properties allow us to easily analyze stochastic properties of the solvable diffusions.

\section{Construction of Nonlinear Solvable Diffusions}\label{sect1}

\subsection{Underlying Diffusion}\label{subsect1.1}

Let $(X_t)_{t\geq 0}$ be a one-dimensional time-homogeneous regular diffusion on  $\I\equiv (l,r)$, $-\infty\leq l<r\leq\infty$,
defined by its infinitesimal generator:
\begin{equation}
(\G\, f)(x)\triangleq {1\over
2}\nu^2(x)f^{\prime\prime}(x) + \lambda(x)f^\prime(x)\,,\quad
   x \in \I\,.
\label{Generator}
\end{equation}
The functions $\lambda(x)$ and $\nu(x)$ denote, respectively, the
(infinitesimal) drift and diffusion coefficients of the process.
Throughout we assume that the functions $\lambda(x)$, $\lambda'(x)$,
$\nu(x)$ and $\nu''(x)$ are continuous on the open interval $\I$ and that
$\nu(x)$ is strictly positive on $\I$. The diffusion $(X_t)_{t\ge 0}$ has \textit{speed
measure}~$M(dx)$ and \textit{scale function}~$S(x)$ (see, e.g.,
\cite{BS02}) that are absolutely continuous with respect to the Lebesgue measure
and have smooth derivatives. The scale and speed density functions are defined as follows:
\begin{equation}
\s(x)=\frac{dS(x)}{dx}=\exp\left(-\int^x\frac{2\lambda(z)}{\nu^2(z)}dz\right)
   \mbox{ \ and \ } \m(x)=\frac{M(dx)}{dx}=\frac{2}{\nu^2(x)\s(x)}. \label{smx}
   \end{equation}

Given an $X$-diffusion, we can choose any pair of \textit{fundamental solutions} to
the differential equation $(\G\, \varphi)(x) = s \varphi(x)$, $s\in \C,$ $x\in\I$,
that are denoted by $\varphi^+_s$ and $\varphi^-_s$. For positive real values $s=\rho>0$,
$\varphi^+_\rho(x)$ and $\varphi^-_\rho(x)$ are linearly independent and respectively
increasing and decreasing positive functions of $x\in\I$.
The \textit{Wronskian} of these functions is given as:
 \begin{equation} \label{wronskian} W[\varphi^-_s,\varphi^+_s](x)
   \triangleq \varphi^-_s(x)\frac{d \varphi^+_s(x)}{d x} -
      \varphi^+_s(x)\frac{d \varphi^-_s(x)}{d x}=w_s\s(x)\,,
 \end{equation}
where $w_s$ is a constant w.r.t. $x$ and $w_\rho > 0$ for real $\rho > 0$.

We denote by $p_X(t;x_0,x)$ a transition probability density function
(PDF) for $(X_t)_{t\ge 0}$ w.r.t. the Lebesgue measure, i.e. it is a
fundamental solution to the Kolmogorov PDE where
$\P\big(X_{t}\in D|X_0=x_0\big)=\int_D p_X(t;x_0,x)\,dx$, $\forall
x_0 \in \I$, $t\geq 0$, $D\subseteq\I$.
We recall that the Green function $G_X(x,x_0,s)$ and the transition PDF are related via the
Laplace inverse transform w.r.t. $s$, i.e. $p_X(t;x_0,x) = {\mathcal L}_s^{-1}[G_X(x,x_0,s)][t]$.
The Green function $G_X$, for $x,x_0\in \I$, is written in terms of
a pair of functions $\psi_s,\phi_s$ and $\m(x)$ in the standard form
\cite{BS02}:
\begin{equation}
G_X(x,x_0,s) = {\mathcal W}_s^{-1}\m(x)\psi_s(x_<)\phi_s(x_>),
\label{greenfunc}
\end{equation}
where $x_<\equiv\min\{x,x_0\}$ and $x_>\equiv\max\{x,x_0\}$.
The functions $\{\psi_s,\phi_s\}$, that also solve
$(\G\, \varphi)(x) = s \varphi(x)$,
are generally not necessarily the same as the above chosen (elementary) pair. In particular,
these functions are linear combinations of $\{\varphi_s^+, \varphi_s^-\}$, i.e.
$\psi_s = A_1 \varphi^+_s + B_1 \varphi^-_s$,
$\phi_s = A_2 \varphi^+_s + B_2 \varphi^-_s$ with coefficients
$A_i = A_i(s), B_i = B_i(s)$, $i=1,2$. The Wronskian factor is given by
$W[\phi_s,\psi_s](x)/\s(x) = {\mathcal W}_s = (A_1B_2 - A_2B_1)w_s$.
The coefficients $A_i,B_i$ (where $A_1B_2 - A_2B_1 \ne 0$) and hence the functions,
are uniquely characterized
(within a multiplicative constant) by requiring that, for real $s=\alpha>0$,
$\psi_\alpha$ and $\phi_\alpha$ are respectively increasing and decreasing functions
and by additionally posing boundary conditions at regular
(non-singular) boundaries of $X$ (see~\cite{BS02}).
For a regular left boundary $l$, $\psi_\alpha(l+) = 0$ if $l\notin \I$
is specified as killing or $\frac{1}{\s(l+)}\frac{d\psi_\alpha(l+)}{dx}=0$
if $l$ is specified as reflecting and included in the state space.
If $l$ is a singular boundary, the functions have the following boundary properties:
If $l$ is entrance($\equiv$entrance-not-exit), then
$\psi_\alpha(l+) > 0, \frac{1}{\s(l+)}\frac{d\psi_\alpha(l+)}{dx}=0$;
if $l$ is exit($\equiv$exit-not-entrance), then $\psi_\alpha(l+) = 0,
\frac{1}{\s(l+)}\frac{d\psi_\alpha(l+)}{dx}>0$;
if $l$ is a natural boundary, then $\psi_\alpha(l+) =0,
\frac{1}{\s(l+)}\frac{d\psi_\alpha(l+)}{dx}=0$.
Analogous conditions hold for the right boundary $r$ involving the right limits,
i.e. $\phi_\alpha(r-) > (=) 0$ and $\frac{1}{\s(r-)}\frac{d\phi_\alpha(r-)}{dx} < (=) 0$.
Moreover, we note that if $l$ is singular then we can set $\psi_s(x) = \varphi_s^+(x)$
and similarly if $r$ is singular then $\phi_s(x) = \varphi_s^-(x)$.

\subsection{Change of Measure}\label{subsect1.2}

Consider a class of one-dimensional time-homogeneous regular
diffusions $(X^{(\rho)}_t)_{t\ge 0}\in \I$ with infinitesimal generator
 \begin{equation} \label{Generator_rho}
 (\G^{(\rho)}\,f)(x)\triangleq\frac{1}{2}\nu^2(x)f^{\prime\prime}(x)
 +\left(\lambda(x)+\nu^2(x)\frac{\hat{u}'_\rho(x)}{\hat{u}_\rho(x)}\right)f^{\prime}(x)\,.
 \end{equation}
A {\it strictly positive} generating function
$\hat{u}_\rho(x),$ $\rho>0$, is a linear combination of the chosen fundamental
pair $\varphi^\pm_\rho$:
\begin{equation}  \label{uhat} \hat{u}_\rho(x)=q_1 \varphi^+_\rho(x)
+ q_2 \varphi^-_\rho(x),
\end{equation}
with parameters $q_1,q_2\ge 0$ and at least one of them being
strictly positive. The speed and scale densities for an $X^{(\rho)}$-diffusion are given in terms
of those for the underlying $X$-diffusion:
 \begin{equation}
   \m_\rho(x)=\hat{u}^2_\rho(x)\,\m(x)\;\mbox{ and }\;
   \s_\rho(x)=\frac{\s(x)}{\hat{u}^2_\rho(x)}\,.\label{smxrho}
 \end{equation}

By comparing the generators (\ref{Generator}) and (\ref{Generator_rho}), observe that $X^{(\rho)}$-diffusions can also be viewed as arising from the underlying
$X$-diffusion by the application of a measure change. In fact, the $X^{(\rho)}$-diffusion can
be realized from the $X$-diffusion upon employing a time-homogeneous space-time transform, i.e.
a Doob-$h$ transform, where $h=\hat{u}_\rho$, which is $\rho$-excessive (see \cite{BS02}).
Both processes are regular on the same state space $\I=(l,r)$.

Given a generating function $\hat{u}_\rho$, we define the pair
$\varphi^{(\rho)+}_s \triangleq {\varphi^+_{\rho+s}\over \hat u_\rho}$ and
$\varphi^{(\rho)-}_s \triangleq {\varphi^-_{\rho+s}\over \hat u_\rho}$.
By applying the differential operator $\G^{(\rho)}$, it follows that
these functions solve
$(\G^{(\rho)}\,\varphi^{(\rho)}_s )(x) = s \varphi^{(\rho)}_s(x)$, $s\in \C,$ $x\in\I$.
From (\ref{wronskian}) and (\ref{smxrho}),
the Wronskian of these solutions is given by
\[W[\varphi^{(\rho)-}_s,\varphi^{(\rho)+}_s](x)
= {1\over \hat{u}_\rho^2}W[\varphi^-_{\rho+s},\varphi^+_{\rho+s}](x) = w_{\rho+s}\,\s_\rho(x)\,.
\]
Hence, $\{\varphi^{(\rho)-}_s(x),\varphi^{(\rho)+}_s(x)\}$
are a fundamental set of solutions that are linearly independent and
strictly positive functions of $x\in\I$ for real values $s = \alpha > 0$.

The Green function for $X^{(\rho)}$-diffusions on $\I$ then has the general form
\begin{equation}
G^{(\rho)}_{X}(x,x_0,s) = ({\mathcal W}^{(\rho)}_s)^{-1} \m_\rho(x)
\psi^{(\rho)}_s(x_<)\phi^{(\rho)}_s(x_>),
\label{greenfunc_rho}
\end{equation}
where, in analogy with the $X$-diffusion, $\{\psi^{(\rho)}_s,\phi^{(\rho)}_s\}$ solve
$(\G^{(\rho)}\,\varphi^{(\rho)}_s )(x) = s \varphi^{(\rho)}_s(x)$ and are linear combinations of
$\{\varphi^{(\rho)-}_s,\varphi^{(\rho)+}_s\}$, i.e.
$\psi^{(\rho)}_s = \hat{A}_1 \varphi^{(\rho)+}_s + \hat{B}_1 \varphi^{(\rho)-}_s
= {\hat{A}_1 \varphi^+_{\rho + s} + \hat{B}_1 \varphi^-_{\rho + s} \over \hat{u}_\rho}$
and $\phi^{(\rho)}_s = \hat{A}_2 \varphi^{(\rho)+}_s + \hat{B}_2 \varphi^{(\rho)-}_s
= {\hat{A}_2 \varphi^+_{\rho + s} + \hat{B}_2 \varphi^-_{\rho + s} \over \hat{u}_\rho}$
with coefficients
$\hat{A}_i = \hat{A}_i(\rho,s+\rho), \hat{B}_i = \hat{B}_i(\rho,s+\rho)$, $i=1,2$.
The Wronskian factor is then given by
${\mathcal W}^{(\rho)}_s =
W[\phi^{(\rho)}_s,\psi^{(\rho)}_s](x)/\s_\rho(x) =
(\hat{A}_1\hat{B}_2 - \hat{A}_2\hat{B}_1)w_{\rho + s}$,
where $\hat{A}_1\hat{B}_2 - \hat{A}_2\hat{B}_1 \ne 0$.
These coefficients are uniquely characterized
(within a multiplicative constant) by requiring that, for real $s=\alpha>0$,
$\psi^{(\rho)}_\alpha$ and $\phi^{(\rho)}_\alpha$
are respectively increasing and decreasing functions
and by additionally posing boundary conditions at regular
boundaries of $X^{(\rho)}$.
For a regular left boundary $l$, $\psi^{(\rho)}_\alpha(l+) = 0$ if $l\notin \I$
is specified as killing or
$\frac{1}{\s_\rho(l+)}\frac{d\psi^{(\rho)}_\alpha(l+)}{dx}=0$
if $l$ is specified as reflecting and included in the state space. Note that this reflecting boundary condition
is equivalently written as $\hat{A}_1W[\hat{u}_\rho,\varphi^+_{\rho + s}](l+)/\s(l+)
+ \hat{B}_1W[\hat{u}_\rho,\varphi^-_{\rho + s}](l+)/\s(l+) = 0$.
If $l$ is a singular boundary, the functions have the following boundary properties: if $l$ is entrance, then
$\psi^{(\rho)}_\alpha(l+) > 0, \frac{1}{\s_\rho(l+)}\frac{d\psi^{(\rho)}_\alpha(l+)}{dx}=0$;
if $l$ is exit, then $\psi^{(\rho)}_\alpha(l+) = 0,
\frac{1}{\s_\rho(l+)}\frac{d\psi^{(\rho)}_\alpha(l+)}{dx}>0$;
if $l$ is a natural boundary, then $\psi^{(\rho)}_\alpha(l+) =0,
\frac{1}{\s_\rho(l+)}\frac{d\psi^{(\rho)}_\alpha(l+)}{dx}=0$.
Analogous conditions hold for the right boundary $r$ involving the right limits,
i.e. $\phi^{(\rho)}_\alpha(r-) > 0$ (or $=0$) and
$\frac{1}{\s_\rho(r-)}\frac{d\phi^{(\rho)}_\alpha(r-)}{dx} < 0$ (or $=0$).

It clearly follows from (\ref{greenfunc}) and (\ref{greenfunc_rho}) that
any Green function for a diffusion $X^{(\rho)}$ can be related to some Green function
for a diffusion $X$ by
\begin{equation}
{G^{(\rho)}_{X}(x,x_0,s) \over \m_\rho(x)} = {1 \over \hat{u}_\rho(x)\hat{u}_\rho(x_0)}
{G_X(x,x_0,s+\rho) \over \m(x)}.
\label{greenfuncs_relation}
\end{equation}
For diffusion $X^{(\rho)}$, a transition PDF is obtained from its corresponding
Green function by Laplace inversion, i.e.
\begin{equation}\label{Laplace_inversion}
p_X^{(\rho)}(t;x_0,x) = {\mathcal L}_s^{-1}[G^{(\rho)}_{X}(x,x_0,s)][t]
= \m_\rho(x){\mathcal L}_s^{-1}[({\mathcal W}^{(\rho)}_s)^{-1} \psi^{(\rho)}_s(x_<)\phi^{(\rho)}_s(x_>)][t].
\end{equation}
By Laplace inverting (\ref{greenfuncs_relation}) we see that a
transition density $p_X^{(\rho)}(t;x_0,x)$ for a diffusion $X^{(\rho)}$
is related to a transition density for a diffusion $X$ by
\begin{equation}
p_X^{(\rho)}(t;x_0,x)=\frac{\hat{u}_\rho(x)}{\hat{u}_\rho(x_0)}
e^{-\rho t}p_X(t;x_0,x),\;x,x_0\in\I\,,\;t>0\,. \label{prho}
\end{equation}

\subsection{Boundary Classification} \label{subsect1.3}

Given an underlying $X$-diffusion and $\rho > 0$, any regular diffusion $(X_t^{(\rho)})_{t\ge 0}\in (l,r)$
with generator $\G^{(\rho)}$ in (\ref{Generator_rho}) falls into one of three general families:
  \begin{enumerate}[(i)]
    \item $\{q_1 = 0,\,q_2 > 0\}$ where $\hat{u}_\rho(x) = q_2\varphi_\rho^-(x)$,
    \item $\{q_1 > 0,\,q_2 = 0\}$ where $\hat{u}_\rho(x) = q_1\varphi_\rho^+(x)$,
    \item $\{q_1>0,\,q_2 > 0\}$  where $\hat{u}_\rho(x) = q_1\varphi_\rho^+(x) + q_2\varphi_\rho^-(x)$.
  \end{enumerate}

For $s_1, s_2\in \C$ with positive real parts, we denote
\begin{equation} \nonumber
n(x;s_1,s_2) = {\varphi_{s_1}^+(x)\over \varphi_{s_2}^-(x)},\,\,\,\,
n(l+;s_1,s_2) = \lim\limits_{x\to l+}n(x;s_1,s_2), \,\,\, n(r-;s_1,s_2) = \lim\limits_{x\to r-} n(x;s_1,s_2).
\end{equation}
We recall (see \cite{BS02}) that for singular (non-regular) boundaries of $(X_t)_{t\ge 0}$,
i.e. entrance ($\equiv$entrance-not-exit), exit ($\equiv$exit-not-entrance)
or natural, we have $n(l+;s_1,s_2)=0$ and $n(r-;s_1,s_2)=\infty$. For regular boundaries of
$(X_t)_{t\ge 0}$ it is also possible, depending on the choice of fundamental solutions and the
type of boundary conditions imposed
at $x=e\in\{l+,r-\}$, that $n(e;s_1,s_2)$ is finite for all $s_1, s_2\in \C$ with positive real parts.
In particular, we generally have $0 \le n(l+;\rho,\rho) <\infty$ and $0 < n(r-;\rho,\rho) \le \infty$.

The fundamental solutions generally satisfy the square integrability conditions w.r.t. the speed measure:
$(\varphi_\rho^+,\varphi_\rho^+)_{(l,x]}<\infty $ and $(\varphi_\rho^-,\varphi_\rho^-)_{[x,r)}<\infty$
for $\rho > 0$ and $x\in\I$. Throughout this paper we conveniently define the inner product of
two functions $f$, $g$ w.r.t. $\m$ on a closed interval $[a,b]$ as
$(f,g)_{[a,b]} \triangleq \int_a^b f(x)g(x)\m(x)dx$
and $(f,g)_{(a,b]}\triangleq\lim\limits_{\varepsilon\to
a^+}(f,g)_{[\varepsilon,b]}$, $(f,g)_{[a,b)} \triangleq
\lim\limits_{\varepsilon \to b^-}(f,g)_{[a, \varepsilon]}$,
$||f||_{(a,b)}^2 \triangleq (f,f)_{(a,b)}$.

\begin{lemma} \label{Lemma_Xrho} The above three families (i)--(iii) of regular diffusions $(X^{(\rho)}_t)_{t\ge 0}$ on $(l,r)$ with generator
$\G^{(\rho)}$, $\rho > 0$, defined by (\ref{Generator_rho}) and
(\ref{uhat}) have the following boundary classification:
\begin{enumerate}[(i)]
\item $q_1 = 0,\,q_2 > 0$: $l$ is attracting natural if
$(\varphi_\rho^+,\varphi_\rho^-)_{(l,x]} = \infty$, is exit (or attracting natural)
when $n(l+;\rho,\rho)=0$ (or $\neq 0$) if $(\varphi_\rho^+,\varphi_\rho^-)_{(l,x]} < \infty$ and
$(\varphi_\rho^-,\varphi_\rho^-)_{(l,x]} = \infty$, and is otherwise regular
if $(\varphi_\rho^-,\varphi_\rho^-)_{(l,x]} < \infty$.

The boundary $r$ is non-attracting (or attracting) natural when
$n(r-;\rho,\rho)=\infty$ (or $<\infty$) if $(\varphi_\rho^+,\varphi_\rho^-)_{[x,r)} = \infty$, and
is entrance (or regular) when $n(r-;\rho,\rho)=\infty$ (or $<\infty$)
if $(\varphi_\rho^+,\varphi_\rho^-)_{[x,r)} < \infty$.

\item $q_1 > 0,\,q_2 = 0$: $r$ is attracting natural if
$(\varphi_\rho^+,\varphi_\rho^-)_{[x,r)} = \infty$, is exit (or attracting natural)
when $n(r-;\rho,\rho)=\infty$ (or $<\infty$) if
$(\varphi_\rho^+,\varphi_\rho^-)_{[x,r)} < \infty$ and $(\varphi_\rho^+,\varphi_\rho^+)_{[x,r)} = \infty$, and
is otherwise regular if $(\varphi_\rho^+,\varphi_\rho^+)_{[x,r)} < \infty$.

The boundary $l$ is non-attracting (or attracting) natural when $n(l+;\rho,\rho)=0$ (or $\neq 0$)
if $(\varphi_\rho^+,\varphi_\rho^-)_{(l,x]} = \infty$, and is entrance
(or regular) when $n(l+;\rho,\rho)=0$ (or $\neq 0$) if $(\varphi_\rho^+,\varphi_\rho^-)_{(l,x]} < \infty$.
\item $q_1 > 0,\,q_2 > 0$: The boundary $l$ has the same
classification as in (i) and $r$ has the same classification as in
(ii).
\end{enumerate}
\end{lemma}
\begin{proof}
Let $x<y$, $x,y\in (l,r)$, and denote the scale measure $\mathcal{S}_\rho[x,y]= \!\int_{x}^{y} \s_\rho(z)dz$,
$\mathcal{S}_\rho(l,y]= \lim\limits_{x\to l+}\mathcal{S}_\rho[x,y]$, $\mathcal{S}_\rho[x,r)= \lim\limits_{y\to r-}\mathcal{S}_\rho[x,y]$, and let
\begin{align*}
  \Sigma_\rho(l) &= \int_l^x \mathcal{S}_\rho(l,z]\m_\rho(z)dz,\; &\Sigma_\rho(r) &= \int_x^r \mathcal{S}_\rho[z,r)\m_\rho(z)dz,\\ N_\rho(l) &= \int_l^x \mathcal{S}_\rho[z,x]\m_\rho(z)dz, \; &N_\rho(r) &= \int_x^r \mathcal{S}_\rho[x,z]\m_\rho(z)dz.\end{align*}
The proof now follows by applying the Feller conditions for the respective $X^{(\rho)}$-diffusions (i)--(iii)
with scale and speed densities in (\ref{smxrho}), That is, $e\in\{l+,r-\}$ is regular if $\Sigma_\rho(e) < \infty$ and $N_\rho(e) < \infty$,
exit if $\Sigma_\rho(e) < \infty$ and $N_\rho(e) = \infty$, entrance if $\Sigma_\rho(e) = \infty$ and $N_\rho(e) < \infty$, and
natural if $\Sigma_\rho(e) = \infty$ and $N_\rho(e) = \infty$; $l$ ($r$) is attracting if and only if
$\mathcal{S}_\rho(l,x]$ ($\mathcal{S}_\rho[x,r)$) is finite. From (\ref{wronskian}) and (\ref{smxrho}), we have
$\s_\rho = {1\over q_2 w_\rho}\big({\varphi_\rho^+\over \hat{u}_\rho}\big)'$, if $q_2 > 0$, and
$\s_\rho = -{1\over q_1 w_\rho}\big({\varphi_\rho^-\over \hat{u}_\rho}\big)'$, if $q_1 > 0$.
Hence, $\mathcal{S}_\rho[x,y] = {1\over q_2 w_\rho}\big[{\varphi_\rho^+(y)\over \hat{u}_\rho(y)} -
{\varphi_\rho^+(x)\over \hat{u}_\rho(x)}\big]$, if $q_2 > 0$;
$\mathcal{S}_\rho[x,y] = {1\over q_1 w_\rho}\big[{\varphi_\rho^-(y)\over \hat{u}_\rho(y)} -
{\varphi_\rho^-(x)\over \hat{u}_\rho(x)}\big]$, if $q_1 > 0$. Consider family (i). Then,
$\mathcal{S}_\rho(l,x] = {1\over q_2^2 w_\rho}\big[n(x;\rho,\rho) - n(l+;\rho,\rho)\big] < \infty$, so $l$ is attracting;
$\mathcal{S}_\rho[x,r) = {1\over q_2^2 w_\rho}\big[n(r-;\rho,\rho) - n(x;\rho,\rho)\big]$, so $r$ is attracting
if and only if $n(r-;\rho,\rho) < \infty$.
$\Sigma_\rho(l) = {1\over w_\rho}\{(\varphi_\rho^+,\varphi_\rho^-)_{(l,x]} - n(l+;\rho,\rho)(\varphi_\rho^-,\varphi_\rho^-)_{(l,x]}\}$
and $\Sigma_\rho(r) = {1\over w_\rho}\{n(r-;\rho,\rho)(\varphi_\rho^-,\varphi_\rho^-)_{[x,r)} - (\varphi_\rho^+,\varphi_\rho^-)_{[x,r)}\}$.
Hence, $\Sigma_\rho(l) < \infty$ if and only if $(\varphi_\rho^+,\varphi_\rho^-)_{(l,x]}< \infty$ when $n(l+;\rho,\rho) = 0$ and
$\Sigma_\rho(l) < \infty$ if and only if $(\varphi_\rho^-,\varphi_\rho^-)_{(l,x]}< \infty$ when $n(l+;\rho,\rho) \ne 0$,
since $(\varphi_\rho^-,\varphi_\rho^-)_{(l,x]} < \infty \implies (\varphi_\rho^+,\varphi_\rho^-)_{(l,x]}< \infty$, while
$\Sigma_\rho(r) < \infty$ if and only if $n(r-;\rho,\rho)< \infty$ and $(\varphi_\rho^-,\varphi_\rho^-)_{[x,r)}< \infty$.
$N_\rho(l) = {1\over w_\rho}\{n(x;\rho,\rho)(\varphi_\rho^-,\varphi_\rho^-)_{(l,x]} - (\varphi_\rho^+,\varphi_\rho^-)_{(l,x]}\}$ and
$N_\rho(r) = {1\over w_\rho}\{(\varphi_\rho^+,\varphi_\rho^-)_{[x,r)} - n(x;\rho,\rho)(\varphi_\rho^-,\varphi_\rho^-)_{[x,r)}\}$. Hence,
$N_\rho(l) < \infty$ if and only if $(\varphi_\rho^-,\varphi_\rho^-)_{(l,x]} < \infty$, since $n(x;\rho,\rho) < \infty$, while
$N_\rho(r) < \infty$ if and only if $(\varphi_\rho^+,\varphi_\rho^-)_{[x,r)} < \infty$, since $(\varphi_\rho^-,\varphi_\rho^-)_{[x,r)} < \infty$.
The above combined conditions are then summarized as stated in the Lemma for family (i).
The stated boundary classification for families (ii) and (iii) is proven by applying similar steps as in family (i).
\end{proof}

\subsection{Generating $F$-Diffusions: Dual Transformations}\label{subsect1.4}

We now consider $F$-diffusions $\{F_t \triangleq \F
(X^{(\rho)}_t), t\ge 0\}$ defined by strictly monotonic real-valued mapping
$\F:\I\to\I_{\sf F}$ with $\F ', \F ''$ continuous on $\I$ with unique inverse $\X \triangleq \F^{-1}$.
Such an elementary (It\^{o}) transformation gives a diffusion process $(F_t)_{t\ge 0}$
with infinitesimal generator
\begin{equation}\label{GeneratorF}
(\mathcal{G}_{\sf F} h)(F) = \frac{1}{2}\,\sigma^2(F)
h''(F) + \alpha(F) h'(F),\quad F\in\I_{\sf F}=(F^l,F^r)
\end{equation}
where $\alpha$ and $\sigma$ are the respective drift and diffusion coefficients:
\begin{equation} \label{F_coefficients}
\alpha(F)=(\mathcal{G}^{(\rho)} \F)(\X(F)),\;
    \sigma(F)=\nu(\X(F))/|\X'(F)|\,. \end{equation}
$(F_t)_{t\ge 0}$ is a regular diffusion on $\I_{\sf F}=(F^l,F^r)$ with endpoints $F^l=\min\{\F(l+),\F(r-)\}$ and
$F^r=\max\{\F(l+),\F(r-)\}$.

The map $\F$ can be specified so as to create a process
with a desired drift or diffusion coefficient.
To obtain a linear-drift $F$-diffusion, the drift coefficient is specified by a linear function, i.e. we set
$\alpha(F) = a+bF$. Hence, by the first relation in (\ref{F_coefficients}), with $x=\X(F)$, $F=\F(x)$,
we see that $\F$ is obtained by solving the 2nd order linear nonhomogeneous ODE:
 \begin{equation} \label{FODE}
  (\G^{(\rho)}\,\F )(x) = a+b\F (x)\,.
 \end{equation}
Given any strictly monotonic smooth solution $\F(x)$, then $(F_t)_{t\ge 0}$ is
a process with specified {\it affine (linear) drift} $\alpha(F) = a+bF$
and having generally nonlinear diffusion coefficient with
infinitesimal generator in (\ref{GeneratorF}),
where $\sigma^2(F) = \nu^2(\X(F))/[\X'(F)]^2$ follows automatically
from the second relation in (\ref{F_coefficients}).

An alternative approach is to specify the diffusion coefficient rather than the drift function.
One way is to directly specify a strictly positive function $\sigma(F)$ with continuous first derivative.
Then, the second equation in (\ref{F_coefficients}), i.e.
${dF\over \sigma(F)} = \pm {dx\over  \nu(x)}$, is integrated to give
$F=\F(x)$, and its inverse relation $x=\X(F)$, where $\sgn ({\sf F}'(x)) = \pm 1$
allows for either a strictly increasing or a decreasing map.
Another way is to explicitly specify a strictly nonzero continuously differentiable function
$\F'(x)$, i.e. specify $\widetilde{\sigma}(x)=\sigma(\F(x)) = \nu(x)|\F '(x)|$, and then integrate giving
$\F(x) = \bar{F}\pm \int_{\bar{x}}^x \frac{\widetilde{\sigma}(y)}{\nu(y)}\,dy$,
with $\bar{F}=\F(\bar{x})\in \I_F$ and $\bar{x}\in \I$ as arbitrary constants. The diffusion function is then given by
$\sigma(F) = \widetilde{\sigma}(\X(F))$. Either way, the resulting strictly monotonic smooth map $\F$
is used to produce $F$-diffusions defined by the infinitesimal generator in (\ref{GeneratorF})
with a specified diffusion coefficient function $\sigma(F)$ and a resulting generally {\it nonlinear drift} function:
\begin{equation}\label{drift_F2}
\alpha(F) = {\lambda(\X(F))\over \X'(F)} + \bigg({\nu(\X(F)) \over \X'(F)}\bigg)^2
\left[ \X'(F) \frac{\hat{u}'_\rho(\X(F))}{\hat{u}_\rho(\X(F))}
- \frac{1}{2}\frac{\X''(F)}{\X'(F)}
\right]\,.
\end{equation}
This expression follows from the first equation in (\ref{F_coefficients}) while using
(\ref{Generator_rho}) where ${\sf F}'(x) = [{\sf X}'(F)]^{-1}$, $\F''(x) = -\X''(F)/(\X'(F))^3$.

By either of the above {\it dual} transformation approaches,
several families of analytically solvable $F$-diffusion models can be constructed using
known solvable underlying $\X$-diffusion processes defined by (\ref{Generator}). The $F$-diffusion models given
by (\ref{GeneratorF}) either have a nonlinear state dependent volatility with a specified affine (linear) drift or
have nonlinear state dependent drift with a volatility that is specified
as either affine or as nonlinear state dependent. We refer to the above general framework
as the \textit{``diffusion canonical transformation''} methodology.

\begin{lemma} \label{Theorem_F_classification} The boundary classification for an $F$-diffusion defined
by $\{F_t=\F(X^{(\rho)}_t),t\geq 0\}$ with strictly monotonic mapping $\F$ is equivalent
to the corresponding $X^{(\rho)}$-diffusion.
\end{lemma}
\begin{proof}
This follows trivially by the diffeomorphism $X^{(\rho)}_t \to F_t=\F(X^{(\rho)}_t)$.
\end{proof}

\section{Three Choices of Underlying Solvable Diffusions} \label{sect2}

\subsection{The Squared Bessel Process} \label{subsect2.1}
Consider a~$\lambda_0$-dimensional squared Bessel (SQB) process
obeying the SDE $dX_t=\lambda_0 dt+\nu_0\sqrt{X_t}dW_t$
with constants $\nu_0>0$ and $\lambda_0\in\R$. This diffusion has regular state space
$\I=(0,\infty)\equiv \R_+$ with generator
$(\G\, f)(x)\triangleq {\nu_0^2\over 2}[xf''(x) + (\mu + 1)f'(x)],\,\,\,\,\mu \in \R$, and scale and speed
densities $\s(x)=x^{-\mu-1}$ and $\m(x)=\frac{2}{\nu_0^2}x^{\mu},$ where
$\mu\equiv\frac{2\lambda_0}{\nu_0^2}-1$.
The origin is entrance if $\mu\geq 0$, regular if $\mu\in(-1,0)$ and exit if $\mu\leq -1$;
$\infty$ is natural (attracting for $\mu>0$). As a pair of fundamental solutions to
$(\G\, \varphi)(x) = s \varphi(x)$, $s\in \C$, for $x\in\R_+$ we choose
\begin{equation} \label{SQBfund}
\varphi^+_s(x) =
x^{-\mu/2}I_{\vert\mu\vert}\big(2\sqrt{2 s x}/\nu_0\big) \mbox{ \ and \ }
\varphi^-_s(x) = x^{-\mu/2}K_{\mu}\big(2\sqrt{2 s x}/\nu_0\big)\,,
\end{equation}
where $I_\mu(z)$ and $K_\mu(z)$ are the modified Bessel functions
(of order $\mu$) of the first and second kind, respectively (see \cite{AS72}).
Note that by symmetry $K_{\mu}(z) = K_{-\mu}(z) = K_{\vert\mu\vert}(z)$.
The pair $\varphi^\pm_s(x)$ satisfies (\ref{wronskian}) where $w_s=1/2$.
For $s_1,s_2\in\C$, all Wronskians $W[\varphi^\pm_{s_1},\varphi^\pm_{s_2}](x)$ and $W[\varphi^-_{s_1},\varphi^+_{s_2}](x)$
are readily obtained using differential recurrences $zI^{\prime}_\mu(z) = \mu I_\mu(z) + zI_{\mu + 1}(z)$
and $zK^{\prime}_\mu(z) = \mu K_\mu(z) - zK_{\mu + 1}(z)$.

The well-known Green function in (\ref{greenfunc}) for the SQB on $\R_+$ is readily formed
by taking appropriate linear combinations of $\varphi^\pm_s$ giving
\begin{equation}\label{SQB_psi_phi}
\psi_s(x) = x^{-\mu/2}I_{\tilde\mu}\big(2\sqrt{2 s x}/\nu_0\big) \mbox{ \ and \ }
\phi_s(x) = x^{-\mu/2}K_{\mu}\big(2\sqrt{2 s x}/\nu_0\big)\,,
\end{equation}
where $\tilde\mu = \mu$ if $\mu \ge 0$ or if $\mu\in(-1,0)$ and 0 is reflecting,
and $\tilde\mu = -\mu = \vert\mu\vert$ if $\mu \le -1$ or if $\mu\in(-1,0)$ and 0 is killing.
In all cases, the Wronskian factor is simply ${\mathcal W}_s = w_s = 1/2$.
Laplace inverting the Green function, while using (\ref{SQB_psi_phi}) and the identity
${\mathcal L}_s^{-1}[I_\nu(x\sqrt{2s})K_\nu(y\sqrt{2s})][t] =
{1\over 2t}e^{-{x^2 + y^2\over 2t}}I_\nu\big({xy\over t}\big)$, for $0 < x\le y$, $t>0$,
readily gives the known transition PDF for $X_t\in\R_+$ as
\begin{equation}
p_X(t;x_0,x) = \left({x\over
x_0}\right)^{\frac{\mu}{2}}
    \,{e^{-2(x + x_0)/\nu_0^2t} \over \nu_0^2t/2}
     I_{\tilde\mu}\left({4\sqrt{xx_0}\over \nu_0^2t}\right)
 \label{PrkernelSQB}
\end{equation}
for all $x,x_0>0$, $t>0$ and $\tilde\mu$ given as above for the respective cases.

The Bessel family of $X^{(\rho)}$-diffusions has generator in (\ref{Generator_rho}) with
$\varphi^+_\rho(x)$ and $\varphi^-_\rho(x)$ in $\hat{u}_\rho(x)$ defined by the functions
in (\ref{SQBfund}) for positive real values of $s=\rho > 0$ and the
following lemma gives the boundary classification for these processes.

\begin{lemma} \label{lemma_class_Bessel}
The Bessel family of regular $X^{(\rho)}$-diffusions on $\R_+$ has the following boundary classification:
The origin is entrance if $q_2=0$, is regular if $q_2>0$ and $|\mu| < 1$, and
is exit if $q_2>0$ and $|\mu|\geq 1$; $\infty$ is non-attracting (or attracting) natural if $q_1=0$ (or $q_1>0$).
\end{lemma}
\begin{proof}
From the asymptotic relations for $\varphi_\rho^\pm$ in (i) of Appendix~\ref{subsect_a11}, we have
$n(0+;\rho,\rho)=0$, $n(\infty;\rho,\rho)=\infty$, $(\varphi_\rho^+,\varphi_\rho^-)_{(0,x]} < \infty$,
$(\varphi_\rho^-,\varphi_\rho^-)_{(0,x]} < \infty$ if and only if $\vert\mu\vert < 1$, and
$(\varphi_\rho^+,\varphi_\rho^-)_{[x,\infty)} = \infty$.
The stated boundary classification then follows by Lemma~\ref{Lemma_Xrho}.
\end{proof}
Hence, in all subfamilies (i)--(iii) the point at infinity is a natural boundary.
The Bessel subfamily of type (ii), where $q_1>0, q_2=0$, is conservative with the origin as an
entrance. For subfamilies (i) and (iii), where $q_2 > 0$, the origin is attainable.

From the theory in Section \ref{subsect1.3},
we can readily construct Green functions in the form of (\ref{greenfunc_rho}) and subsequently obtain the
corresponding transition PDF by Laplace inversion. We now give some examples for the Bessel family of
$X^{(\rho)}$-diffusions on $\R_+$. For subfamily (ii), $q_2=0$ and, without loss of generality, we take $q_1=1$, i.e.
$\hat{u}_\rho(x) = \varphi^+_{\rho}(x)$. The Green function in (\ref{greenfunc_rho}) is uniquely
specified by
\begin{equation*}
\psi^{(\rho)}_s(x) = {\varphi^+_{\rho + s}(x)\over \varphi^+_{\rho}(x)}
= {I_{\vert\mu\vert}\big({2\over\nu_0}\sqrt{2 (\rho + s) x}\big)
\over I_{\vert\mu\vert}\big({2\over\nu_0}\sqrt{2 \rho x}\big)}
\mbox{ \ and \ }
\phi^{(\rho)}_s(x) = {\varphi^-_{\rho + s}(x)\over \varphi^+_{\rho}(x)}
= {K_{\mu}\big({2\over\nu_0}\sqrt{(\rho + s) x}\big) \over I_{\vert\mu\vert}\big({2\over\nu_0}\sqrt{2 \rho x}\big)}\,,
\end{equation*}
where in this case ${\mathcal W}^{(\rho)}_s = w_{s+\rho} = {1\over 2}$ and
$\m_\rho(x) = \m(x)x^{-\mu} I_{\vert\mu\vert}^2\big(2\sqrt{2 \rho x}/\nu_0\big)$.
The transition PDF for the Bessel subfamily (ii) is given explicitly via (\ref{Laplace_inversion}),
upon using the above Laplace inverse identity for order $\nu = {\vert\mu\vert}$ and the property
${\mathcal L}_s^{-1}[f(\rho + s)][t] = e^{-\rho t}{\mathcal L}_s^{-1}[f(s)][t]$):
\begin{align}
p^{(\rho)}_X(t;x_0,x) &= {I_{\vert\mu\vert}\big({2\over\nu_0}\sqrt{2\rho x}\big)
\over I_{\vert\mu\vert}\big({2\over\nu_0}\sqrt{2\rho x_0}\big)}
{e^{-\rho t -2(x + x_0)/\nu_0^2t} \over \nu_0^2t/2}
     I_{\vert\mu\vert}\left({4\sqrt{xx_0}\over \nu_0^2t}\right)\,,
 \label{PDF_Bessel_1}
\end{align}
$x,x_0>0$, $t>0$, $\mu\in\R$. We note that this has the form in (\ref{prho}) where
$p_X(t;x_0,x)$ is given by the r.h.s. of (\ref{PrkernelSQB}) for $\tilde\mu = \vert\mu\vert$.

For subfamilies (i) and (iii), i.e. $q_1\ge 0, q_2 > 0$, the origin is exit for $|\mu|\geq 1$ and regular for
$|\mu| < 1$. In particular, by specifying the origin as killing for $|\mu| < 1$ then
the process is absorbed at the origin for all $\mu \in \R$.
In this case, the Green function in (\ref{greenfunc_rho}) is formed by taking
$\psi^{(\rho)}_s(x) = \varphi^+_{\rho + s}(x)/\hat{u}_{\rho}(x)$ and
$\phi^{(\rho)}_s(x) = \varphi^-_{\rho + s}(x)/\hat{u}_{\rho}(x)$:
\begin{equation*}
\psi^{(\rho)}_s(x) = {I_{\vert\mu\vert}\big({2\over\nu_0}\sqrt{2 (\rho + s) x}\big)
\over q_1 I_{\vert\mu\vert}\big({2\over\nu_0}\sqrt{2 \rho x}\big)
+ q_2 K_\mu\big({2\over\nu_0}\sqrt{2 \rho x}\big)}\,,\,\,\,
\phi^{(\rho)}_s(x) = {K_{\mu}\big({2\over\nu_0}\sqrt{(\rho + s) x}\big)
\over q_1 I_{\vert\mu\vert}\big({2\over\nu_0}\sqrt{2 \rho x}\big)
+ q_2 K_\mu\big({2\over\nu_0}\sqrt{2 \rho x}\big)},
\end{equation*}
with Wronskian ${\mathcal W}^{(\rho)}_s = {1\over 2}$ and
$\m_\rho(x) = \m(x)x^{-\mu} [q_1 I_{\vert\mu\vert}\big(2\sqrt{2 \rho x}/\nu_0\big)
+ q_2 K_\mu\big(2\sqrt{2 \rho x}/\nu_0\big)]^2$. By using the same above Laplace inversion identities,
the corresponding transition PDF is given explicitly via (\ref{Laplace_inversion}):
\begin{eqnarray}
p^{(\rho)}_X(t;x_0,x) = {q_1 I_{\vert\mu\vert}\big({2\over\nu_0}\sqrt{2\rho x}\big)
+ q_2 K_\mu\big({2\over\nu_0}\sqrt{2\rho x}\big)
\over q_1 I_{\vert\mu\vert}\big({2\over\nu_0}\sqrt{2\rho x_0}\big)
+ q_2 K_\mu\big({2\over\nu_0}\sqrt{2\rho x_0}\big)}
{e^{-\rho t -2(x + x_0)/\nu_0^2t} \over \nu_0^2t/2}
     I_{\vert\mu\vert}\left({4\sqrt{xx_0}\over \nu_0^2t}\right)\,,
 \label{PDF_Bessel_2}
\end{eqnarray}
$x,x_0>0$, $t>0$, $\mu\in\R$. Again, note that this has the form in (\ref{prho}).

For all processes with $q_2 > 0$, i.e. subfamilies (i) and (iii), the origin is attainable and the distribution of
the first-hitting time at the origin, $\tau^{(\rho)}_0 = \inf\{t\ge 0: X^{(\rho)}_t = 0\}$,
for the $X^{(\rho)}$-diffusion started at $X^{(\rho)}_0=x_0>0$
is readily computed in closed form.
Following the theory in \cite{CM06,Campolieti2008}, the PDF of $\tau^{(\rho)}_0$ is given by the limit
\begin{equation}\label{FHT_PDF}
f^{(\rho)}(t;x_0,0) = \displaystyle\frac{1}{\s_\rho(x)} \frac{\partial}{\partial
  x}\left(\frac{p^{(\rho)}_X(t;x_0,x)}{\m_\rho(x)}\right)\bigg|_{x=0+}.
\end{equation}
By substituting (\ref{PDF_Bessel_2}) and making use of the small argument
asymptotics of the modified Bessel functions, the above limit is computed explicitly to give
\begin{equation}
f^{(\rho)}(t;x_0,0) = P_0\frac{(b/a)^{\nu/2}}{2K_\nu(\sqrt{ab})}t^{-\nu-1}e^{-(at+b/t)/2}\,,
\end{equation}
where $t>0$, $a=2\rho$, $b=4x_0/\nu_0^2$, $\nu=\vert\mu\vert$,
and $P_0 = \P\left( \tau^{(\rho)}_0 < \infty\right) =\frac{q_2K_\nu(\sqrt{ab})}{q_1I_\nu(\sqrt{ab})+q_2K_\nu(\sqrt{ab})}$ is
the probability for eventually hitting the origin.
For subfamily (i), i.e. $q_1 = 0$, the point at infinity is non-attracting and hence $P_0 = 1$.
For subfamily (iii), $q_1 > 0$, the point at infinity is attracting so $P_0 < 1$.
The first-hitting time at the origin has the generalized inverse Gaussian distribution. The above PDF
generalizes that obtained in \cite{CM06} where only $\mu>0$ was considered.

We simply note here that for $q_2>0$ and $\vert\mu\vert<1$, the regular boundary at $0$ can also be specified as instantaneously
reflecting. In this case, the function $\phi^{(\rho)}_s(x)$ is still given as just above, whereas $\psi^{(\rho)}_s(x)$ must now satisfy
the reflecting condition: $\frac{1}{\s_\rho(0+)}\frac{d\psi^{(\rho)}_\alpha(0+)}{dx}=0$.
The Green function in (\ref{greenfunc_rho}) can then be explicitly constructed from these functions and their Wronskian.
Laplace inversion via (\ref{Laplace_inversion}) leads to the transition PDF for
the Bessel family of $X^{(\rho)}$-diffusions that are reflected at the origin.

\subsection{The CIR Process}\label{subsect2.2}
Consider the Cox-Ingerssol-Ross (CIR) process \cite{CIRref}
$(X_t)_{t\ge 0}$ on the regular state space $\I=(0,\infty)\equiv \R_+$ with SDE
$dX_t=(\lambda_0-\lambda_1X_t)dt+\nu_0\sqrt{X_t}dW_t$,
with parameters $\lambda_0, \lambda_1\in \R$, $\lambda_1\ne 0$, $\nu_0>0$. By defining the parameters
$\mu\equiv\frac{2\lambda_0}{\nu_0^2}-1$ and $\kappa \equiv \frac{2\lambda_1}{\nu_0^2}$,
the generator for this diffusion is then given by
$(\G\, f)(x)\triangleq {1\over 2}\nu_0^2[xf''(x) + (\mu + 1 - \kappa x)f'(x)]$.
The respective scale and speed densities are $\s(x) = x^{-\mu-1}e^{\kappa x}$
and $\m(x) = \frac{2}{\nu_0^2}x^{\mu}e^{-\kappa x}= \frac{\kappa}{\lambda_1}x^{\mu}e^{-\kappa x}$.
[We note that this process is also the squared radial Ornstein-Uhlenbeck process.
In fact, setting $\nu_0=2$, and renaming the parameters $\mu\to\nu$, $\kappa \to \gamma$,
$\lambda_1 \to 2\gamma$ recovers precisely the process discussed on pages 140--142 of \cite{BS02}.]
The endpoint $x=\infty$ is natural. The origin is an entrance if $\mu\geq 0$, is a regular
boundary if $-1<\mu<0$, and is an exit if $\mu\leq -1$.

We begin by fixing a suitable pair of fundamental solutions to $(\G\, \varphi)(x) = s \varphi(x)$ such that,
for positive real values of $s=\rho>0$, $\phi^+_\rho(x)$ and
$\phi^-_\rho(x)$ are respectively increasing and decreasing strictly positive
functions of $x\in\R_+$. We now choose such a pair of fundamental solutions as follows.
For the case (a) $\kappa > 0$ (i.e. $\lambda_1 > 0$) we take
\begin{eqnarray} \label{CIRfund_1_M}
\varphi^+_s(x) = (\kappa x)^{-(\mu+1)/2}e^{\kappa x/2}\,M_{-\frac{s}{\lambda_1} + \frac{\mu+1}{2}\,,\, \frac{\vert \mu \vert}{2}}(\kappa x)
= (\kappa x)^{\mu_-}\!\M\bigg(\frac{s}{\lambda_1} + \mu_-,1 + \vert\mu\vert,\kappa x\bigg)
\end{eqnarray}
and
\begin{eqnarray} \label{CIRfund_1_U}
\varphi^-_s(x) = (\kappa x)^{-(\mu+1)/2}e^{\kappa x/2}\,W_{-\frac{s}{\lambda_1} + \frac{\mu+1}{2}\,,\, \frac{\vert \mu \vert}{2}}(\kappa x)
&=& (\kappa x)^{\mu_-}\U\bigg(\frac{s}{\lambda_1} + \mu_-,1 + \vert\mu\vert,\kappa x\bigg)
\nonumber\\
&=& \U\bigg(\frac{s}{\lambda_1},1 + \mu,\kappa x\bigg),
\end{eqnarray}
where $\mu_- = (\vert\mu\vert - \mu)/2 = \max\{0,-\mu\}$ is the negative part of $\mu$.
In what follows we also denote $\mu_+ = (\vert\mu\vert + \mu)/2 = \max\{0,\mu\}$
for the positive part of $\mu$.
The functions $\M(a,b,z)$ and $\U(a,b,z)$ are confluent hypergeometric functions,
i.e. the standard Kummer and Tricomi functions, respectively.
Note that the alternate forms in (\ref{CIRfund_1_M}) and (\ref{CIRfund_1_U})
follow from the relations $M_{k,m}(z) = z^{m + 1/2}e^{-z/2}\M({1\over 2} + m - k, 1 + 2m,z)$ and
$W_{k,m}(z) = z^{m + 1/2}e^{-z/2}\U({1\over 2} + m - k, 1 + 2m,z)$
where $M$ and $W$ are the Whittaker functions (see \cite{AS72}). The last expression in (\ref{CIRfund_1_U}) follows
from the Kummer transformation identity $z^{b-1}\U(a,b,z) = \U(1+a - b,2-b,z)$.
For $s_1,s_2\in\C$, $W[\varphi^\pm_{s_1},\varphi^\pm_{s_2}](x)$ and $W[\varphi^-_{s_1},\varphi^+_{s_2}](x)$
are obtained by using differential recurrences ${d\over dz}\M(a,b,z) = (a/b)\M(a+1,b+1,z)$ and
${d\over dz}\U(a,b,z) = -a \U(a+1,b+1,z)$.
The above functions $\varphi^\pm_s(x)$ satisfy (\ref{wronskian}) where
$w_s = \kappa^{-\mu}\frac{\Gamma(1 + \vert\mu\vert)}{\Gamma\big(\frac{s}{\lambda_1} + \mu_-\big)}$,
i.e. $w_s = \kappa^{-\mu}\frac{\Gamma(1 + \mu)}{\Gamma(s/\lambda_1)}$ for $\mu\ge 0$ and
$w_s = \kappa^{-\mu}\frac{\Gamma(1 - \mu)}{\Gamma(\frac{s}{\lambda_1} - \mu)}$ for $\mu < 0$.

For case (b) $\kappa < 0$ (i.e. $\lambda_1 < 0$) we take
\footnote{We note typographical errors at the bottom of page 142 in \cite{BS02}.
The factors $e^{-\gamma x/2}$ and $e^{-\gamma y/2}$
should instead be $e^{-\vert\gamma\vert x/2}$ and $e^{-\vert\gamma\vert y/2}$
in both Green functions for case (B) $\gamma < 0$. Also, the parameter $\theta$ should be $\gamma$
in the last Wronskian $\omega_\alpha$ on page 142 and in the Green function on page 141.}
\begin{align} \label{CIRfund_2_M}
\varphi^+_s(x) &= (\vert\kappa\vert x)^{-(\mu+1)/2}e^{-\vert\kappa\vert x/2}\,M_{-\frac{s}{\vert\lambda_1\vert} - \frac{\mu+1}{2}\,,\,
\frac{\vert \mu \vert}{2}}(\vert\kappa\vert x)
\nonumber \\
&= (\vert\kappa\vert x)^{\mu_-}e^{-\vert\kappa\vert x}\M\bigg(\frac{s}{\vert\lambda_1\vert} + 1 + \mu_+,1 +
\vert\mu\vert,\vert\kappa\vert x\bigg)
\end{align}
and
\begin{align} \label{CIRfund_2_U}
\varphi^-_s(x) &= (\vert\kappa\vert x)^{-(\mu+1)/2}e^{-\vert\kappa\vert x/2}\,W_{-\frac{s}{\vert\lambda_1\vert} - \frac{\mu+1}{2}\,,\, \frac{\vert \mu \vert}{2}}(\vert\kappa\vert x)
\nonumber \\
&= (\vert\kappa\vert x)^{\mu_-}e^{-\vert\kappa\vert x}\U\bigg(\frac{s}{\vert\lambda_1\vert} + 1 + \mu_+,1 +
\vert\mu\vert,\vert\kappa\vert x\bigg) =
e^{-\vert\kappa\vert x}\,\U\bigg(\frac{s}{\vert\lambda_1\vert} + 1 + \mu,1 + \mu,\vert\kappa\vert x\bigg).
\end{align}
The Wronskian between these two functions is given by (\ref{wronskian}) where
$w_s = \vert\kappa\vert^{-\mu}\frac{\Gamma(1 + \vert\mu\vert)}{\Gamma\big(\frac{s}{\vert\lambda_1\vert} + 1 + \mu_+\big)}$,
i.e. $w_s = \vert\kappa\vert^{-\mu}\frac{\Gamma(1 + \mu)}{\Gamma\big(\frac{s}{\vert\lambda_1\vert} + 1 + \mu\big)}$ for $\mu\ge 0$ and
$w_s = \vert\kappa\vert^{-\mu}\frac{\Gamma(1 - \mu)}{\Gamma\big(\frac{s}{\vert\lambda_1\vert} + 1\big)}$ for $\mu < 0$.

The Green function in (\ref{greenfunc}) for the CIR on $\R_+$, for the two cases $\lambda_1 > 0$ and $\lambda_1 < 0$,
is readily constructed by taking appropriate linear combinations of the functions $\varphi^\pm_s$.
In all cases, $\phi_s(x) = \varphi^-_s(x)$ since $\infty$ is a natural boundary. In case $\lambda_1 > 0$ then
\begin{equation}\label{CIR_psi_and_phi_1}
\psi_s(x) = \M\bigg(\frac{s}{\lambda_1}, 1 + \mu,\kappa x\bigg),
\,\,\phi_s(x) = \U\bigg(\frac{s}{\lambda_1}, 1 + \mu,\kappa x\bigg)
\end{equation}
if $\mu \ge 0$ or if $\mu\in(-1,0)$ and 0 is reflecting;
\begin{equation}\label{CIR_psi_and_phi_2}
\psi_s(x) = (\kappa x)^{-\mu}\!\M\bigg(\frac{s}{\lambda_1} - \mu,1 - \mu,\kappa x\bigg),\,\,
\phi_s(x) = (\kappa x)^{-\mu}\U\bigg(\frac{s}{\lambda_1} - \mu,1 - \mu,\kappa x\bigg)
\end{equation}
if $\mu \le -1$ or if $\mu\in(-1,0)$ and 0 is killing. The Wronskian
factor is ${\mathcal W}_s = \kappa^{-\mu}\frac{\Gamma(1 + \mu)}{\Gamma(s/\lambda_1)}$
for the pair in (\ref{CIR_psi_and_phi_1}) and
${\mathcal W}_s = \kappa^{-\mu}\frac{\Gamma(1 - \mu)}{\Gamma({s\over \lambda_1} - \mu)}$
for the pair in (\ref{CIR_psi_and_phi_2}).
For the case $\lambda_1 < 0$:
\begin{equation}\label{CIR_psi_and_phi_3}
\psi_s(x) = e^{-\vert\kappa\vert x}\M\bigg(\frac{s}{\vert\lambda_1\vert} + 1 + \mu,1 +
\mu,\vert\kappa\vert x\bigg),\,\,\,
\phi_s(x) = \varphi^-_s(x)
\end{equation}
if $\mu \ge 0$ or if $\mu\in(-1,0)$ and 0 is reflecting;
\begin{equation}\label{CIR_psi_and_phi_4}
\psi_s(x) = (\vert\kappa\vert x)^{-\mu} e^{-\vert\kappa\vert x}
\M\bigg(\frac{s}{\vert\lambda_1\vert} + 1, 1 - \mu,\vert\kappa\vert x\bigg),\,\,
\phi_s(x) = \varphi^-_s(x)
\end{equation}
if $\mu \le -1$ or if $\mu\in(-1,0)$ and 0 is killing,
where $\varphi^-_s(x)$ is defined in (\ref{CIRfund_2_U}). The Wronskian
factor ${\mathcal W}_s = \vert\kappa\vert^{-\mu}\frac{\Gamma(1 + \mu)}
{\Gamma\big({s\over \vert\lambda_1\vert} + \mu + 1\big)}$
for (\ref{CIR_psi_and_phi_3}) and
${\mathcal W}_s = \vert\kappa\vert^{-\mu}\frac{\Gamma(1 - \mu)}
{\Gamma\big({s\over \vert\lambda_1\vert} + 1\big)}$ for (\ref{CIR_psi_and_phi_4}).

In all of the above cases, the transition PDF for $X_t\in\R_+$ is readily obtained by
Laplace inverting the relevant Green function with the use of the identity
\begin{align}\label{Kummer_Laplace}
&{\mathcal L}_s^{-1}[\Gamma(s)\M\big(s,1+\mu,x\big)
\U\big(s,1+\mu,y\big)][t]
\nonumber \\
&= {\Gamma(1+\mu)e^{(1+\mu)t/2}\over 2(xy)^{\mu/2}\sinh(t/2)}
\exp\bigg(-{(x + y)e^{-t/2}\over 2\sinh(t/2)} \bigg)
I_\mu\bigg({\sqrt{xy}\over \sinh(t/2)}\bigg),
\end{align}
for $0 < x\le y, \mu>-1, t>0$ and the property
${\mathcal L}_s^{-1}[F(s/a + b)][t] = a e^{-a bt}{\mathcal L}_s^{-1}[F(s)][a t]$ for $a,b>0$.
In case $\lambda_1 > 0$, Laplace inversion of the Green function in (\ref{greenfunc}) for
$\mu \ge 0$, or $\mu\in(-1,0)$ and 0 as reflecting, gives the known transition PDF:
\begin{align} \label{CIRden1}
 p_X(t;x_0,x) &= \m(x){\kappa^\mu \over \Gamma(1+\mu)}{\mathcal L}_s^{-1}[\Gamma\big({s\over \lambda_1}\big)
 \M\big({s\over \lambda_1},1+\mu,\kappa x_<\big)
\U\big({s\over \lambda_1},1+\mu,\kappa x_>\big)][t]
 \nonumber \\
 &= {\kappa e^{(1+\mu)\lambda_1 t/2}e^{-\kappa x}(x/x_0)^{\mu \over 2} \over 2\sinh(\lambda_1t/2)}
\exp\left(-{\kappa e^{-\lambda_1 t/2}(x + x_0) \over 2\sinh(\lambda_1t/2)} \right)
I_\mu\left({\kappa\sqrt{xx_0}\over \sinh(\lambda_1 t/2)}\right).
\end{align}
Similarly, the known transition PDF for $X_t\in\R_+$ in case $\lambda_1 > 0$ and $\mu \le -1$,
or $\mu\in(-1,0)$ and 0 as killing, is given by
\begin{align} \label{CIRden2}
 p_X(t;x_0,x) &= \m(x){\kappa^{-\mu}(xx_0)^{-\mu} \over \Gamma(1-\mu)}
 {\mathcal L}_s^{-1}[\Gamma({s\over \lambda_1} - \mu)\M\big({s\over \lambda_1}-\mu,1-\mu,\kappa x_<\big)
\U\big({s\over \lambda_1}-\mu,1-\mu,\kappa x_>\big)][t]
 \nonumber \\
 &= {\kappa e^{(1+\mu)\lambda_1 t/2}e^{-\kappa x}(x/x_0)^{\mu \over 2} \over 2\sinh(\lambda_1t/2)}
\exp\left(-{\kappa e^{-\lambda_1 t/2}(x + x_0) \over 2\sinh(\lambda_1t/2)} \right)
I_{\vert\mu\vert}\left({\kappa\sqrt{xx_0}\over \sinh(\lambda_1 t/2)}\right).
\end{align}

In case $\lambda_1 < 0$ ($\kappa < 0$), the transition PDFs follow similarly by Laplace inversion:
\begin{align} \label{CIRden3}
 p_X(t;x_0,x) &= \m(x){\vert\kappa\vert^\mu \over \Gamma(1+\mu)}e^{\kappa (x + x_0)}
 \nonumber \\
 &{\mathcal L}_s^{-1}[\Gamma({s\over \vert\lambda_1\vert}+ 1 + \mu)\M\big({s\over \vert\lambda_1\vert}+ 1 + \mu,1+\mu,\vert\kappa\vert x_<\big)
\U\big({s\over \vert\lambda_1\vert}+ 1 + \mu,1+\mu,\vert\kappa\vert x_>\big)][t]
 \nonumber \\
 &= {\kappa e^{(1+\mu)\lambda_1 t/2}e^{\kappa x_0}(x/x_0)^{\mu \over 2} \over 2\sinh(\lambda_1t/2)}
\exp\left(-{\kappa e^{\lambda_1 t/2}(x + x_0) \over 2\sinh(\lambda_1t/2)} \right)
I_\mu\left({\kappa\sqrt{xx_0}\over \sinh(\lambda_1 t/2)}\right)
\end{align}
for $\mu \ge 0$, or $\mu\in(-1,0)$ and 0 as reflecting, and
\begin{align} \label{CIRden4}
 p_X(t;x_0,x) &= \m(x){\vert\kappa\vert^{-\mu} \over \Gamma(1-\mu)}(xx_0)^{-\mu}e^{\kappa (x + x_0)}
 \nonumber \\
 &{\mathcal L}_s^{-1}[\Gamma({s\over \vert\lambda_1\vert}+ 1)\M\big({s\over \vert\lambda_1\vert}+ 1,1-\mu,\vert\kappa\vert x_<\big)
\U\big({s\over \vert\lambda_1\vert}+ 1,1-\mu,\vert\kappa\vert x_>\big)][t]
 \nonumber \\
 &= {\kappa e^{(1+\mu)\lambda_1 t/2}e^{\kappa x_0}(x/x_0)^{\mu \over 2} \over 2\sinh(\lambda_1t/2)}
\exp\left(-{\kappa e^{\lambda_1 t/2}(x + x_0) \over 2\sinh(\lambda_1t/2)} \right)
I_{\vert\mu\vert}\left({\kappa\sqrt{xx_0}\over \sinh(\lambda_1 t/2)}\right)
\end{align}
for $\mu \le -1$, or $\mu\in(-1,0)$ and 0 as killing.

The confluent hypergeometric (CIR) family of $X^{(\rho)}$-diffusions has generator in (\ref{Generator_rho}) with
respectively increasing and decreasing positive functions
$\varphi^+_\rho(x)$ and $\varphi^-_\rho(x)$ defined by (\ref{CIRfund_1_M}) and (\ref{CIRfund_1_U}) for $\lambda_1 > 0$ and
by (\ref{CIRfund_2_M}) and (\ref{CIRfund_2_U}) for $\lambda_1 < 0$, where $s=\rho > 0$.
The following lemma gives the boundary classification for these processes.

\begin{lemma} \label{lemma_class_CIR}
The confluent hypergeometric family of $X^{(\rho)}$-processes on $\R_+$ have the same boundary classification as the Bessel family of
$X^{(\rho)}$-processes, as stated in Lemma~\ref{lemma_class_Bessel} in terms of the parameters $q_1,q_2$ and $\mu$.
\end{lemma}
\begin{proof}
The results follow from Lemma~\ref{Lemma_Xrho} and the asymptotics for $\varphi_\rho^\pm$
in Appendix~\ref{subsect_a12}. For both cases (a) and (b): $n(0+;\rho,\rho)=0$,
$n(\infty;\rho,\rho)=\infty$, $(\varphi_\rho^+,\varphi_\rho^-)_{(0,x]} < \infty$,
$(\varphi_\rho^-,\varphi_\rho^-)_{(0,x]} < \infty$ if and only if $\vert\mu\vert < 1$, and
$(\varphi_\rho^+,\varphi_\rho^-)_{[x,\infty)} = \infty$.
\end{proof}
As in the Bessel family of $X^{(\rho)}$-processes, the point at infinity is a natural boundary.
The confluent hypergeometric subfamily of type (ii), with $q_1>0, q_2=0$, is conservative with the origin as an
entrance. For subfamilies (i) and (iii), where $q_2 > 0$, the origin is attainable (regular for $|\mu| < 1$ and
exit for $|\mu|\geq 1$). Green functions in the form of (\ref{greenfunc_rho}) are readily obtained and the
corresponding transition PDF for the confluent hypergeometric $X^{(\rho)}$-processes are then given by Laplace inversion.

Consider case (a) $\lambda_1 > 0$. For subfamily (ii), $q_2=0$ and we take
$\hat{u}_\rho(x) = \varphi^+_{\rho}(x)$. The Green function in (\ref{greenfunc_rho}) is uniquely
specified by $\psi^{(\rho)}_s(x) = \varphi^+_{s+\rho}(x) / \varphi^+_{\rho}(x)$ and
$\phi^{(\rho)}_s(x) = \varphi^-_{s+\rho}(x) / \varphi^+_{\rho}(x)$ using (\ref{CIRfund_1_M}) and (\ref{CIRfund_1_U}):
\begin{equation*}
\psi^{(\rho)}_s(x) = {\M\bigg(\frac{s+\rho}{\lambda_1} + \mu_-,1 + \vert\mu\vert,\kappa x\bigg)
\over \M\bigg(\frac{\rho}{\lambda_1} + \mu_-,1 + \vert\mu\vert,\kappa x\bigg)}
\mbox{ \ , \ }
\phi^{(\rho)}_s(x) = {\U\bigg(\frac{s+\rho}{\lambda_1} + \mu_-,1 + \vert\mu\vert,\kappa x\bigg)
\over \M\bigg(\frac{\rho}{\lambda_1} + \mu_-,1 + \vert\mu\vert,\kappa x\bigg)}
\end{equation*}
where ${\mathcal W}^{(\rho)}_s = w_{s+\rho} = \kappa^{-\mu}\frac{\Gamma(1 + \vert\mu\vert)}{\Gamma\big(\frac{s+\rho}{\lambda_1} + \mu_-\big)}$ and
$\m_\rho(x) = \m(x)(\kappa x)^{2\mu_-}\M^2\bigg(\frac{\rho}{\lambda_1} + \mu_-,1 + \vert\mu\vert,\kappa x\bigg)$.
The transition PDF for the confluent subfamily (ii) follows explicitly by (\ref{Laplace_inversion}), using (\ref{Kummer_Laplace}):
\begin{align} \label{CIRden_rho1}
 p^{(\rho)}_X(t;x_0,x) &= \m(x) {\hat{u}_\rho(x) \over \hat{u}_\rho(x_0)}e^{-\rho t}{\kappa^{\mu} \over \Gamma(1+\vert\mu\vert)}
  \nonumber \\
  &\times {\mathcal L}_s^{-1}[\Gamma({s\over \lambda_1} + \mu_-)\M\big({s\over \lambda_1} + \mu_-, 1 + \vert\mu\vert,\kappa x_<\big)
\U\big({s\over \lambda_1} + \mu_-, 1 + \vert\mu\vert,\kappa x_>\big)][t]
 \nonumber \\
 &= {x^{\mu_-}\M\big(\frac{\rho}{\lambda_1} + \mu_-,1 + \vert\mu\vert,\kappa x\big)
 \over x_0^{\mu_-}\M\big(\frac{\rho}{\lambda_1} + \mu_-,1 + \vert\mu\vert,\kappa x_0\big)}e^{-\rho t}p_X(t;x_0,x)
\end{align}
where $p_X(t;x_0,x)$ is given by the expression in (\ref{CIRden2}) where $x,x_0,t > 0, \mu\in\R$.
Note that (\ref{CIRden_rho1}) has the form in (\ref{prho}).

For subfamilies (i) and (iii), i.e. $q_1\ge 0, q_2 > 0$, the origin is exit for $|\mu|\geq 1$ and regular for
$|\mu| < 1$. In particular, by specifying the origin as killing for $|\mu| < 1$ then
the process is absorbed at the origin for all $\mu \in \R$.
In this case, the Green function in (\ref{greenfunc_rho}) is formed by taking
$\psi^{(\rho)}_s(x) = \varphi^+_{\rho + s}(x)/\hat{u}_{\rho}(x)$ and
$\phi^{(\rho)}_s(x) = \varphi^-_{\rho + s}(x)/\hat{u}_{\rho}(x)$ using (\ref{CIRfund_1_M}) and (\ref{CIRfund_1_U}):
\begin{align*}
\psi^{(\rho)}_s(x) &= {\M\big(\frac{s+\rho}{\lambda_1} + \mu_-,1 + \vert\mu\vert,\kappa x\big)
\over q_1 \M\big(\frac{\rho}{\lambda_1} + \mu_-,1 + \vert\mu\vert,\kappa x\big) +
q_2\, \U\big(\frac{\rho}{\lambda_1} + \mu_-,1 + \vert\mu\vert,\kappa x\big)}\,,
\nonumber \\
\phi^{(\rho)}_s(x) &= {\U\big(\frac{s+\rho}{\lambda_1} + \mu_-,1 + \vert\mu\vert,\kappa x\big)
\over q_1 \M\big(\frac{\rho}{\lambda_1} + \mu_-,1 + \vert\mu\vert,\kappa x\big) +
q_2\, \U\big(\frac{\rho}{\lambda_1} + \mu_-,1 + \vert\mu\vert,\kappa x\big)}\,,
\end{align*}
with Wronskian ${\mathcal W}^{(\rho)}_s = w_{s+\rho} = \kappa^{-\mu}\frac{\Gamma(1 + \vert\mu\vert)}{\Gamma\big(\frac{s+\rho}{\lambda_1} + \mu_-\big)}$ and speed measure $\m_\rho(x) = \m(x)\hat{u}^2_\rho(x)$.
The transition PDF for this case is obtained in the same fashion as (\ref{CIRden_rho1}):
\begin{align} \label{CIRden_rho2}
 p^{(\rho)}_X(t;x_0,x) = \left({x \over x_0}\right)^{\mu_-}\!\!{q_1 \M\big(\frac{\rho}{\lambda_1} + \mu_-,1 + \vert\mu\vert,\kappa x\big) +
q_2\, \U\big(\frac{\rho}{\lambda_1} + \mu_-,1 + \vert\mu\vert,\kappa x\big)
 \over q_1 \M\big(\frac{\rho}{\lambda_1} + \mu_-,1 + \vert\mu\vert,\kappa x_0\big) +
q_2\, \U\big(\frac{\rho}{\lambda_1} + \mu_-,1 + \vert\mu\vert,\kappa x_0\big)}
e^{-\rho t}p_X(t;x_0,x)
\end{align}
where $p_X(t;x_0,x)$ is given by the expression in (\ref{CIRden2}) for $x,x_0,t > 0, \mu\in\R$.

For subfamilies (i) and (iii), i.e. $q_2 > 0$, the origin is attainable and the PDF $f^{(\rho)}(t;x_0,0)$ of
the first hitting time $\tau^{(\rho)}_0$ at the origin for the $X^{(\rho)}$-diffusion started at $X^{(\rho)}_0=x_0>0$
is obtained in closed form by substituting the transition density in (\ref{CIRden_rho2}) into (\ref{FHT_PDF}). The limit is computed explicitly by
using the small argument asymptotics of the confluent hypergeometric functions and the Bessel-$I$ function.
This gives the so-called Tricomi PDF:
\begin{align}\label{FHT_PDF_CIR_1}
 f^{(\rho)}(t;x_0,0) = P_0 f^{(\rho)}_{\,\U}(t;x_0,0)
 \end{align}
 where
 \begin{align}\label{FHT_PDF_CIR_2}
 f^{(\rho)}_{\,\U}(t;x_0,0) = |\tau'(t)|\frac{e^{-z\tau(t)} [\tau(t)]^{a-1} [1+\tau(t)]^{b-a-1}}{\mathcal{U}(a,b,z)\Gamma(a)}
 \end{align}
is the corresponding PDF of the first hitting time at the origin for the (confluent-$\U$) subfamily (i) when $q_1 = 0$;
$\tau(t)=(e^{\lambda_1 t}-1)^{-1}$, $|\tau'(t)| = \lambda_1e^{\lambda_1 t}[\tau(t)]^2$,
$t>0$, $z=\kappa x_0$, $a=\rho/\lambda_1 + \mu_-$, $b=1+\vert\mu\vert$. The quantity
$P_0 = \P\big( \tau^{(\rho)}_0 < \infty \big) = \frac{q_2\mathcal{U}(a,b,z)}{q_1\mathcal{M}(a,b,z)+q_2\mathcal{U}(a,b,z)}$
is the probability for eventually hitting the origin.
For subfamily (i), i.e. $q_1 = 0$, the point at infinity is non-attracting and hence $P_0 = 1$.
Indeed, from the integral representation of the Tricomi function, we observe that the PDF integrates to unity, i.e.
$\P\left( \tau_0 < \infty\right) = \int_0^\infty f^{(\rho)}_{\,\U}(t;x_0,0) dt = 1$.
However, if $q_1 > 0$ then $P_0 < 1$ since $\infty$ is attracting. The above first hitting time PDFs
are valid for all $\mu \in \R$ and extend previously derived results for the case $\mu > 0$\cite{CM06}.

For case (b) $\lambda_1 < 0$ the analysis follows in similar fashion as above. For
subfamily (ii), $q_2=0$, $\hat{u}_\rho(x) = \varphi^+_{\rho}(x)$.
The Green function in (\ref{greenfunc_rho}) is uniquely
specified by $\psi^{(\rho)}_s(x) = \varphi^+_{s+\rho}(x) / \varphi^+_{\rho}(x)$ and
$\phi^{(\rho)}_s(x) = \varphi^-_{s+\rho}(x) / \varphi^+_{\rho}(x)$ using (\ref{CIRfund_2_M}) and (\ref{CIRfund_2_U}):
\begin{equation*}
\psi^{(\rho)}_s(x) = {\M\bigg(\frac{s+\rho}{\lambda_1} + 1 + \mu_+, 1 + \vert\mu\vert,\vert\kappa\vert x\bigg)
\over \M\bigg(\frac{\rho}{\lambda_1} + 1 + \mu_+,1 + \vert\mu\vert,\vert\kappa\vert x\bigg)}
\mbox{ \ , \ }
\phi^{(\rho)}_s(x) = {\U\bigg(\frac{s+\rho}{\lambda_1} + 1 + \mu_+, 1 + \vert\mu\vert,\vert\kappa\vert x\bigg)
\over \M\bigg(\frac{\rho}{\lambda_1} + 1 + \mu_+,1 + \vert\mu\vert,\vert\kappa\vert x\bigg)},
\end{equation*}
${\mathcal W}^{(\rho)}_s = w_{s+\rho} = \vert\kappa\vert^{-\mu}\frac{\Gamma(1 + \vert\mu\vert)}{\Gamma\big(\frac{s+\rho}{\lambda_1} + 1 + \mu_+\big)}$,
$\m_\rho(x) = \m(x)[\varphi^+_{\rho}(x)]^2$.
The transition PDF for this confluent subfamily (ii) follows by (\ref{Laplace_inversion}). In this case,
\begin{align} \label{CIRden_rho3}
 p^{(\rho)}_X(t;x_0,x)
 &= \m(x){\hat{u}_\rho(x) \over \hat{u}_\rho(x_0)}e^{-\rho t}{\mathcal L}_s^{-1}[w_s^{-1} \varphi^+_{s}(x_<)\varphi^-_{s}(x_>)][t]
\nonumber \\
 &= {x^{\mu_-}e^{\kappa x}\M\big(\frac{\rho}{\lambda_1} + 1 + \mu_+, 1 + \vert\mu\vert,\vert\kappa\vert x\big)
 \over x_0^{\mu_-}e^{\kappa x_0}\M\big(\frac{\rho}{\lambda_1} + 1 + \mu_+, 1 + \vert\mu\vert,\vert\kappa\vert x_0\big)}e^{-\rho t}p_X(t;x_0,x)
\end{align}
where $p_X(t;x_0,x)$ is given by (\ref{CIRden4}) for $x,x_0,t > 0, \mu\in\R$.
Note that (\ref{CIRden_rho3}) has the form in (\ref{prho}).

For subfamilies (i) and (iii), the origin is exit for $|\mu|\geq 1$ and regular for
$|\mu| < 1$. By specifying the origin as killing for $|\mu| < 1$ then
the process is absorbed at the origin for all $\mu \in \R$.
In this case, the Green function in (\ref{greenfunc_rho}) is formed by taking
$\psi^{(\rho)}_s(x) = \varphi^+_{\rho + s}(x)/\hat{u}_{\rho}(x)$ and
$\phi^{(\rho)}_s(x) = \varphi^-_{\rho + s}(x)/\hat{u}_{\rho}(x)$ using (\ref{CIRfund_2_M}) and (\ref{CIRfund_2_U}):
\begin{align*}
\psi^{(\rho)}_s(x) &= {\M\big(\frac{s+\rho}{\lambda_1} + 1 + \mu_+, 1 + \vert\mu\vert,\vert\kappa\vert x\big)
\over q_1 \M\big(\frac{\rho}{\lambda_1} + 1 + \mu_+,1 + \vert\mu\vert,\vert\kappa\vert x\big) +
q_2\, \U\big(\frac{\rho}{\lambda_1} + 1 + \mu_+,1 + \vert\mu\vert,\vert\kappa\vert x\big)}\,,
\nonumber \\
\phi^{(\rho)}_s(x) &= {\U\big(\frac{s+\rho}{\lambda_1} + 1 + \mu_+, 1 + \vert\mu\vert,\vert\kappa\vert x\big)
\over q_1 \M\big(\frac{\rho}{\lambda_1} + 1 + \mu_+,1 + \vert\mu\vert,\vert\kappa\vert x\big) +
q_2\, \U\big(\frac{\rho}{\lambda_1} + 1 + \mu_+,1 + \vert\mu\vert,\vert\kappa\vert x\big)}\,,
\end{align*}
with Wronskian ${\mathcal W}^{(\rho)}_s = w_{s+\rho} =
\vert\kappa\vert^{-\mu}\frac{\Gamma(1 + \vert\mu\vert)}{\Gamma\big(\frac{s+\rho}{\lambda_1} + 1 + \mu_+\big)}$
and speed measure $\m_\rho(x) = \m(x)\hat{u}^2_\rho(x)$.
The transition PDF for this case follows as in (\ref{CIRden_rho3}):
\begin{align} \label{CIRden_rho4}
 p^{(\rho)}_X(t;x_0,x) =& {x^{\mu_-}e^{\kappa x}[q_1 \M\big(\frac{\rho}{\lambda_1} + 1 + \mu_+,1 + \vert\mu\vert,\vert\kappa\vert x\big) +
q_2\, \U\big(\frac{\rho}{\lambda_1} + 1 + \mu_+,1 + \vert\mu\vert,\vert\kappa\vert x\big)]
 \over x_0^{\mu_-}e^{\kappa x_0}[q_1 \M\big(\frac{\rho}{\lambda_1} + 1 + \mu_+,1 + \vert\mu\vert,\vert\kappa\vert x_0\big) +
q_2\, \U\big(\frac{\rho}{\lambda_1} + 1 + \mu_+,1 + \vert\mu\vert,\vert\kappa\vert x_0\big)]}
\nonumber \\
&\,\,\times e^{-\rho t}p_X(t;x_0,x)
\end{align}
where $p_X(t;x_0,x)$ is given by (\ref{CIRden4}), for $x,x_0,t > 0, \mu\in\R$, and hence (\ref{CIRden_rho4}) has the form in (\ref{prho}).

For $\lambda_1 < 0$, the origin is attainable for subfamilies (i) and (iii) with $q_2 > 0$. The PDF $f^{(\rho)}(t;x_0,0)$ of
the first hitting time at the origin follows by a similar derivation to the above case for $\lambda_1 > 0$.
Substituting (\ref{CIRden_rho4}) into (\ref{FHT_PDF}) and computing the limit explicitly by
using the small argument asymptotics of the confluent hypergeometric functions and the Bessel-$I$ function recovers
(\ref{FHT_PDF_CIR_1}) and (\ref{FHT_PDF_CIR_2}), where (\ref{FHT_PDF_CIR_2}) is the PDF of the first hitting time at the origin
for the (confluent-$\U$) subfamily (i) when $q_1 = 0$. However, now
$\tau(t)=(e^{\vert\lambda_1\vert t}-1)^{-1}$, $|\tau'(t)| = \vert\lambda_1\vert e^{\vert\lambda_1\vert t}[\tau(t)]^2$,
$t>0$, $z=\vert\kappa\vert x_0$, $a=\rho/\vert\lambda_1\vert + 1 + \mu_+$, $b =1+\vert\mu\vert$.
The quantity $P_0 = \P\big( \tau^{(\rho)}_0 < \infty\big)
= \frac{q_2\mathcal{U}(a,b,z)}{q_1\mathcal{M}(a,b,z)+q_2\mathcal{U}(a,b,z)}$
is again the probability for eventually hitting the origin.
The above first hitting time PDFs for both cases $\lambda_1 > 0$ (or $< 0$) can hence be combined into one expression
for $f^{(\rho)}(t;x_0,0)$ and $f_\U^{(\rho)}(t;x_0,0)$ where $\tau(t) = (e^{\vert\lambda_1\vert t}-1)^{-1}$ and
$a = \rho/\vert\lambda_1\vert + \tilde{\mu}$ with $\tilde{\mu} = \mu_-$ for $\lambda_1 > 0$ and $\tilde{\mu} = 1 + \mu_+$ for $\lambda_1 < 0$.

\subsection{The Ornstein-Uhlenbeck Process}\label{subsect2.3}
Consider the Ornstein-Uhlenbeck (OU) process with SDE
$dX_t=(\lambda_0 -\lambda_1X_t)dt+\nu_0dW_t,$ where
$\lambda_0,\lambda_1,\nu_0 > 0$. Both boundaries, $l=-\infty$ and $r=\infty$,
of the state space $\I=(-\infty,\infty)$ are non-attracting natural.
Without loss in generality, we set
$\lambda_0=0$. Otherwise we can consider the shifted
process $Y_t=X_t-\frac{\lambda_0}{\lambda_1}$ and the formulas
follow by simply shifting $x\to x -\frac{\lambda_0}{\lambda_1}$,
$x_0\to x_0 -\frac{\lambda_0}{\lambda_1}$. We define the positive constant
$\kappa\equiv 2\lambda_1 / \nu_0^2$ so that the generator for the diffusion takes the form
$(\G\, f)(x)\triangleq {1\over 2}\nu_0^2[f''(x) - \kappa x f'(x)]$.
The speed and scale densities are $\s(x)=e^{\kappa x^2/2}$ and
$\m(x)=(2/\nu_0^2)e^{-\kappa x^2/2}.$ We note that this corresponds to an OU process indexed by two positive
parameters. For brevity, we omit the case where $\lambda_1 < 0$ (i.e. $\kappa < 0$) as the analytical treatment follows very similarly
(e.g. see page 137 of \cite{BS02}). The case where $\lambda_1 = 0$ simply corresponds to Brownian motion.

A pair of fundamental solutions to $(\G\, \varphi)(x) = s \varphi(x)$, such that
for real values of $s=\rho>0$ are respectively increasing and decreasing positive
functions for $x\in\R_+$, are
\begin{equation} \label{OUfund}
\varphi^+_s(x) = e^{\kappa x^2\!/4}\, D_{-s/\lambda_1}(-\sqrt{\kappa}\,x)
\quad\mbox{and}\quad
\varphi^-_s(x)= e^{\kappa x^2\!/4}\,D_{-s/\lambda_1}(\sqrt{\kappa}\,x)
\end{equation}
where $D_{-\upsilon}(x)$ is Whittaker's parabolic cylinder function (see \cite{AS72} for definitions and properties).
Note the symmetry $\varphi^+_s(x) = \varphi^-_s(-x)$
The Wronskian constant in equation (\ref{wronskian}) is
$w_s=\frac{\sqrt{2\kappa\pi}}{\Gamma(s/\lambda_1)}.$
For $s_1,s_2\in\C$, $W[\varphi^\pm_{s_1},\varphi^\pm_{s_2}](x)$ and $W[\varphi^-_{s_1},\varphi^+_{s_2}](x)$
are obtained using differential recurrences ${d\over dz}D_{-\upsilon}(z) = -(z/2)D_{-\upsilon}(z) - \upsilon D_{-\upsilon - 1}(z)$.

The boundaries $\pm\infty$ are natural and hence the Green function in (\ref{greenfunc})
for $x,x_0\in \R$ is uniquely given by taking $\psi_s(x) = \varphi^+_s(x)$ and $\phi_s(x) = \varphi^-_s(x)$ where
${\cal W}_s = w_s$. The well-known (Gaussian) transition PDF on $\R$ follows by
Laplace inverting the Green function with the use of the identity
${\mathcal L}_s^{-1}[\Gamma(s)D_{-s}(x)D_{-s}(y)][t] =
\frac{e^{t/2}}{\sqrt{2\sinh t}}\exp\bigg(-{(x^2 + y^2)\cosh t + 2xy \over 4 \sinh t} \bigg)$, for $x\le y$, $t>0$, giving
\begin{equation} \label{OUden}
 p_X(t;x_0,x) = \sqrt{\frac{\kappa}{2\pi(1-e^{-2\lambda_1 t})}}
 \exp\left(-\frac{\kappa(x-x_0 e^{-\lambda_1 t})^2}{2(1-e^{-2\lambda_1
 t})}\right).
\end{equation}

The OU family of $X^{(\rho)}$-diffusions has generator in (\ref{Generator_rho}) with
$\varphi^\pm_\rho(x)$ given by (\ref{OUfund}) for real $s=\rho > 0$ and the
following lemma gives the boundary classification for these processes.

\begin{lemma} \label{lemma_class_OU}
The OU family of regular diffusions $X_t^{(\rho)}$, have the following boundary classification. The endpoint $l=-\infty$ is non-attracting natural if $q_2=0$ and is attracting natural if $q_2>0$. The endpoint $r=\infty$ is non-attracting natural if $q_1=0$ and is attracting natural if $q_1>0$.
\end{lemma}
\begin{proof}
The results follow from Lemma~\ref{Lemma_Xrho} and the asymptotic relations
in (i) of Appendix~\ref{subsect_a13} that give
$n(-\infty;\rho,\rho)=0$, $n(\infty;\rho,\rho)=\infty$, $(\varphi_\rho^+,\varphi_\rho^-)_{(-\infty,x]} = \infty$,
$(\varphi_\rho^+,\varphi_\rho^-)_{[x,\infty)} = \infty$.
\end{proof}

The OU subfamilies (i)--(iii) of $X^{(\rho)}$-processes on $\R$ are all conservative with both endpoints as natural.
The unique Green function in the form of (\ref{greenfunc_rho}) is given by
$\psi^{(\rho)}_s = \varphi_{s+\rho}^+/\hat{u}_\rho$, $\phi^{(\rho)}_s = \varphi_{s+\rho}^-/\hat{u}_\rho$,
${\cal W}^{(\rho)}_s = w_{s+\rho}$. Hence, Laplace inverting gives the transition PDF in the form of (\ref{prho})
with $p_X(t;x_0,x)$ given by (\ref{OUden}):
\begin{align} \label{OUden_rho}
 p^{(\rho)}_X(t;x_0,x) = {e^{{\kappa \over 4}x^2}\big[q_1 D_{-\rho/\lambda_1}(-\sqrt{\kappa} x) + q_2 D_{-\rho/\lambda_1}(\sqrt{\kappa} x)\big]
 \over e^{{\kappa \over 4}x_0^2}\big[q_1 D_{-\rho/\lambda_1}(-\sqrt{\kappa} x_0) + q_2 D_{-\rho/\lambda_1}(\sqrt{\kappa} x_0)\big] }
e^{-\rho t}p_X(t;x_0,x),
\end{align}
$x,x_0\in\R$, $t > 0$.

\section{$F$-Diffusions with Linear Drift}\label{sect3}

\subsection{Construction of the Mapping}\label{subsect3.1}

We now consider $F$-diffusions $\{F_t \triangleq \F
(X^{(\rho)}_t), t\ge 0\}$ having infinitesimal generator (\ref{GeneratorF}).
For driftless diffusions, see \cite{ACCL,CM06,CM07}. The diffusion coefficient function $\sigma(F)$, as given by (\ref{sigmaF})
below, is generally nonlinear and where $a$ and $b$ are arbitrary real constants {\it such that $b=0$ implies $a=0$}.

The transition PDF $p_F$ for an $F$-diffusion $(F_t)_{t\ge 0}$ is related to
the transition PDF for the underlying $X$ (or $X^{(\rho)}$) diffusion as follows:
\begin{equation}\label{PrKernelF}
  p_F(t;F_0,F) =
\displaystyle\frac{\nu(\X (F))}{\sigma(F)}p_X^{(\rho)}(t;\X (F_0),\X (F)) =
  \frac{\nu(\X (F))}{\sigma(F)}
  \frac{\hat u_\rho\left(\X (F) \right)}{\hat u_\rho\left(\X (F_0) \right)}
  e^{-\rho t} p_X(t;\X (F_0),\X (F))
  \,.
\end{equation}
where $F,F_0\in \I_F$, $t>0$. Here $\X \triangleq\F^{-1}$ is the inverse
map so that $\left|\X'(F)\right|=\displaystyle\frac{\nu(\X (F))}{\sigma(F)}.$

This methodology was originally developed for \textit{driftless}
$F$-diffusions where $a=b=0$. For such cases, the volatility function has the form $\sigma(F) = {\sigma_0\nu(x)\s(x) \over
  \hat{u}^2_\rho(x)},$ $x=\X (F),$ $\sigma_0>0.$
The map~$\F (x)$ (that solves equation (\ref{FODE}) for the special case $a=b=0$) admits the general quotient form:
 \begin{equation} \label{FMAPdl}
    \F (x) =
   \ds\frac{c_1\varphi^+_\rho(x)+c_2\varphi^-_\rho(x)}{q_1\varphi^+_\rho(x)+q_2\varphi^-_\rho(x)}\,,
 \end{equation}
where $c_1,c_2,q_1,q_2\in\R$ are parameters such that
$q_1c_2-q_2c_1\neq 0\,.$
%Several families of $F$-diffusions arising from various choices of underlying diffusions (such as the squared Bessel, Ornstein-Uhlenbeck, CIR and Jacobi processes) are studied in \cite{ACCL,CM06, Campolieti2008}. In this paper, we expand the families constructed by finding new mapping functions $\F$ so that the newly obtained diffusions have generator in (\ref{GeneratorF}), i.e. they satisfy an SDE with affine drift: $dF_t=(a+bF_t)dt+\sigma(F_t)d W_t.$

\begin{lemma} \label{lemma1}
 Let $b\neq 0$ and $\rho, \rho+b>0$ hold. Then the solution to equation~(\ref{FODE}) takes the general form
 \begin{equation} \label{FMAP} \F (x) = -\frac{a}{b} + \frac{c_1 \varphi^+_{\rho+b}(x)+c_2
 \varphi^-_{\rho+b}(x)}{q_1 \varphi^+_{\rho}(x)+q_2
 \varphi^-_{\rho}(x)}\equiv -\frac{a}{b}+\frac{\hat{v}_{\rho+b}(x)}{\hat{u}_\rho(x)}
 \end{equation}
where $c_1$ and $c_2$ are arbitrary real constants.
\end{lemma}
\begin{proof}
The numerator $\hat{v}=\hat{v}_{\rho+b}(x)$, defined in (\ref{FMAP}), is a linear combination of $\varphi^\pm_{\rho+b}$ and hence
solves $\G\hat{v}  = (\rho + b) \hat{v}$. The denominator $\hat{u}=\hat{u}_\rho(x)$ solves $\G\hat{u} = \rho\hat{u}$.
Differentiating and using the identity $\G^{(\rho)}f(x) = \G f(x) + \nu^2(x)(\hat{u}'_\rho(x)/\hat{u}_\rho(x))f^{\prime}(x)$ readily gives
\[\G^{(\rho)}\frac{\hat{v}}{\hat{u}} = \frac{1}{\hat{u}}\left(\G\hat{v} - \frac{\hat{v}}{\hat{u}}\G\hat{u} \right)
= \frac{1}{\hat{u}}\left((\rho + b) \hat{v} -
\frac{\hat{v}}{\hat{u}}(\rho\hat{u}) \right) =
b\frac{\hat{v}}{\hat{u}}\,.\]
Hence, $\hat{v}/\hat{u}$ is a general solution to the corresponding homogeneous ODE (eq.~(\ref{FODE}) for $a=0$)
since the Wronskian
$W\left[\frac{\varphi^-_{\rho+b}}{\hat{u}_\rho}, \frac{\varphi^+_{\rho+b}}{\hat{u}_\rho}\right](x) =
\frac{W[\varphi^-_{\rho+b}, \varphi^+_{\rho+b}](x)}{\hat{u}^2_\rho(x)}\ne 0$ from (\ref{wronskian}). The constant function
$\F_p(x) =-a/b$ is a particular solution of (\ref{FODE}).
\end{proof}

The derivative of the mapping in (\ref{FMAP}) is simply
 \begin{equation} \label{Fder} \F '(x) = \frac{\hat{v}'_{\rho+b}(x)\hat{u}_\rho(x) -
   \hat{u}'_\rho(x)\hat{v}_{\rho+b}(x)}{\hat{u}^2_\rho(x)}
    =\frac{W(x)}{\hat{u}^2_\rho(x)}\,,
 \end{equation}
where we define the Wronskian
 \begin{equation} \label{Wfunc}
   W(x) \equiv W(x;\rho,\rho+b)\triangleq W[\hat{u}_\rho,\hat{v}_{\rho+b}](x)\,.
 \end{equation}
Assuming $\F$ is strictly monotonic, and given an $X$-diffusion, the $F$-diffusion coefficient function $\sigma(F)$
is then given by substituting (\ref{Fder}) into $\sigma(F) = \nu(x)|\F'(x)|$, giving
\begin{equation} \label{sigmaF}
 \sigma(F)=\frac{\nu(x)|W(x)|}{\hat{u}^2_\rho(x)}\,,\quad x=\X (F)\,,\quad F\in\I_{\sf F}\,.
\end{equation}
We note that this expression holds for all parameter choices $a,b$ except when $a\ne 0$ and $b=0$.
For the latter special case (i.e. constant nonzero drift function) equation (\ref{FODE})
reads $\G^{(\rho)}\,\F (x) = a$ and hence simply reduces to a linear first order ODE in $\F'$.
Solving leads to various monotonic maps which in turn give rise to nonzero
constant drift $F$-diffusions with various nonlinear specifications for the diffusion coefficients. In this paper,
we shall not discuss the details of such special families as we focus on linear drift functions with $b\ne 0$.

\subsection{Monotonic Maps}\label{subsect3.2}

The map $\F:\I\to\I_{\F}$ in (\ref{FMAP}) does not generally satisfy $\F '(x)
\ne 0$. To guarantee that $(F_t)_{t\ge 0}$ is a regular diffusion
on $\I_{\sf F}=(F^l,F^r)$ the map
$\F$ has to be strictly monotonic. Then $\F '(x) \ne 0$ and hence the diffusion
coefficient function $\sigma(F)$ is strictly positive on $\I_F$.

From (\ref{Fder}) we observe that $\sgn (\F '(x)) = \sgn (W(x))$.
Using the representations of $\hat{u}_\rho$ and~$\hat{v}_{\rho + b}$ in terms of
$\varphi^\pm_\rho$, $\varphi^\pm_{\rho + b}$ gives
\begin{equation} \label{Wfuncpm}
 \begin{array}{rcl}
  W(x) &=&
  q_1c_1\,W[\varphi^+_\rho,\varphi^+_{\rho+b}](x)+q_1c_2\,W[\varphi^+_\rho,\varphi^-_{\rho+b}](x)\\
    && \mbox{} +q_2c_1\,W[\varphi^-_\rho,\varphi^+_{\rho+b}](x)+q_2c_2\,W[\varphi^-_\rho,\varphi^-_{\rho+b}](x)\,.
  \end{array}
\end{equation}
There are two important cases where $\F $ is strictly monotonic. Recall that the fundamental solutions
$\varphi^+_\rho(x)$ and $\varphi^-_\rho(x)$ are correspondingly
strictly increasing and decreasing functions of $x$. Therefore, the ratios
 $ \frac{\varphi^+_{\rho+b}(x)}{\varphi^-_\rho(x)} \mbox{ and }
    \frac{\varphi^-_{\rho+b}(x)}{\varphi^+_\rho(x)} $
are strictly increasing and decreasing functions, respectively. In particular, we have the strict inequalities
$W[\varphi^+_\rho,\varphi^-_{\rho+b}](x)<0$ and $W[\varphi^-_\rho,\varphi^+_{\rho+b}](x)>0$.
Thus, the two choices of parameters $c_2=q_1=0,$ $c_1/q_2=\pm c$ or $c_1=q_2=0,$
$c_2/q_1=\pm c$ in equation~(\ref{FMAP}) lead to {\it dual subfamilies} of strictly
monotonic maps defined by $\F=\F ^{(1)}_\pm$:
 \begin{equation} \label{FMAPpm1} \F ^{(1)}_\pm(x) \triangleq  -\frac{a}{b} + \epsilon\,c\frac{\varphi^\pm_{\rho+b}(x)}{\varphi^\mp_{\rho}(x)}\,,
 \end{equation}
where $\epsilon=\pm 1$ and $c>0$ is constant. Other parameter choices that lead to other
families of monotonic maps are discussed in Section~\ref{subsect4.2}.
The following propositions are useful in verifying whether or not $W(x)$, and hence $\F'(x)$, changes sign in $\I$.
\begin{proposition} \label{prop1}
 $W(x)$ in (\ref{Wfunc}) satisfies
 $\frac{1}{2}\nu^2(x) W'(x) + \lambda(x) W(x)  = b \hat{u}_\rho\hat{v}_{\rho + b}$, and for any $x,x_0\in \I$ the solution admits the
 following representation:
 \begin{equation}\label{wronskian_solution} \frac{W(x)}{\s(x)} = \frac{W(x_0)}{\s(x_0)}
    + b\int\limits_{x_0}^x \m(y) \hat{u}_\rho(y)
    \hat{v}_{\rho+b}(y)\,dy\,.
 \end{equation}
\end{proposition}
\begin{proof}
The proof follows by direct verification upon using $\G\hat{u}_\rho = \rho\hat{u}_\rho$ and $\G\hat{v}_{\rho + b} = (\rho + b)\hat{v}_{\rho + b}$.
\end{proof}

\noindent Note that for the driftless case, with $a=b=0$, the function $\frac{W(x)}{\s(x)}$ is constant and from (\ref{wronskian}) and (\ref{Wfuncpm}):
$W(x) = (c_1q_2 - c_2q_1) w_\rho\s(x)$. Therefore, the equation defining $\sigma(F)$ for all families of driftless $F$-diffusions
is recovered as a particular case of the more general
specification given by equation~(\ref{sigmaF}). Moreover, the specification in (\ref{sigmaF}) gives rise to state
dependent volatility functions that can also have a dependence on the drift parameters $a$ and $b$.

\begin{proposition} \label{prop2}
Assume that $c_1$ and $c_2$ in (\ref{FMAP}) are both nonnegative or nonpositive (with at least one of them being nonzero)and that $W(x)$ in (\ref{Wfunc}) preserves its sign as $x$
approaches either endpoint $l$ or $r$; that is, $\sgn \big(W(l+)\big) = \sgn \big(W(r-)\big)$. Then, $W(x)\ne 0$ for all $x\in\I$, i.e. $W(x)$
is either strictly positive or negative on $\I$.
\end{proposition}
\begin{proof}
Under the assumed conditions on $c_1$ and $c_2$, the
function $\hat{v}_{\rho+b}$ is either strictly positive or strictly
negative on $\I$. Recall that $\hat{u}_\rho$ is strictly positive. Hence, the function $W(x)/\s(x)$
given by (\ref{wronskian_solution}) is monotonic in $x$ and, since $\s(x) > 0$, $W(x)$ has at most one zero
in $\I$. Then $\sgn \big(W(l+)\big) = \sgn \big(W(r-)\big)$ implies $W(x)\ne 0$, i.e. either $W(x) > 0$ or $W(x) < 0$
for all $x\in\I$.
\end{proof}

\begin{proposition} \label{prop3}
Assume that $c_1$ and $c_2$ for the map $\F$ in (\ref{FMAP})
are both nonzero and have opposite signs.
  \begin{enumerate}[(i)]
    \item If $b>0$, then $W(x)$ is either strictly positive or negative on $\I$.
    \item Let $b<0$ and $x\in\I$. If $\sgn(W(l+)) = \sgn(W(r-)) = +1$ and $c_1>0$, then $W(x) > 0$.
If $\sgn(W(l+)) = \sgn(W(r-)) =-1$ and $c_1<0$, then $W(x) < 0$.
  \end{enumerate}
 \end{proposition}
\begin{proof}
By definition of $\hat{v}_{\rho+b}$, and the fact that $\varphi^+_{\rho+b}$ and $\varphi^-_{\rho+b}$
are respectively increasing and decreasing functions, we have $\hat{v}_{\rho+b}'(x) > 0$ (or $\hat{v}_{\rho+b}'(x) < 0$)
when $c_1>0$, $c_2<0$ (or $c_1<0$, $c_2>0$), i.e. $\hat{v}_{\rho+b}(x)$ is either a strictly increasing (or decreasing)
function. Moreover, it follows from such monotonicity and the boundary conditions of $n(x;\rho+b,\rho)$ that $\hat{v}_{\rho+b}(x)$
has exactly one zero, at $x=\hat{x}_0\in \I$ where $\varphi^+_{\rho+b}(\hat{x}_0)/\varphi^-_{\rho+b}(\hat{x}_0)
= \vert c_2/c_1\vert$. Then, $W(\hat{x}_0)=\hat{u}_\rho(\hat{x}_0)\hat{v}_{\rho+b}'(\hat{x}_0)$
is accordingly strictly positive (or negative). Setting $x_0=\hat{x}_0$ in (\ref{wronskian_solution}) now gives
$$ \frac{W(x)}{\s(x)} = \frac{W(\hat{x}_0)}{\s(\hat{x}_0)}
    + \epsilon \,b\int_{\min(\hat{x}_0,x)}^{\max(\hat{x}_0,x)} \m(y) \hat{u}_\rho(y)
    \vert \hat{v}_{\rho+b}(y)\vert\,dy
$$
where $\epsilon = +1 (-1)$ if $c_1>0$ ($c_1<0$). Hence, if $b>0$ then either $W(x) > 0$ or $W(x) < 0$, for all $x\in\I$,
in the respective cases. If $b<0$, then $W(x)$, as given in the last expression, can either have no zeros or at most two zeroes in $\I.$
$W(x)$ has no zeros if and only if $\sgn (W(\hat{x}_0))= \sgn (W(l+)) = \sgn (W(r-))$. This hence proves part (ii) for the
respective cases.
\end{proof}

From the above propositions, the monotonicity of a map $\F$ is determined simply by
examining the asymptotic behaviour of $W(x)$, as $x$ approaches endpoint $l$ or $r$.
Below we consider three families of $F$-diffusions arising from
the three separate underlying diffusions discussed in Sections \ref{subsect2.1} - \ref{subsect2.3};
namely, the Bessel, the confluent hypergeometric and the OU families of process. For all these families, the
asymptotic properties of the fundamental solutions $\varphi^\pm_s$ and
of the corresponding Wronskian functions, for the respective underlying $X$-diffusions, are presented
in Appendix~\ref{sect_a1} .

\subsection{Martingale Property} \label{subsect3.3}

For any time-homogeneous $F$-diffusion defined by (\ref{GeneratorF}),
we introduce the rate of change of the conditional expectation:
 \begin{equation} \label{mathexp_rate1}
  \frac{\partial}{\partial
  t}\E[F_{\tau+t}\mid F_\tau=Y]=\int_{F^l}^{F^r} F\,\frac{\partial p_F}{\partial
  t}(t;Y,F)\,dF\,,\quad
  Y\in\I_{\F},\;t>0,\tau\ge 0\,.
 \end{equation}
The transition PDF $p_F$ given by (\ref{PrKernelF}) satisfies the
forward Kolmogorov equation
\begin{equation}\label{forward_KPDE}
\frac{\partial p_F}{\partial t}=\frac{\partial}{\partial F}\left(\frac{1}{\s_{\sf F}}\frac{\partial}
{\partial F}\left(\frac{p_F}{\m_{\sf F}}\right)\right),
\end{equation}
with scale and speed densities given in terms of those for the $X^{(\rho)}$-diffusion:
 \begin{equation} \label{speed_scale_FandX}
\m_{\sf F}(F)=\m_\rho(\X (F))\,\left|\X'(F)\right|
   \quad\mbox{and}\quad
   \s_{\sf F}(F)=\s_\rho(\X (F))\,\left|\X'(F)\right|\,.
 \end{equation}
Here we consider $F$-diffusions with affine drift $\alpha(F) = a + bF$.
Then, using (\ref{forward_KPDE}) within (\ref{mathexp_rate1}), applying integration by parts
twice and making use of the derivative ${d\over dF}\left({1\over \s_{\sf F}(F)}\right) = (a+bF)\m_{\sf F}(F)$,
gives the rate in (\ref{mathexp_rate1}) expressed as a sum:
\begin{equation} \label{mathexp_conserve}
  \frac{\partial}{\partial
  t}\E[F_{\tau+t}\mid F_\tau=Y]=\int_{F^l}^{F^r} (a+bF)\,p_F(t;Y,F)\,dF+
  \mathcal{E}(Y,t)\,.
\end{equation}
The ``bias'' term $\mathcal{E}(Y,t)$ is given by the difference of two limits:
\begin{equation} \label{mebias}
  \begin{array}{rcl}
  \mathcal{E}(Y,t) &=&
   \left[ \ds\frac{F}{\s_{\sf F}(F)}\, \frac{\partial}{\partial F}
          \left(\frac{p_F(t;Y,F)}{\m_{\sf F}(F)}\right)
          -
          \frac{1}{\s_{\sf F}(F)}\,\frac{p_F(t;Y,F)}{\m_{\sf F}(F)}\right]_{F=F^l}^{F=F^r}\\[15pt]
  &=& \left[ \ds\frac{\F (x)}{\s_\rho(x)} \, \frac{\partial}{\partial x}
          \left(\frac{p_X^{(\rho)}(t;y,x)}{\m_\rho(x)}\right)
          -
          \frac{\F'(x)}{\s_\rho(x)}\,\frac{p_X^{(\rho)}(t;y,x)}{\m_\rho(x)}\right]_{x=l+}^{x=r-}\,.
 \end{array}
\end{equation}
The last expression follows by changing variables $x=\X(F), y = \X(Y)$ ($\X \equiv \F^{-1}$) and
by combining (\ref{speed_scale_FandX}) and (\ref{PrKernelF}).

Consider the case where $\mathcal{E}(Y,t)\equiv 0$, for any $F_\tau=Y\in\I_{\F},t>0$.
Then, for such diffusions one obtains a simple representation for the rate of change in (\ref{mathexp_conserve}) as
\begin{equation} \label{mathexp_rate2}
 \frac{\partial}{\partial t}\E[F_{\tau+t}\mid F_\tau]
 = a\E[\indfunc_{F_{\tau+t}\in\I_{\F}} \mid F_\tau]
 + b\E[F_{\tau+t}\mid F_\tau]\,,
\end{equation}
i.e. the rate of change of the conditional expectation equals the conditional expectation of the drift function for the process.

If $(F_t)_{t\ge 0}$ conserves probability, i.e. $\E[\indfunc_{F_{\tau+t}\in\I_\F} \mid F_\tau]\equiv\P\{F_{\tau+t}\in\I_{\F} | F_\tau\in\I_\F\} = 1$ for all $\tau\ge 0,t>0$, or if $a=0$ holds, then $E(t)\triangleq\E[F_{\tau+t}\mid F_\tau = Y]$ satisfies the trivial linear ODE $E'(t)=a + b\,E(t)$, subject to $E(0)=Y$. Let $a=0$, then there exists a well-known
\textit{geometric drift} solution $E(t) = Ye^{bt}$, i.e. $\E[F_{\tau+t}\mid
F_\tau=Y]=e^{bt} Y\,\,\,\mbox{ for all }Y\in\I_\F,\tau\ge 0,t>0.$
In other words, the ``discounted
process'' $(e^{-bt} F_t)_{t\ge 0}$ is a martingale in case $a=0$ and $\mathcal{E}\equiv 0$. Note also that
the discounted process is a strict supermartingale (submartingale) when $a=0$ and $\mathcal{E} < 0$ ($\mathcal{E} > 0$) for all $t>0$.
Setting $\tau=0$ recovers the unconditional expectation $\E[F_t] = F_0 e^{bt}$.

The following theorem now gives necessary and sufficient conditions for the validity of the relation (\ref{mathexp_rate2}). These conditions
involve limit expressions that can be readily evaluated from the boundary asymptotic properties of the fundamental solutions $\varphi^\pm$ for the underlying
$X$-diffusion. It is implied that the diffusion process $(X_t^{(\rho)})_{t\ge 0}$ has the generator in
(\ref{Generator_rho}) with boundary conditions specified via (\ref{greenfunc_rho}). Moreover, we assume that $\E[\vert F_t\vert] < \infty$, i.e.
$\int_l^r \vert\F(x)\vert p_X^{(\rho)}(t;x_0,x)dx < \infty$, $x_0\in\I, t > 0$.
The map $\F$ is assumed to be monotonic and given by (\ref{FMAP}) for $b\ne 0$ and by (\ref{FMAPdl}) for $a=b=0$.
We note that a similar result and proof for special cases of driftless $F$-diffusions is given in \cite{CM06}.

\begin{theorem} \label{theorem1}
The diffusion $F_t=\F(X_t^{(\rho)}), t\ge 0$, conserves the expectation
rate, i.e. the relation (\ref{mathexp_rate2}) is true if and only if the
following boundary conditions hold:
\begin{eqnarray}
\label{cond_expectation}
     \lim\limits_{x\to l+} {W[\F , \psi^{(\rho)}_{s}](x) \over \s_\rho(x)} = 0 \,\,
     \text{  and  }\,\,
     \lim\limits_{x\to r-} {W[\F , \phi^{(\rho)}_{s}](x) \over \s_\rho(x)} = 0.
%    \label{cond_expXl}
%     \lim\limits_{x\to l+}\bigg[
%     \F (x)\,\frac{W[\hat{u}_\rho,\hat\psi_{\rho+s}](x)}{\s(x)}
%     -
%     \frac{W[\hat{u}_\rho,\hat{v}_{\rho+b}](x)}{\s(x)}\,\frac{\hat\psi_{\rho+s}(x)}{\hat{u}_\rho(x)}\bigg]=0\,,
%     \\[8pt]
%    \label{cond_expXr}
%     \lim\limits_{x\to r-}\bigg[
%     \F (x)\,\frac{W[\hat{u}_\rho,\hat\phi_{\rho+s}](x)}{\s(x)}
%     - \frac{W[\hat{u}_\rho,\hat{v}_{\rho+b}](x)}{\s(x)}\,\frac{\hat\phi_{\rho+s}(x)}{\hat{u}_\rho(x)}\bigg]=0\,,
\end{eqnarray}
for all complex-valued $s$ such that $\mathrm{Re\,}s>c$, for some real constant $c$.
%, where
%$\hat\psi_{s+\rho}(x) \equiv \hat{A}_1 \varphi^+_{s + \rho}(x) + \hat{B}_1 \varphi^-_{s + \rho}(x)$ and
%$\hat\phi_{s+\rho}(x) \equiv \hat{A}_2 \varphi^+_{s + \rho}(x) + \hat{B}_2 \varphi^-_{s + \rho}(x)$.
\end{theorem}
\begin{proof}
Assume $\mathrm{Re\,}s>c$ with sufficiently large $c$ such that $G_X^{(\rho)}(x,y,s)$ is analytic in $s$.
The condition $\mathcal{E}(Y,t)\equiv 0$ is then equivalent to the Laplace transform condition
${\mathcal L}_t[\mathcal{E}(Y,t)][s] \equiv 0$.
Laplace transforming the second expression within the limits in (\ref{mebias}), while changing
the order of ${\mathcal L}_t$ and differentiation, using (\ref{greenfunc_rho}) and (\ref{Laplace_inversion}),
and finally combining both left and right limits gives the condition ${\mathcal L}_t[\mathcal{E}(Y,t)][s] \equiv 0$ in the form:
\begin{eqnarray}\label{lap_mart_cond}
\psi^{(\rho)}_{s}(y) \cdot {W[\F , \phi^{(\rho)}_{s}](x) \over \s_\rho(x)}\bigg\vert_{x=r-}
- \phi^{(\rho)}_{s}(y) \cdot {W[\F , \psi^{(\rho)}_{s}](x) \over \s_\rho(x)}\bigg\vert_{x=l+} \equiv 0.
\end{eqnarray}
For (\ref{lap_mart_cond}) to hold true, for all $y\in\I$, the two limits must vanish
since $\{\psi^{(\rho)}_{s},\phi^{(\rho)}_{s}\}$ is a linearly independent pair on $\I$.
\end{proof}
\noindent {\it Remark} : We recall from (\ref{greenfunc_rho})
that the fundamental solutions $\{\psi^{(\rho)}_{s}, \phi^{(\rho)}_{s}\}$ for $X^{(\rho)}$ have the form
$\psi^{(\rho)}_{s} = \hat\psi_{s + \rho}/\hat{u}_\rho$  and
$\phi^{(\rho)}_{s} = \hat\phi_{s + \rho}/\hat{u}_\rho$, where we conveniently define
$\hat\psi_{s + \rho}\triangleq  \hat{A}_1\varphi^+_{\rho + s} + \hat{B}_1\varphi^-_{\rho + s}$ and
$\hat\phi_{s + \rho}\triangleq  \hat{A}_2\varphi^+_{\rho + s} + \hat{B}_2\varphi^-_{\rho + s}$. Hence, if
the left boundary is singular (exit, entrance or natural) or regular killing then $\hat\psi_{s + \rho} = \varphi^+_{s + \rho}$.
Similarly, for a singular or regular killing right boundary we have $\hat\phi_{s + \rho} = \varphi^-_{s + \rho}$.

An important general class of $F$-diffusions follows by setting $a=0$, i.e. with
drift function is now $\alpha(F) = bF$. Such diffusions are useful, for example, for modelling
asset prices and for equity option pricing in finance. The corollary below gives necessary and
sufficient conditions for the discounted process $(e^{-bt} F_t)_{t\ge 0}$ to be
a martingale when $a=0$. As in Theorem \ref{theorem1}, we assume $\F$ is monotonic and that $\E[\vert F_t\vert] < \infty$.

\begin{corollary} \label{corollary1}
Let $\F = \hat{v}_{\rho + b}/\hat{u}_\rho$, i.e. the map in (\ref{FMAP}) with $a=0,b\ne 0$ and
by (\ref{FMAPdl}) for $a=b=0$. Then, the discounted process $(e^{-bt} F_t)_{t\ge 0}$ is a martingale if and only if
\begin{eqnarray}
\label{cond_martingale}
     \lim\limits_{x\to l+} {W[\hat{v}_{\rho + b}, \hat{\psi}_{\rho + s}](x) \over \s(x)} = 0 \,\,
     \text{  and  }\,\, \lim\limits_{x\to r-} {W[\hat{v}_{\rho + b}, \hat{\phi}_{\rho + s}](x) \over \s(x)} = 0,
\end{eqnarray}
for all complex-valued $s$ such that $\mathrm{Re\,}s>c$, for some real constant $c$.
\end{corollary}
\begin{proof} Setting $a=0$ in Theorem \ref{theorem1}, we have that the discounted process is a martingale, i.e. conserves the expectation rate,
if and only if conditions in (\ref{cond_expectation}) hold.
Then, using $\F(x) = {\hat{v}_{\rho + b}(x)\over \hat{u}_\rho(x)}$,
$\psi^{(\rho)}_{s}(x) = {\hat\psi_{s + \rho}(x) \over \hat{u}_\rho(x)}$,
$\phi^{(\rho)}_{s}(x) = {\hat\phi_{s + \rho}(x) \over \hat{u}_\rho(x)}$ and
$\s(x) = \hat{u}_\rho^2(x)\s_\rho(x)$ reduces the Wronskian conditions to (\ref{cond_martingale}).
\end{proof}

\section{The Bessel, Confluent Hypergeometric, and Ornstein-Uhlen\-beck Families of Affine-Drift $F$-Diffusions} \label{sect4}

\subsection{Three Main Families of Affine Drift Diffusions: Classification and Properties} \label{subsect4.1}
Using the general construction presented in Section~\ref{sect1} and the
three underlying $X$-diffusions from Section~\ref{sect2}, we can
construct three new families of $F$-diffusions with affine drift as defined by~(\ref{GeneratorF}) and (\ref{sigmaF}). For all such new families, the
transition PDFs are given in analytically closed form. Generally, a transition PDF $p_F$ is
given by (\ref{PrKernelF}). It is trivially related to the corresponding transition PDF of the $X^{(\rho)}$-diffusion, where the monotonic map
$\F(x)$ (with its inverse map $\X(F)$) is given in (\ref{FMAP}). In turn, the transition PDF for an $X^{(\rho)}$-diffusion is expressible in terms of a transition PDF for the underlying $X$ via the generating function $\hat{u}_\rho$ in (\ref{uhat}).
By choosing the SQB, CIR, or OU diffusions as underlying
$X$-diffusion, we respectively obtain the so-called {\it Bessel, confluent hypergeometric, or Ornstein-Uhlenbeck (OU) families of $F$-diffusions with affine drift}. Some properties and classification of the corresponding driftless diffusions are discussed in \cite{CM06}.
In the respective sections \ref{subsect2.1}--\ref{subsect2.3}, we have constructed analytically exact transition PDFs $p_X^{(\rho)}$
on the respective regular state spaces $\I=(0,\infty)$ or $\I=(-\infty,\infty)$ for the corresponding three main
families of $X^{(\rho)}$-diffusions for subfamilies of type (i) $q_1=0, q_2 > 0$, (ii) $q_1>0, q_2 = 0$  and (iii) $q_1>0, q_2 > 0$.
%For corresponding diffusions with imposed killing at one or two interior points of the
%state space, the transition densities are given in terms of closed-form spectral expansions and these, in turn, lead to
%closed-form spectral expansions for transition PDFs $p_F$, as well as other fundamental quantities such as first hitting (passage) time
%densities for $F$-diffusions (see \cite{Campolieti2008}).
The fundamental elementary solutions $\varphi^\pm$ used in generating these three main families of transformed processes are
given by either (\ref{SQBfund}), (\ref{CIRfund_1_M})--(\ref{CIRfund_1_U}) or (\ref{CIRfund_2_M})--(\ref{CIRfund_2_U}), or (\ref{OUfund}).
The boundary classification for the three respective main families of $X^{(\rho)}$ processes is given in
Lemmas \ref{lemma_class_Bessel}, \ref{lemma_class_CIR} and \ref{lemma_class_OU}. Hence,
the boundary classification for the Bessel, confluent hypergeometric and OU families of $F$-diffusions follows immediately via Lemma \ref{Theorem_F_classification}.

Each $F$-diffusion is described by the set of
parameters that are inherited from the chosen underlying
diffusion. In addition to that set, the nonnegative parameters
$q_1,$ $q_2$, and $\rho>0$ are added due to the
measure change $X\to X^{(\rho)}$. Two parameters $a$ and $b$ describe the affine drift
coefficient in (\ref{GeneratorF}). Finally, up to two other parameters
$c_1$ and $c_2$ are used in the map function.
As observed from the diffusion coefficient function $\sigma(F)$ in (\ref{sigmaF}), the combination of all such parameters
make the new diffusions quite flexible for modelling various stochastic processes.
Different choices of monotonic maps $\F=\F^{(i)}_\pm$, $i=1,\ldots,5$, lead to different $F$-diffusions. In any case, the diffusion function
is specified by (\ref{sigmaF}).

The choice $\F=\F^{(1)}_\pm$ leads to $F$-diffusions that have applications
in finance. In particular, the dual maps defined by (\ref{FMAPpm1}) with $\epsilon = +1$ give rise to sets of {\it dual subfamilies} (i) and (ii) of
affine drift $F$-diffusions with respective volatility specification:
\begin{equation}
\sigma(F) = c \nu(x) \left\{\begin{array}{ll}
{W[\varphi^-_\rho,\varphi^+_{\rho+ b}](x) \over [\varphi^-_\rho(x)]^2} & \mbox{(i)},\\[5pt]
{W[\varphi^-_{\rho + b},\varphi^+_\rho](x) \over [\varphi^+_\rho(x)]^2} & \mbox{(ii)}, \end{array}\right.
%= c \nu(x) \left\{\begin{array}{ll}
%  - \frac{\varphi^+_{\rho+ b}(x)\frac{\partial}{\partial x} \varphi^-_{\rho}(x)}{[\varphi^-_\rho(x)]^2} +\frac{\frac{\partial}{\partial x} \varphi^+_{\rho+ %b}(x)}{\varphi^-_\rho(x)} & \mbox{(i)}\\[5pt]
%\phantom{-}\frac{\varphi^-_{\rho+ b}(x)\frac{\partial}{\partial x} \varphi^+_{\rho}(x)}{[\varphi^+_\rho(x)]^2} - \frac{\frac{\partial}{\partial x} %\varphi^-_{\rho+ b}(x)}{\varphi^+_\rho(x)}  & \mbox{(ii)} \end{array}\right.
\label{dual_vol_Fmap}\!
\end{equation}
where $x = {\sf X}(F)$, with ${\sf X}\triangleq {\sf F}^{-1}$ as unique inverse map for the respective subfamilies
(i) $\F(x) = -\frac{a}{b} + c \frac{\varphi^+_{\rho+ b}(x)}{\varphi^-_{\rho}(x)}$ and
(ii) $\F(x) = -\frac{a}{b} + c \frac{\varphi^-_{\rho+ b}(x)}{\varphi^+_{\rho}(x)}$.
Both subfamilies have regular state space $F\in (-{a\over b},\infty)$.
Computing the respective Wronskians in equation (\ref{dual_vol_Fmap}) gives the diffusion coefficient function for three main dual
subfamilies (i) and (ii) as follows.

For the Bessel family, the maps are (i)
$\F(x) = \F^{(1)}_+ \equiv -\frac{a}{b} + c \frac{I_{\vert\mu\vert}\left(\frac{2}{\nu_0}\sqrt{2(\rho + b) x}\right)}{K_{\mu}\left(\frac{2}{\nu_0}\sqrt{2\rho x}\right)}$
and (ii) $\F(x) = \F^{(1)}_- \equiv -\frac{a}{b} + c \frac{K_\mu\left(\frac{2}{\nu_0}\sqrt{2(\rho + b) x}\right)}{I_{\vert\mu\vert}\left(\frac{2}{\nu_0}\sqrt{2\rho x}\right)}$, where
\begin{equation}\label{vol_Bessel}
  \sigma(F)= c\sqrt{2}\left\{\!\begin{array}{ll}
     \frac{\sqrt{\rho}\,I_{\vert\mu\vert}\left(\frac{2}{\nu_0}\sqrt{2(\rho + b) x}\right)K_{\vert\mu\vert+1}\left(\frac{2}{\nu_0}\sqrt{2\rho x}\right)}
{K_{\mu}^2\left(\frac{2}{\nu_0}\sqrt{2\rho x}\right)}
+ \frac{\sqrt{\rho + b}\,I_{\vert\mu\vert+1}\left(\frac{2}{\nu_0}\sqrt{2(\rho + b)x}\right)}{K_{\mu}\left(\frac{2}{\nu_0}\sqrt{2\rho x}\right)} &
                 \mbox{(i)},\\[5mm]
     \frac{\sqrt{\rho}\,K_\mu\left(\frac{2}{\nu_0}\sqrt{2(\rho + b) x}\right)I_{\vert\mu\vert+1}\left(\frac{2}{\nu_0}\sqrt{2\rho x}\right)}
{I_{\vert\mu\vert}^2\left(\frac{2}{\nu_0}\sqrt{2\rho x}\right)}
+ \frac{\sqrt{\rho + b}\,K_{\vert\mu\vert+1}\left(\frac{2}{\nu_0}\sqrt{2(\rho + b)x}\right)}{I_{\vert\mu\vert}\left(\frac{2}{\nu_0}\sqrt{2\rho x}\right)} &
                 \mbox{(ii)}.
  \end{array} \right.
\end{equation}

For the confluent hypergeometric family, we define $\upsilon \triangleq  \frac{\rho}{\lambda_1} + \mu_-$,
$\upsilon_b \triangleq \upsilon + {b\over \lambda_1}$, for the case $\lambda_1 > 0$, and $\upsilon \triangleq  \frac{\rho}{\vert\lambda_1\vert} + 1 + \mu_+$,
$\upsilon_b \triangleq \upsilon + {b\over \vert\lambda_1\vert}$ for $\lambda_1 < 0$. The dual maps are
(i) $\F(x) = -\frac{a}{b} + c \frac{\M\left(\upsilon_b,1 + \vert\mu\vert,\vert\kappa\vert x\right)}
{\U\left(\upsilon,1 + \vert\mu\vert,\vert\kappa\vert x\right)}$
and (ii) $\F(x) = -\frac{a}{b} + c \frac{\U\left(\upsilon_b,1 + \vert\mu\vert,\vert\kappa\vert x\right)}
{\M\left(\upsilon,1 + \vert\mu\vert,\vert\kappa\vert x\right)}$,
where $\vert\kappa\vert = \kappa$ ($-\kappa$) for $\lambda_1 > 0$ ($<0$).
The respective volatility functions for the dual subfamilies are
\begin{equation}
    \sigma(F)= c\vert\kappa\vert\nu_0\sqrt{x}\left\{\!\begin{array}{ll}
    \frac{\upsilon \M\left(\upsilon_b,1 + \vert\mu\vert,\vert\kappa\vert x\right)
    \U\left(\upsilon + 1,2 + \vert\mu\vert,\vert\kappa\vert x\right)}{\U^2\left(\upsilon,1 + \vert\mu\vert,\vert\kappa\vert x\right)} +
    \frac{\upsilon_b\,\M\left(\upsilon_b+1,2 + \vert\mu\vert,\vert\kappa\vert x\right)}
    {(1 + \vert\mu\vert)\,\U\left(\upsilon,1 + \vert\mu\vert,\vert\kappa\vert x\right)}
    & \mbox{(i)}, \\[5mm]
    \frac{\upsilon\,\M\left(1 + \upsilon,2 + \vert\mu\vert,\vert\kappa\vert x\right)\U\left(\upsilon_b,1 + \vert\mu\vert,\vert\kappa\vert x\right)}
    {(1 + \vert\mu\vert)\,\M^2\left(\upsilon,1 + \vert\mu\vert,\vert\kappa\vert x\right)}
    + \frac{\upsilon_b\,\U\left(\upsilon_b+1,2 + \vert\mu\vert,\vert\kappa\vert x\right)}{\M\left(\upsilon,1 + \vert\mu\vert,\vert\kappa\vert x\right)}
    & \mbox{(ii)}. \end{array} \right.
    \label{vol_Confluent}
\end{equation}

For the OU family of $F$-diffusions the dual subfamilies (i) and (ii) coalesce into a single family of processes.
This follows by the reflection symmetry $\varphi^+_s(x)=\varphi^-_s(-x)$. In this case we have
$W[\varphi^-_{\rho},\varphi^+_{\rho + \vartheta}](x) = W[\varphi^-_{\rho + \vartheta},\varphi^+_{\rho}](-x)$. Moreover, the respective
maps ${\sf F}(x)$ of subfamilies (i) and (ii) coincide upon interchanging $x\to -x$. Hence, the diffusion functions (i)
and (ii) in (\ref{dual_vol_Fmap}) are identical. We can therefore consider a single map
defined by the increasing function ${\sf F}(x)= -{a\over b} + c\frac{\varphi^+_{\rho + b}(x)}{\varphi^-_\rho(x)}
= -{a\over b} + c\frac{D_{-\upsilon_b}(-\sqrt{\kappa}\,x)}
{D_{-\upsilon}(\sqrt{\kappa}\,x)}$ giving the volatility function
\begin{equation}\label{vol_OU}
  \sigma(F) = c\nu_0\sqrt{\kappa} \left\{\upsilon_b \frac{D_{-(\upsilon_b+1)}(-\sqrt{\kappa}\,x)}
{D_{-\upsilon}(\sqrt{\kappa}\,x)}
+ \upsilon \frac{D_{-\upsilon_b}(-\sqrt{\kappa}\,x)D_{-(\upsilon+1)}(\sqrt{\kappa}\,x)}
{D^2_{-\upsilon}(\sqrt{\kappa}\,x)}\right\}\,,
\end{equation}
where $\upsilon\triangleq \rho/\lambda_1$, $\upsilon_b \triangleq \upsilon + {b\over \lambda_1}$.
In all of the above volatility functions, the value $x = {\sf X}(F) \equiv {\sf F}^{-1}(F)$ is given by the respective inverse map.

%Using (\ref{dual_vol_F1}), the diffusion function for the $\vartheta$-drifted OU family is
%\begin{equation}\label{Sigma_OU}
%  \sigma(F)=\displaystyle a_0\nu_0 e^{-{\kappa \over 2}x^2}
%\frac{W[\varphi^-_{\rho},\varphi^+_{\rho + \vartheta}](x)}{[D_{-\rho/\alpha_1}(\sqrt{\kappa}\,x)]^2},
%\end{equation}
%with $W$ in (\ref{Parabolic_Wronskian}) and $x={\sf X}(F)\equiv {\sf F}^{-1}(F)$ given by the inverse
%of (\ref{FmapUOU}). Dividing by ${\sf F}(x)$, the local volatility function takes the form
%\begin{equation}\label{Local_vol_OU}
%  \sigma_{loc}(F) =
%\displaystyle \sqrt{\kappa}\nu_0 \bigg[\frac{\rho+\vartheta}{\alpha_1}\frac{D_{-(1+{\rho + \vartheta\over \alpha_1})}(-\sqrt{\kappa}\,x)}
%{D_{-{\rho + \vartheta \over \alpha_1}}(-\sqrt{\kappa}\,x)}
%+ \frac{\rho}{\alpha_1}\frac{D_{-(1+{\rho \over \alpha_1})}(\sqrt{\kappa}\,x)}
%{D_{-{\rho \over \alpha_1}}(\sqrt{\kappa}\,x)}\bigg],
%\end{equation}

Subsets of these diffusion families with $a=0, b\ne 0$ have regular state space $F\in (0,\infty)$ and are useful for
modelling asset prices in finance. Figures~\ref{fig1} and~\ref{fig2}
display some computed curves of the local volatility function $\sigma_{loc}(F) \triangleq  \sigma(F)/F$ for the subfamilies in
equations (\ref{vol_Bessel}i), (\ref{vol_Confluent}i) and (\ref{vol_OU}) when $a=0, b\ne 0$. These three main subfamilies (respectively named here as the Bessel-$\mathsf{K}$, Confluent-$\mathcal{U}$ and OU models) are of interest since,
according to Proposition \ref{prop_class} below, the discounted processes $(e^{-bt} F_t)_{t\ge 0}$ obey the martingale property. By using the respective asymptotic properties of the fundamental functions provided in Appendix~\ref{sect_a1}, we derive the following asymptotic relations
for the local volatility functions of these subfamilies. For $\mu \ne 0$:
\[ \begin{array}{llll}
  \text{Bessel-}\mathsf{K}: & \text{Confluent-}\mathcal{U}: & \text{Ornstein-Uhlenbeck}: &\\
\sigma_{loc}(F) \sim C_0 F^{-\frac{1}{2\vert\mu\vert}} &  \sigma_{loc}(F) \sim C_1 F^{-\frac{1}{2\vert\mu\vert}} & \sigma_{loc}(F) \sim  C_2 \sqrt{\ln (1/ F)} & \text{as } F\to 0+,\\
\sigma_{loc}(F) \to \sqrt{2\rho}+\sqrt{2(\rho+b)} &  \sigma_{loc}(F) \sim C_3 \sqrt{\ln F} & \sigma_{loc}(F) \sim  C_4 \sqrt{\ln F} & \text{as } F\to \infty, \end{array} \]
where $C_i$ denote some positive constants. For the case $\mu=0$, $\sigma_{loc}(F) \sim c_0 e^{c_1/F}$,
as $F\to 0+$, for both Bessel-$\mathsf{K}$ and Confluent-$\mathcal{U}$ models where $c_0, c_1$ are positive constants.
As duals to the Bessel-$\mathsf{K}$ and Confluent-$\mathcal{U}$ models, we refer to the respective subfamilies of type (ii)
above as the the Bessel-$\mathsf{I}$ and Confluent-$\mathcal{M}$ models. By similar analysis, asymptotic expressions for
the local volatility functions of such models can also be readily derived.

As seen in Figures~\ref{fig1} and~\ref{fig2}, by adjusting parameters,
the models are all readily calibrated to attain a prescribed level of local volatility for a given value of $F$.
The respective sets of freely adjustable model parameters for the Bessel, Confluent and OU subfamilies are:
($\rho,b, \mu,c,\nu_0$), ($\rho,b,\kappa,\mu,c,\nu_0$)
and ($\rho,b,\kappa,c,\nu_0$).

The following proposition characterizes the Bessel, confluent hypergeometric, and Ornstein-Uhlenbeck families of $F$-diffusions with linear drift
in terms of the conservation of the expectation rate and the martingale property of the discounted process.
This result generalizes the special results obtained previously for the driftless families of $F$-diffusions (see \cite{CM06}).
\begin{proposition} \label{prop_class}
Consider the Bessel, Confluent and OU regular diffusions $F_t=\F(X_t^{(\rho)})\in\R_+$,
with linear drift function $\alpha(F) = bF$ and nonlinear diffusion functions defined by (\ref{vol_Bessel})--(\ref{vol_OU}), respectively.
Then,
\begin{enumerate}[(1)]
   \item the discounted processes $(e^{-bt} F_t)_{t\ge 0}$ of the Bessel-$\mathsf{K}$ and Confluent-$\mathcal{U}$ subfamilies (i),
   with origin specified as killing in case $\vert\mu\vert < 1$, are martingales for all allowable choices of the model parameters;
   \item the discounted processes $(e^{-bt} F_t)_{t\ge 0}$ of the Bessel-$\mathsf{I}$ and Confluent-$\mathcal{M}$ subfamilies (ii), are strict supermartingales for all allowable choices of the model parameters;
   \item the discounted processes $(e^{-bt} F_t)_{t\ge 0}$ of the OU family are martingales for all allowable choices of the model parameters.
\end{enumerate}
\end{proposition}
\begin{proof}
The martingale property, for the discounted Bessel and Confluent subfamilies (i) and for the discounted
OU family, is proven by applying Corollary \ref{corollary1}, where $\hat\psi_{s + \rho} = \varphi^+_{s + \rho}$ and $\hat\phi_{s + \rho} = \varphi^-_{s + \rho}$, and $\F = c{\varphi^{+}_{\rho + b}\over \varphi^{-}_\rho}$, i.e. $\hat{v}_{\rho + b} = c\varphi^{+}_{\rho + b}$.
The boundary conditions in (\ref{cond_martingale}) now read ${W[\varphi^{+}_{\rho + b}, \varphi^+_{\rho + s}](l+) \over \s(l+)} = 0$ and ${W[\varphi^{+}_{\rho + b}, \varphi^-_{\rho + s}](r-) \over \s(r-)} = 0$, which hold true by letting $\mathrm{Re\,}s > b$ in the asymptotic Wronskian relations in
Appendix~\ref{sect_a1} for the respective SQB, CIR and OU processes (where $l=0, r=\infty$ for the SQB and CIR and $l=-\infty, r=\infty$ for the OU).
For the discounted Bessel and Confluent subfamilies (ii): $\hat{v}_{\rho + b} = c\varphi^{-}_{\rho + b}$, $\F = c{\varphi^{-}_{\rho + b}\over \varphi^{+}_\rho}$. In this case the left and right limits in (\ref{cond_martingale}) evaluate to
${W[\varphi^{-}_{\rho + b}, \varphi^+_{\rho + s}](0+) \over \s(0+)} = const. > 0$ and
${W[\varphi^{-}_{\rho + b}, \varphi^-_{\rho + s}](\infty) \over \s(\infty)} = 0$, for real values of $s > 0$.
This implies that $\mathcal{E} < 0$ in (\ref{mebias}), i.e. the strict supermartingale property holds.
\end{proof}

\begin{figure}[ht]
\begin{center}
\includegraphics[width=0.45\linewidth]{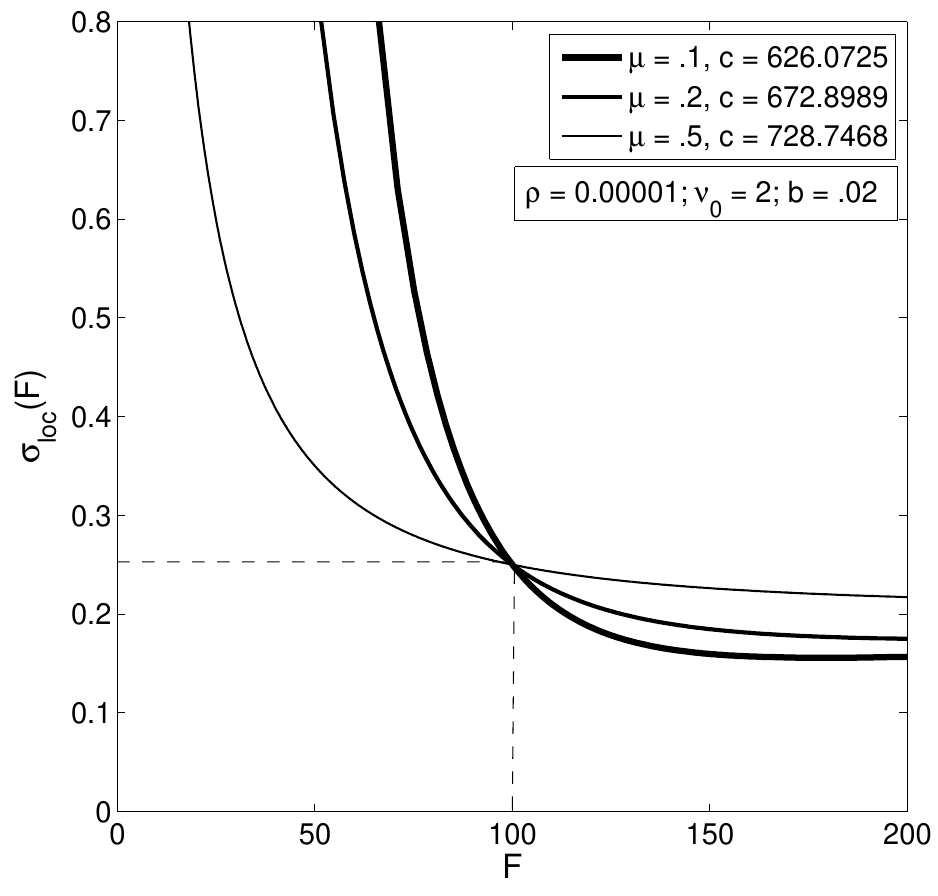}
\caption{Sample local volatility curves for the drifted Bessel-$\mathsf{K}$ model with drift parameters $a=0, b\ne 0$ and other parameters
calibrated such that $\sigma_{loc}(F) = 0.25$ at $F=100$.}%
\label{fig1}
\end{center}
\end{figure}

\begin{figure}
\begin{center}
\subfigure[]{
\resizebox*{0.45\linewidth}{!}{\includegraphics[width=\linewidth]{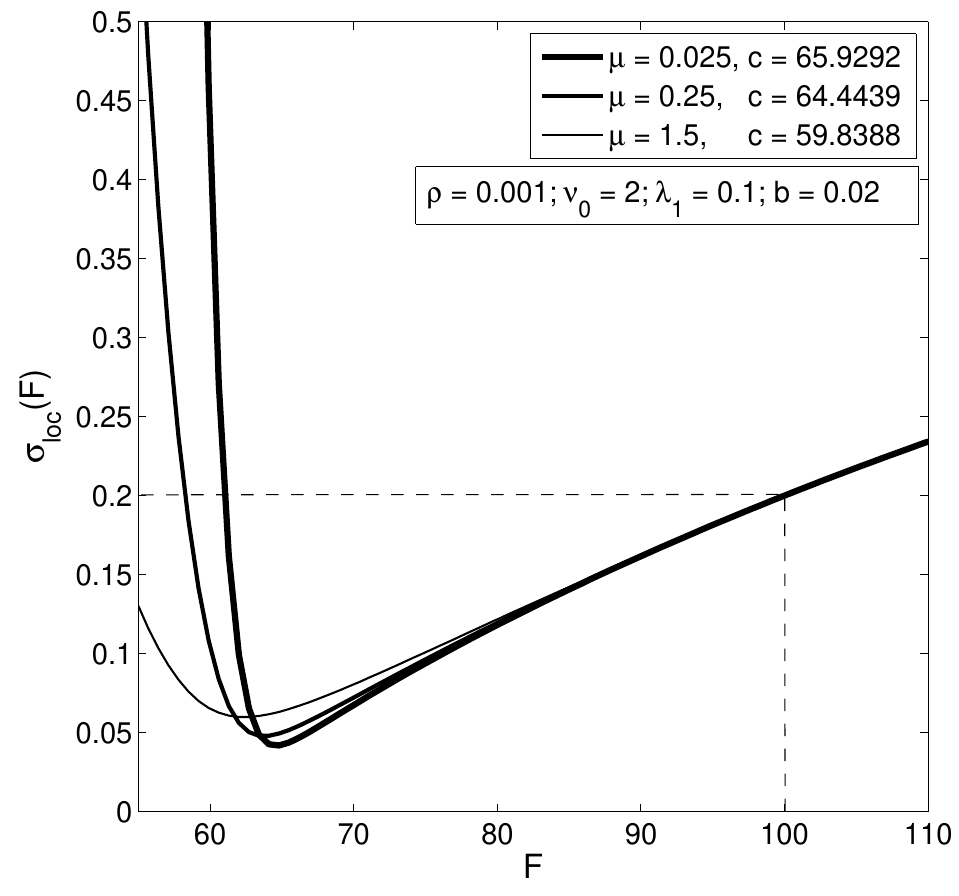}}}%
\subfigure[]{
\resizebox*{0.45\linewidth}{!}{\includegraphics[width=\linewidth]{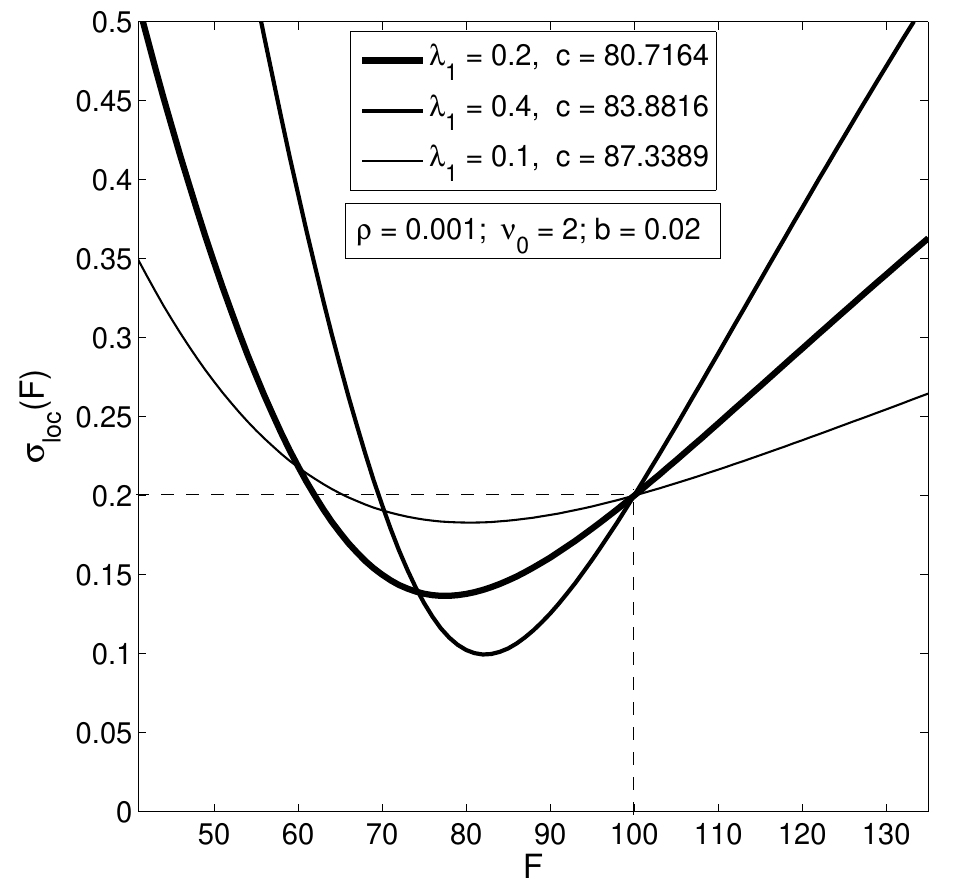}}}%
\caption{Sample local volatility curves for the Confluent-$\U$ (a) and OU (b) models with drift parameters $a=0, b\ne 0$ and other parameters calibrated such that $\sigma_{loc}(F) = 0.20$ at $F=100$.}%
\label{fig2}
\end{center}
\end{figure}

\subsection{Classification of Monotonic Maps} \label{subsect4.2}
In addition to the set of monotonic maps~$\F ^{(1)}$ defined by
(\ref{FMAPpm1}), there are four other classes of
maps, as follows:
\begin{eqnarray}
 \F^{(2)}_\pm(x) \triangleq \epsilon\,c\frac{\varphi^\pm_{\rho+b}(x)}{\varphi^\pm_{\rho}(x)}\,,
 && \F^{(3)}_\pm(x) \triangleq
 \epsilon\,\frac{c_1\varphi^+_{\rho+b}(x)+c_2\varphi^-_{\rho+b}(x)}{\varphi^\mp_{\rho}(x)}\,, \label{FMAPpm23}  \\
 \F^{(4)}_\pm(x) \triangleq
 \epsilon\,\frac{\varphi^\pm_{\rho+b}(x)}{q_1\varphi^+_{\rho}(x)+q_2\varphi^-_{\rho}(x)}\,,
 && \F^{(5)}(x) \triangleq  \epsilon\,\frac{c_1\varphi^+_{\rho+b}(x)-c_2\varphi^-_{\rho+b}(x)}{q_1\varphi^+_{\rho}(x)+q_2\varphi^-_{\rho}(x)}\,, \label{FMAPpm45}
\end{eqnarray}
where $\rho,\rho + b,c,c_1,c_2>0$, $\epsilon=\pm 1$, and $q_1,q_2>0$ with the only exception of $\F ^{(5)}$ for which one of $q_1,q_2$ may be zero.
Notice that we omit the additive term $-\frac{a}{b}$.

As follows from Proposition~\ref{prop2}, a function $\F^{(k)}$ for $k=1,2,3,4$ defines a monotonic map if and only if the
Wronskian $W(x)= W[\hat{u}_\rho,\hat{v}_{\rho+b}](x)$
%c_1W[\varphi^+_{\rho+b},\varphi^\pm_{\rho}](x)+c_2W[\varphi^-_{\rho+b},\varphi^\pm_{\rho}](x)
has the same sign in neighbourhoods of both endpoints $l$ and $r$. For the map $\F^{(5)}$ in (\ref{FMAPpm45}), monotonicity follows from Proposition~\ref{prop3} when $b>0$. If
$b<0$, then we again need to analyse the asymptotic behaviour of
$W(x)$, and hence of $W[\varphi^\pm_{\rho},\varphi^\pm_{\rho + b}](x)$, as $x\to l+$ and $x\to r-$. The asymptotics of such Wronskians
are given in Proposition~\ref{aprop1} of Appendix~\ref{sect_a1}. The following lemma summarizes the monotonicity properties for all
families of maps presented so far.
\begin{lemma} \label{lemma_map_class}
Let $\varphi^\pm$ be the fundamental solutions for the SQB, CIR, or OU diffusion process.
The maps $\F^{(l)}$, $l=1,2,\ldots,5$, defined in (\ref{FMAPpm1}), (\ref{FMAPpm23}) and (\ref{FMAPpm45}), are strictly monotonic under the following conditions:
$\F^{(1,2)}$ --- for all choices of parameters;
$\F^{(3)}$  --- if $b<0$;
$\F^{(4)}$ --- if $b>0$;
$\F^{(5)}$ --- if and only if $b>0$.
The derivative $\mathrm{d\F}^{(l)}_\pm/\mathrm{d}x$, $l=1,3,4$, has sign = $\pm \epsilon$.
For cases 2 and 5 we have: $\sgn(\mathrm{d\F}^{(2)}_\pm/\mathrm{d}x)=\pm \epsilon\,\sgn(b)$ and  $\sgn(\mathrm{d\F}^{(5)}/\mathrm{d}x)=\epsilon$.
 \end{lemma}
\begin{proof}
The proof follows directly from
Propositions~\ref{prop2}, \ref{prop3} and \ref{aprop1}, and asymptotics provided in Appendix~\ref{sect_a1}.
%For cases
%\textit{(iii)--(v)}, the asymptotics provided in
%Appendix~\ref{sect_a1} are also applied.
\end{proof}

\section{The Bessel Family of Mean-Reverting $F$-Diffusions} \label{sect5}

Let us come back to the $X^{(\rho)}$-process with generator~(\ref{Generator_rho}). Consider a strictly monotonic twice continuously differentiable map
$\F$ with inverse $\X$. Such a map generates a diffusion process $\{F_t \triangleq \F(X^{(\rho)}_t), t\ge 0\}$ with generator in (\ref{GeneratorF}).
We recall that a linear-drift $F$-diffusion is obtained by requiring that $\F$ solves equation (\ref{FODE}). An alternative approach is to fix the diffusion coefficient $\widetilde{\sigma}(x)=\sigma(\F(x))$ and then find $\F$ by integrating the derivative $\F'$. For simplicity we assume that $\widetilde{\sigma}(x)$ is a combination of elementary functions such as power and exponential functions.

Here we are interested in diffusions with nonlinear mean-reverting drift $\alpha(F)$. That is there exists $F_0\in(F^l,F^r)$ so that $\alpha(F_0)=0$ and $\alpha(F_1)>0>\alpha(F_2)$ for all $F_1,F_2\in(F^l,F^r)$ with $F_1<F_0<F_2$. The drift coefficient $\alpha$ is a continuous function and it is sufficient to check whether $\sgn(\alpha(F^l+)=1$ and $\sgn(\alpha(F^r-))=-1$ hold.
%and therefore $\F(x_0)=0$ at some point $x_0\in (l,r)$.
To find maps that generate such diffusions, we analyze the asymptotic behavior of the function $\widetilde{\alpha}(x)=\alpha(\F(x))=(\mathcal{G}^{(\rho)} \F)(x)$ as $x$ approaches $l$ or $r$. There are two cases. If $\F'>0$ (i.e. the mapping $\F$ is increasing), then we need to verify that $\sgn(\widetilde{\alpha}(l+))=1$ and $\sgn(\widetilde{\alpha}(r-))=-1$. If $\F'<0$ (i.e. the mapping $\F$ is decreasing), we verify that $\sgn(\widetilde{\alpha}(l+))=-1$ and $\sgn(\widetilde{\alpha}(r-))=1$. Here we assume that if $f(x)\downarrow 0$ (or $f(x)\uparrow 0$) as $x\to e$ (where $e\in\{l+,r-\}$), then $\sgn(f(e))=1$ (or $\sgn(f(e))=-1$).

Below, as explicit examples, we consider families of nonlinear mean-reverting $F$-diffusions generated from the SQB process. Clearly any other underlying diffusion generates its own family of models.

\subsection{One Special Case with a Linear Mapping} \label{sect5.1}
Consider one special case where the derivative $\F'$ is constant, i.e. $F_t=a + bX_t^{(\rho)}$ for some real $a$ and $b$, with $X^{(\rho)}$-process
of subfamily (ii) generated from the SQB process with $\hat{u}_\rho(x)=\varphi^-_\rho(x) = x^{-\mu/2}K_{\mu}(2\sqrt{2\rho x}/\nu_0)$,
$\mu\in\R, \nu_0>0$. Without loss of generality, let $a=0,b=1$ and consider the case where $F_t=X_t^{(\rho)}\in\R_+$, i.e. with map $\F(x) = x$, $\X(F) = F$.
By (\ref{drift_F2}), the drift is now $\alpha(F) = \lambda_0 + \nu_0^2 F \hat{u}_\rho'(F)/\hat{u}_\rho(F)$, $\lambda_0 = \nu_0^2(1+\mu)/2$.
By the differential recurrence property of the Bessel-K function, $\hat{u}_\rho'(x)
= -(\sqrt{2\rho x}/\nu_0)x^{-1-\mu/2}K_{\mu+1}(2\sqrt{2\rho x}/\nu_0)$. Hence, the $F$-diffusion satisfies the SDE
\[ dF_t = \alpha(F_t) dt + \nu_0\sqrt{F_t}dW_t,\,\,\mbox{  where  }\,\alpha(F) = \lambda_0 - \nu_0\sqrt{2\rho F} \frac{K_{\mu+1}(2\sqrt{2\rho F}/\nu_0)}{K_{\mu}(2\sqrt{2\rho F}/\nu_0)}. \]

Using large and small argument asymptotic properties of the Bessel $K$ function gives the following limits
 for the drift coefficient:
\[ \alpha(0+) = \left\{ \begin{array}{ll}
                            \nu_0^2\left(1-\mu\right)/2, & \mu\ge 0  \\
                            \nu_0^2\left(1+\mu\right)/2, & \mu< 0
                          \end{array} \right. \quad\mbox{and}\quad
   \alpha(F) \sim \lambda_0-\nu_0\sqrt{2\rho F}, \mbox{ as }F\to+\infty. \]
Therefore, the above model admits a mean-reverting drift coefficient for all $\mu\in(-1,1)$: the~drift approaches a positive constant at the left endpoint $F^l=0$ and becomes increasingly negative in proportion to the square-root of the process value as it approaches the right endpoint $F^r=\infty$. The diffusion coefficient is a square-root function: $\sigma(F)=\nu_0\sqrt{F}$. Thus, this model can be viewed as a (nonlinear drift) modification of the CIR short-rate model.

We recall from Lemma \ref{lemma_class_Bessel} that the above process $X_t^{(\rho)}\in\R_+$ (with $q_1=0, q_2=1$) has origin as
regular when $|\mu| < 1$ and exit when $|\mu|\geq 1$. As an interest (short)-rate model with mean-reversion,
we take $|\mu| < 1$. In this case, the origin is regular and may be specified as either killing or reflecting. Specifically,
we now specify the origin as killing for $|\mu| < 1$. Hence, the transition PDF is given explicitly by (\ref{PDF_Bessel_2}), i.e.
$p^{(\rho)}_X(t;x_0,x)$ has the form in (\ref{prho}) with $\hat{u}_\rho(x) = x^{-\mu/2}K_{\mu}(2\sqrt{2\rho x}/\nu_0)$ and
$p_X(t;x_0,x)$ given by (\ref{PrkernelSQB}) with $\tilde{\mu} = |\mu|$.
Based on this connection between the $X^{(\rho)}$ and SQB processes, the following result provides us with a closed-form integral expression for all $\mu\in\R$ that may be useful for pricing a zero-coupon bond in case $|\mu| < 1$ and wherein the interest rate process $R_t$ is modeled as $R_t \equiv F_t\triangleq   X_t^{(\rho)}, t\ge 0$.
%\begin{equation}\label{bond_price}
%P(r,t,T)=\E\left[e^{-\int_t^T R_u du} \mid R_t=r\right].
% \end{equation}
\begin{lemma}\label{lemma_ZCBP4BK}
Consider the Bessel subfamily (ii) of $X^{(\rho)}$-processes started at $X_0^{(\rho)}= x_0>0$,
with transition PDF in (\ref{PDF_Bessel_2}) where $q_1=0, q_2=1$.
Let $f:\R_+ \to \R$ and assume $\E_{x_0}[|f(X^{(\rho)}_t)| ] < \infty$, $t\ge 0$, $\lambda > 0$, then
 \begin{equation}\label{xZCBP4BK}
    \E_{x_0}\left[f(X^{(\rho)}_t) e^{-\lambda \int_0^t X^{(\rho)}_u du}\right]
    = {A_t e^{-\rho t -x_0 B_t}\over K_{\mu}(2\sqrt{2\rho x_0}/\nu_0)}
    \int_0^\infty  e^{-x B_t} f(x) I_{|\mu|}\big(2A_t\sqrt{x_0x}\big)K_{\mu}\big(2\sqrt{2\rho x}/\nu_0\big)\,dx,
 \end{equation}
where $A_t = \frac{\sqrt{2\lambda}/\nu_0}{\sinh\big(\nu_0 t\sqrt{\lambda/2}\big)}$, $B_t =\frac{\sqrt{2\lambda}}{\nu_0}\coth\big(\nu_0 t\sqrt{\lambda/2}\big)$.
\end{lemma}
\begin{proof}
By conditioning on the terminal value, $\E_{x_0}[f(X^{(\rho)}_t) e^{-\lambda \int_0^t X^{(\rho)}_u du}]$ is given by
$$\E_{x_0}\left[f(X^{(\rho)}_t) \,\E_{x_0}\big[e^{-\lambda \int_0^t X^{(\rho)}_u du}\mid X_t^{(\rho)}\big]\,\right]
= \int_0^\infty p^{(\rho)}_X(t;x_0,x) f(x) \E_{x_0}\big[e^{-\lambda \int_0^t X^{(\rho)}_u du}\mid X_t^{(\rho)} = x\big]dx$$
where the conditional expectation is equivalent to the expectation of the
negative exponential of the time integrated {\it Bessel Bridge $X^{(\rho)}$-process} that is started at $x_0$ and pinned at $x$ at time $t$.
Now, since $p^{(\rho)}_X(t;x_0,x)$ has the form in (\ref{prho}), where $p_X(t;x_0,x)$ is given by (\ref{PrkernelSQB}) with $\tilde{\mu} = |\mu|$, it follows
that the Bessel Bridge $X^{(\rho)}$-process has the same probability law as the corresponding Squared Bessel (SQB) Bridge process.
That is, given any partition $0<t_1< t_2 < ... <t_n < t$, the path probability densities of the
Bridge $X^{(\rho)}$-process and corresponding SQB Bridge $X$-process are equivalent:
$p^{(\rho)}_X(t_1;x_0,x_1) p^{(\rho)}_X(t_2-t_1;x_1,x_2)\times ...\times p^{(\rho)}_X(t-t_n;x_n,x)/p^{(\rho)}_X(t;x_0,x) =
p_X(t_1;x_0,x_1) p_X(t_2-t_1;x_1,x_2)\times ...\times p_X(t-t_n;x_n,x)/p_X(t;x_0,x)$.
Hence, the conditional expectation reduces to that of the SQB Bridge process, i.e.
$\E_{x_0}[e^{-\lambda\int_0^t X^{(\rho)}_u du}\mid X_t^{(\rho)} = x] = \E_{x_0}[e^{-\lambda\int_0^t X_u du}\mid X_t = x]$.
The latter is given explicitly (e.g. see equation (2.m) in \cite{Pitman} or page 76 in \cite{BS02}, which we adapt for all $\nu_0>0$; see also \cite{RY99}):
$$\E_{x_0}\left[e^{-\lambda \int_0^t X_u du}\mid X_t = x\right] =
\frac{\sqrt{\lambda}\tau}{\sinh \sqrt{\lambda}\tau}
\exp\left\{{2(x+x_0)\over \nu_0^2 t}\left[1 - \sqrt{\lambda}\tau\coth\sqrt{\lambda}\tau\right]\right\}
\frac{I_{|\mu|}(\frac{2\sqrt{2\lambda x_0 x}}{\nu_0\sinh \sqrt{\lambda}\tau } )}
{I_{|\mu|}\big(4\sqrt{x_0x}/\nu_0^2 t \big)},$$
$\tau = \nu_0 t/\sqrt{2}$.
Equation (\ref{xZCBP4BK}) is now obtained by inserting this expression and the explicit form of $p^{(\rho)}_X(t;x_0,x)$ into the above integral.
\end{proof}

\subsection{The Case with a Power Mapping} \label{sect5.2}

Let us consider the important case of a power-type mapping function. Suppose that $\widetilde\sigma(x)=|\delta|x^{\beta}$ with $\delta,\beta\in\R$, $\delta\neq 0$, then the derivative of the mapping is
\[ \F'_p(x) = \pm\frac{\widetilde\sigma(x)}{\nu_0\sqrt{x}}=\frac{\delta x^{\beta}}{\nu_0\sqrt{x}} = \frac{\delta}{\nu_0}x^{\beta-\frac{1}{2}}.\]
As is seen from the above equation, $\F_e$ is monotonically \emph{increasing} if $\delta>0$; it is monotonically \emph{decreasing} if $\delta<0$.
Assume that $\beta\neq \frac{1}{2}$ since for $\beta=\frac{1}{2}$ the power-type mapping reduces to a linear mapping studied in Section~\ref{sect5.1}. The mapping is then obtained by integration of $\F'_p$ over $(0,x)$ if $\F_p$ is increasing (or over $(x,\infty)$ if $\F_p$ is decreasing):
\begin{equation}\label{powerFmap}
    \F_p(x) = \frac{\delta x^{\beta+0.5}}{\nu_0(\beta+0.5)}.
\end{equation}
The condition $(\beta+0.5)\delta>0$ guarantees that the function $\F_p: \R_+ \to \R_+$ is monotonic.
Such a mapping produces a mean-reverting model under the conditions stated in Lemma~\ref{lemma_SQBpwr}. Plots of typical drift and diffusion coefficients,  $\alpha$ and $\sigma$, are given in Figure~\ref{fig3}.

\begin{figure}[ht]
\begin{center}
\subfigure[]{
\resizebox*{0.45\linewidth}{!}{\includegraphics[width=\linewidth]{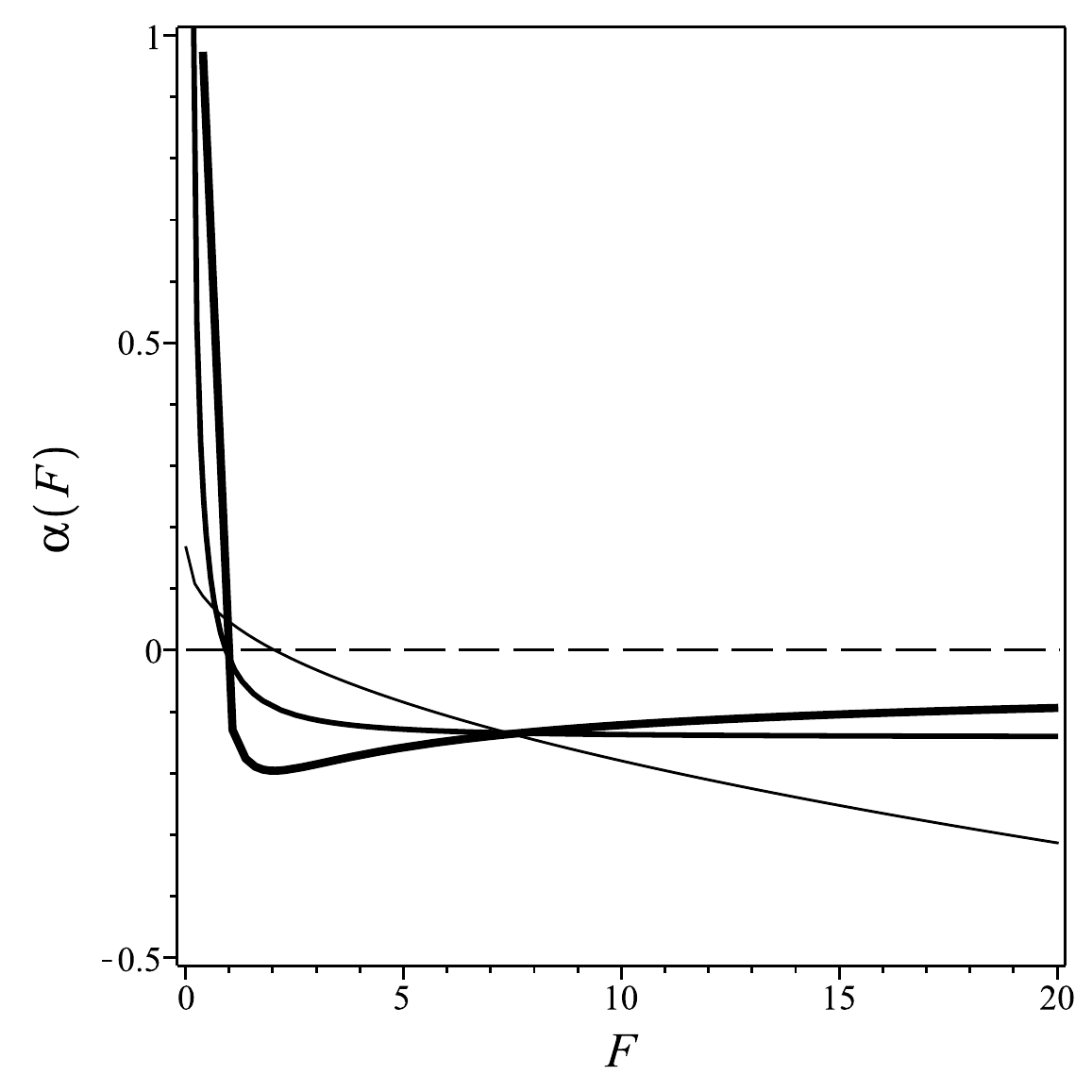}}}%
\subfigure[]{
\resizebox*{0.45\linewidth}{!}{\includegraphics[width=\linewidth]{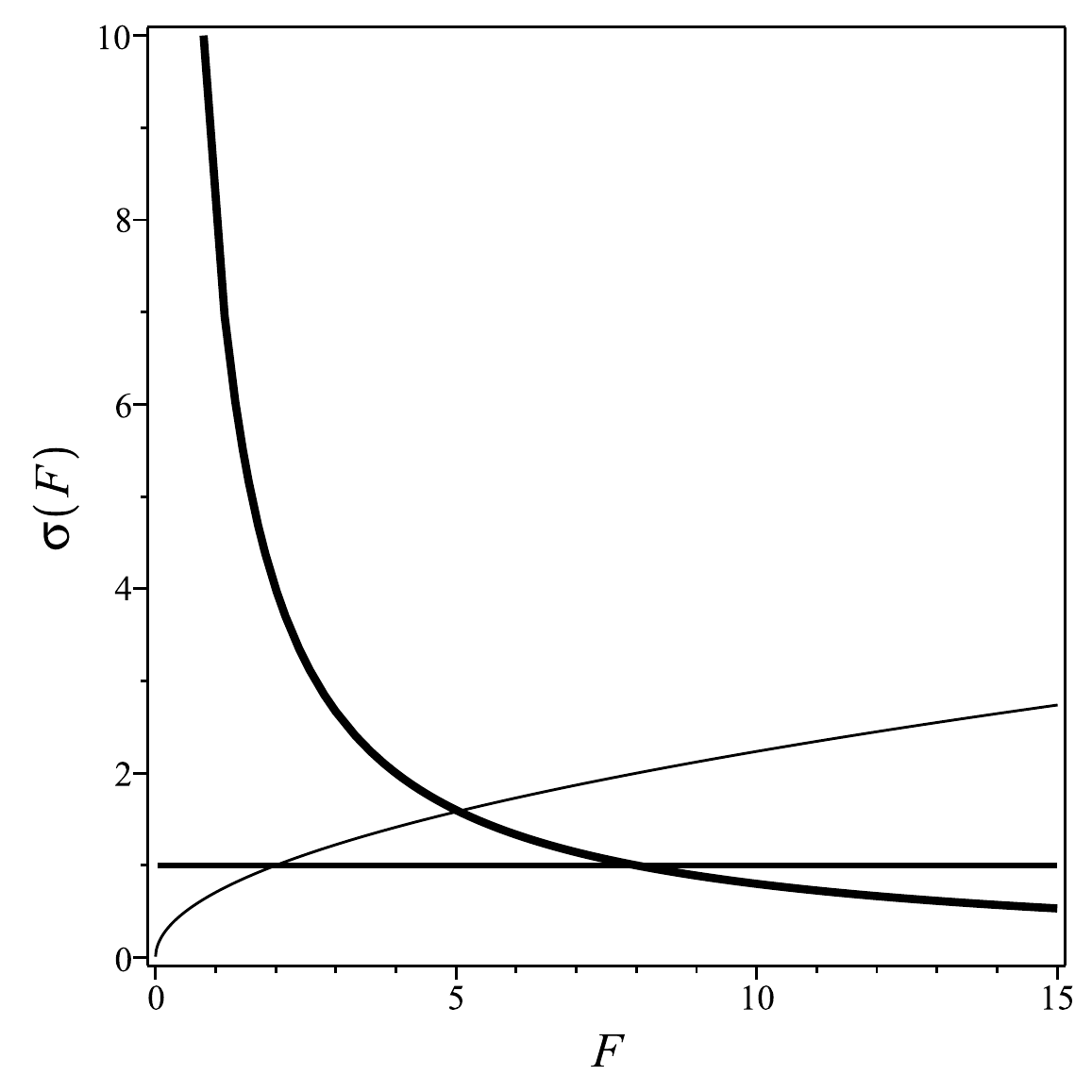}}}%
\caption{The drift and diffusion coefficient functions, $\alpha$ and $\sigma$, of the Bessel mean-reverting models with parameters $q_1=0$, $q_2=1$, $\nu_0=0.5$, $\mu=0.25$, $\rho=0.01$, $\delta=1$, and $\beta=0.5$ (thin line), or $\beta=0$ (moderate line) or $\beta=-0.125$ (thick line).}%
\label{fig3}
\end{center}
\end{figure}

\begin{lemma} \label{lemma_SQBpwr}
  $F$-diffusions generated from the SQB $X^{(\rho)}$-process with the use of the mapping function $\F_p$ defined in~(\ref{powerFmap}) have the following properties.
  \begin{enumerate}[(i)]
    \item $\{q_1=0,\,q_2>0\}$: The diffusion $F_t=\F_p(X_t^{(\rho)})$ is non-exploding and mean-reverting if $\delta>0$, $\beta>-0.5$ (i.e. $\F$ is monotonically increasing), and $|\mu|<\beta+0.5$. The diffusion coefficient is
        \[ \sigma(F)=\sigma_0 F^{1-(2\beta+1)^{-1}},\mbox{ where } \sigma_0=\delta(\nu_0(\beta+0.5)/\delta)^{(\beta+0.5)^{-1}}.\]
        The drift coefficient $\alpha$ has the following asymptotic properties:
        \[ \alpha(F)\sim \alpha_0 F^{1-(\beta+0.5)^{-1}}, \mbox{ as }F\to 0+,\mbox{ and } \alpha(F)\sim -\alpha_1 F^{1-(2\beta+1)^{-1}}, \mbox{ as }F\to \infty, \]
        where $\alpha_0=\nu_0\delta(\beta-|\mu|+0.5)/2$ and $\alpha_1=\delta\sqrt{2\rho}$ are positive constants.
    \item $\{q_1>0,q_2\geq 0\}$: The diffusion $F_t=\F_p(X_t^{(\rho)})$ is not mean-reverting for every choice of $\F_p$.
  \end{enumerate}
\end{lemma}
\begin{proof}
The proof follows directly from asymptotic properties provided in Appendix~\ref{sect_a1}.
\end{proof}

\subsection{The Case with an Exponential Mapping} \label{sect5.3}

Another possible choice of the mapping is a combination of power and exponential functions. Let us consider a mapping $\F_e:\R_+\to\R_+$ with derivative
\begin{equation}\label{expFmapder}
    \F'_e(x)  = \frac{\delta x^\beta \exp(\theta x^\gamma)}{\nu_0\sqrt{x}} = \frac{\delta}{\nu_0} x^{\beta-\frac{1}{2}} \exp(\theta x^\gamma),
\end{equation}
where $\theta\neq 0$ and $\gamma\neq 0$ (otherwise this case reduces to that with a power-type mapping considered in Section~\ref{sect5.2}). Here we assume that $\beta\neq \frac{1}{2}$. The case with a purely exponential mapping can be studied similarly.
The mapping is monotonically \emph{increasing} if $\delta>0$; it is monotonically \emph{decreasing} if $\delta<0$.
$\F_e$ is obtained by integrating its derivative in (\ref{expFmapder}). If $\F_e$ is increasing ($\delta>0$) then $\F_e$ obtains via integration of (\ref{expFmapder}) over $(0,x)$. Zero is an integrable singularity if either $\{\beta>-\frac{1}{2}, \theta>0, \gamma>0\}$ or $\{\theta<0, \gamma<0\}$. If $\F_e$ is decreasing ($\delta<0$) then the mapping is given by integration of (\ref{expFmapder}) over $(x,\infty)$. The derivative $\F_e'$ is integrable at infinity if
either $\{\beta<-\frac{1}{2}, \theta>0, \gamma<0\}$ or $\{\theta<0, \gamma>0\}$.

Therefore, to guarantee that the function $\F_e$ with its derivative given in (\ref{expFmapder}) is monotonic and maps $\R_+$ onto $\R_+$,
we have the following sets of conditions.
For a monotonically increasing $\F_e$, $\delta>0$, $\beta> -\frac{1}{2}$, and either $\{\theta>0, \gamma>0\}$ or $\{\theta<0, \gamma<0\}$ hold. For a monotonically decreasing $\F_e$, $\delta<0$, $\beta<-\frac{1}{2}$, and either $\{\theta>0, \gamma<0\}$ or $\{\theta<0, \gamma>0\}$ hold.
The resulting mapping has the following properties:
\begin{equation}
    \label{expFmap}
    \F_e(x)   \propto x^{\beta-\gamma+0.5} e^{\theta x^\gamma} \; \begin{array}{l} \text{ as }x\to 0+ \text{ and }\gamma<0 \text{ or}\\ \text{ as }x\to \infty\text{ and }\gamma>0;  \end{array} \quad %\frac{\delta}{\nu_0\theta\gamma}
    \F_e(x)   \propto x^{\beta+0.5} \; \begin{array}{l}\text{ as }x\to 0+ \text{ and }\gamma>0\text{ or}\\ \text{ as }x\to \infty \text{ and }\gamma<0. \end{array}
    %\frac{\delta }{\nu_0(\beta+0.5)}
\end{equation}

\begin{lemma} \label{lemma_SQBexp}
  The $F$-diffusion generated from the SQB $X^{(\rho)}$-process with the use of the mapping function $\F_e$ is non-exploding and mean-reverting in the following cases.
  \begin{enumerate}[(i)]
    \item $\{q_1=0,\,q_2>0\}$: The mapping $\F_e$ is monotonically increasing with $\delta>0$, $\beta>-\frac{1}{2}$, $\theta>0$, and $\gamma\in(0,\frac{1}{2})$. The diffusion and drift coefficients, $\sigma$ and $\alpha$, have the following asymptotic properties:
        \[ \begin{array}{l}
             \sigma(F)\sim\sigma_0 F^{1-(2\beta+1)^{-1}}\\
             \alpha(F)\sim \alpha_0 F^{1-(\beta+0.5)^{-1}}
            \end{array} \mbox{ as } F\to 0+; \mbox{  }
            \begin{array}{l}
              \sigma(F)\sim\sigma_1 F (\ln F)^{1-(2\gamma)^{-1}}\\
              \alpha(F)\sim -\alpha_1 F (\ln F)^{1-(2\gamma)^{-1}}
            \end{array}
            \mbox{ as }F\to \infty,
        \]
        where $\sigma_0$, $\sigma_1$, $\alpha_0$, and $\alpha_1$ are positive constants ($\alpha_0$ and $\sigma_0$ are the same as those in Lemma~\ref{lemma_SQBpwr}).
    \item $\{q_1>0,q_2=0\}$: The mapping $\F_e$ is monotonically decreasing with $\delta<0$, $\beta<-0.5$, $\theta<0$, $\gamma\in(0,\frac{1}{2})$, and $\mu>-(\beta+0.5)$ (hence, $\mu>0$). The coefficients $\sigma$ and $\alpha$ satisfy
        \[ \begin{array}{l}
               \sigma(F)\sim\sigma_0 F |\ln F|^{1-(2\gamma)^{-1}}\\
               \alpha(F)\sim \alpha_0 F |\ln F|^{1-(2\gamma)^{-1}}
            \end{array} \mbox{ as } F\to 0+; \mbox{ }
            \begin{array}{l}
               \sigma(F)\sim\sigma_1 F^{1-(2\beta+1)^{-1}}\\
               \alpha(F)\sim -\alpha_1 F^{1-(\beta+0.5)^{-1}}
            \end{array}
            \mbox{ as }F\to \infty,
        \]
         where $\sigma_0$, $\sigma_1$, $\alpha_0$, and $\alpha_1$ are positive constants ($\sigma_1$ is the same as $\sigma_0$ in Lemma~\ref{lemma_SQBpwr}).
    \item $\{q_1\geq 0,q_2\geq 0\}$: The mapping $\F_e$ is monotonically increasing with $\delta>0$, $\beta>-\frac{1}{2}$, $\theta<0$, and $\gamma<0$. The coefficients $\sigma$ and $\alpha$ satisfy
        \[ \begin{array}{l}
             \sigma(F)\sim\sigma_0 F |\ln F|^{1-(2\gamma)^{-1}}\\
             \alpha(F)\sim \alpha_0 F |\ln F|^{2-1/\gamma}
            \end{array} \mbox{ as } F\to 0+; \mbox{  }
            \begin{array}{l}
               \sigma(F)\sim\sigma_1 F^{1-(2\beta+1)^{-1}}\\
               \alpha(F)\sim -\alpha_1 F^{1-(\beta+0.5)^{-1}}
            \end{array}
            \mbox{ as }F\to \infty,
        \]
         where $\sigma_0$, $\sigma_1$, $\alpha_0$, and $\alpha_1$ are positive constants ($\sigma_1$ is the same as $\sigma_0$ in Lemma~\ref{lemma_SQBpwr}).

    \end{enumerate}
\end{lemma}
\begin{proof}
The proof follows directly from asymptotic properties provided in Appendix~\ref{sect_a1}.
\end{proof}

%\section{Numerical Results} \label{sect6}
%\subsection{Calibration of the Confluent-$\U$ Model to Option Prices} %\label{sect6.1}
%\subsection{Calibration of the Bessel-$K$ Model to Bond Prices} %\label{sect6.2}

\section{Conclusions} \label{sect7}
By applying a diffusion canonical transformation method, which combines special classes of monotonic mappings
and measure changes via Doob-h transforms, this paper developed various new families of exactly solvable multiparameter
one-dimensional time-homogeneous diffusion models.
These multiparameter families of solvable diffusions are
generally divided into two main classes; the first class is specified by having affine (linear) drift
with various resulting nonlinear diffusion coefficient functions, while the second class
allows for various specifications of a (generally nonlinear) diffusion
coefficient with a resulting nonlinear drift function. The present paper hence
significantly extends the diffusion canonical transformation methodology to include all of the more restrictive families
of {\it driftless} diffusions that were obtained and applied in previous literature (e.g., see \cite{ACCL,CM06}).
Moreover, the second main class of diffusions gives rise to various solvable models with
both nonlinear diffusion and nonlinear drift coefficients. In particular, within this second class of diffusions we have found
some explicitly solvable diffusion families with a nonlinear mean-reverting drift.
By combining the closed-form transition probability densities for the Doob-h transformed processes and
the fact that an underlying bridge process and its Doob-h transformed bridge process have equivalent probability laws, we derived some
closed-form integral formulas for conditional expectations of functionals involving the discount factor of the process and the process terminal value.
An explicit example involving the Squared Bessel process is given. The formulas are readily applicable to standard bond and bond option pricing.

This paper has also presented the construction of the Green functions for the Doob-h transformed processes and has given
a complete boundary classification of such processes that can generally have singular and/or non-singular (regular) endpoints.
Various closed-form transition densities for these transformed processes then followed simply by Laplace inverting the relevant Green functions.
For Doob-h transformed processes with imposed boundary conditions (e.g. regular killing or reflecting) at any interior point(s) of the diffusion,
the Laplace inversion formulation in this paper is also readily applicable. This method produces various closed-form spectral expansions of the
transition densities for such cases. Moreover, solvability (e.g. closed-form spectral expansions) is not restricted to only transition densities.
In fact, analytically closed-form spectral expansions of the densities and distributions for the
first hitting (exit) times at any interior level(s), for the extrema of the process, and for
various joint extrema and process terminal value, have recently been derived for all the families of Doob-h transformed
processes developed in this paper. The derivations of such closed-form spectral expansions and their applications in finance
is the subject of a related paper (\cite{Campolieti2008}).

For the first class of affine drift diffusions, yet another important component of this paper is Theorem \ref{theorem1}. In particular,
Corollary \ref{corollary1} provides us with a simple way to classify the respective discounted processes in terms of the martingale property.
This is of importance in the standard realm of arbitrage-free equity derivative pricing where discounted asset price processes are required to be
martingales in a given (risk-neutral) equivalent martingale measure.
We have presented three new explicit families of such solvable diffusions,
named \emph{Bessel}, \emph{confluent hypergeometric}, and \emph{Ornstein-Uhlenbeck} families.
These affine drift processes, having nonlinear (skew and smile-like) local volatility functions,
are useful for modeling asset prices and derivative pricing.
In particular, for a number of subfamilies of these models the discounted asset prices are martingales while for
other (dual) subfamilies the discounted asset prices are strict supermartingales. Analytically exact
closed-form expressions for transition densities and state-price
densities are obtained for these families of nonlinear local
volatility models in terms of known special functions (i.e., modified Bessel, confluent hypergeometric, hypergeometric).
Subfamilies of these new models nest the constant elasticity of variance (CEV) diffusion model
and other related models as special cases, and have been shown to exhibit a wide range of implied volatility surfaces with pronounced smiles and skews (see
\cite{CM06,CM07,CM08} for further details). Moreover, exact path sampling methods are available for all models presented in this paper (see \cite{CM07,MG10}). Examples of actual calibrations of the affine \emph{Ornstein-Uhlenbeck} family of affine drift diffusions to various market option data are contained in \cite{CMV11}.

\appendix

\section{Asymptotics of Fundamental Solutions and Wronskians} \label{sect_a1}
Throughout this appendix, we let $s$, $s_1$ and $s_2$ be
complex-valued parameters with positive real parts and
obtain the leading term asymptotic expressions for $\varphi^\pm_s(x)$, as $x$
approaches either left or right boundary point of the respective process. The expressions follow by using
known asymptotic formulas for either the modified Bessel, confluent hypergeometric,
or parabolic cylinder functions. The corresponding asymptotic forms for the Wronskians then follow:
for the OU process in A.3, these obtain directly from
the asymptotic expressions of $\varphi^\pm_s(x)$ and their derivatives,
whereas for A.1 - A.2 we first make use of
the differential recurrence relations of the modified Bessel (or Kummer functions)
and then apply corresponding asymptotic forms of $\varphi^\pm_s(x)$.

Throughout parts A.1, case (a) of A.2, and A.3 of this appendix we denote
$\upsilon \equiv s/\lambda_1$, $\upsilon_i\equiv s_i/\lambda_1$, $i=1,2$. Throughout this appendix we also adopt the usual notation
$\mu_{+}=\max\{0,\mu\}$, $\mu_{-}=\max\{0,-\mu\}$ and $\Psi$ is the standard digamma function.

\subsection{The SQB process} \label{subsect_a11}
 \begin{enumerate}[(i)]
   \item Asymptotic forms of the fundamental solutions defined by (\ref{SQBfund}):
     \begin{align*}
       &\varphi^+_s(x)\sim \textstyle{\frac{(2s/\nu_0^2)^{|\mu|/2}}{\Gamma(|\mu|+1)}}
       x^{\mu_{-}}, \quad
       \varphi^-_s(x)\sim \textstyle{\frac{\Gamma(|\mu|)x^{-\mu_+}}{2 (2s/\nu_0^2)^{|\mu|/2}}}\indfunc_{\mu\neq 0}+\frac{1}{2}\ln( x^{-1})\indfunc_{\mu=0}
       &\mbox{as } x\to 0\,,  \\
       &\varphi^+_s(x)\sim \textstyle{ \frac{1}{2\sqrt{\pi\sqrt{2s}/\nu_0}}} \ds\frac{e^{2\sqrt{2s x}/\nu_0}}{x^{\mu/2+1/4}}, \quad
       \varphi^-_s(x)\sim \textstyle{\frac{\sqrt{\pi}/2}{\sqrt{\sqrt{2s}/\nu_0}}} \ds\frac{e^{-2\sqrt{2s x}/\nu_0}}{x^{\mu/2+1/4}}
       &\mbox{as } x\to \infty\,,
     \end{align*}

   \item Asymptotic forms of the Wronskian functions as $x\to 0$:
     \begin{align*}
       W[\varphi^+_{s_1},\varphi^+_{s_2}](x) &\sim
        \left(\frac{2\sqrt{s_1s_2}}{\nu_0^2}\right)^{|\mu|}\,
        \frac{2(s_2-s_1)}{\nu_0^2\Gamma(|\mu|+1)\Gamma(|\mu|+2)}x^{2\mu_{-}}\,,
        \\
       W[\varphi^+_{s_1},\varphi^-_{s_2}](x) &\sim
        -\frac{1}{2}\left(\frac{s_1}{s_2}\right)^{|\mu|/2}x^{-\mu-1}\,,
        \\
       W[\varphi^-_{s_1},\varphi^-_{s_2}](x) &\sim
        \left\{\begin{array}{ll}
         \frac{\Gamma(|\mu|)\Gamma(|\mu|-1)}{2\nu_0^2}\,
         \left(\frac{\nu_0^2}{2\sqrt{s_1s_2}}\right)^{|\mu|}\,
         (s_1-s_2)\,x^{-2\mu_+} & \mbox{ if }|\mu|>1\,,\\
         \frac{s_1-s_2}{4\sqrt{s_1s_2}}\,
          x^{-\mu-1}\,\ln\left(\frac{1}{x}\right)& \mbox{ if }|\mu|=1\,,\\
         \frac{1}{4}\,(\ln s_1-\ln s_2)x^{-1}& \mbox{ if }\mu=0\,,\\
         \frac{\Gamma(|\mu|)\Gamma(1-|\mu|)}{4}\,
         \frac{s_1^{|\mu|}-s_2^{|\mu|}}{(s_1s_2)^{|\mu|/2}}\,
          x^{-\mu-1}& \mbox{ if }|\mu|\in(0,1)\,.\\
         \end{array}\right.
     \end{align*}
   \item Asymptotic forms of the Wronskian functions as $x\to
   \infty$:
     \begin{align*}
       W[\varphi^+_{s_1},\varphi^+_{s_2}](x) &\sim
        \frac{1}{4\pi}\frac{\sqrt{s_2}-\sqrt{s_1}}{(s_1s_2)^{1/4}}
        x^{-\mu-1}\,
        e^{2(\sqrt{s_1}+\sqrt{s_2})\sqrt{2 x}/\nu_0}\,,
        \\
       W[\varphi^+_{s_1},\varphi^-_{s_2}](x) &\sim
        -{1\over 4}\frac{\sqrt{s_1}+\sqrt{s_2}}{(s_1s_2)^{1/4}}
        x^{-\mu-1}e^{2(\sqrt{s_1}-\sqrt{s_2})\sqrt{2 x}/\nu_0}\,,
        \\
       W[\varphi^-_{s_1},\varphi^-_{s_2}](x) &\sim
        {\pi\over 4}\frac{\sqrt{s_1}-\sqrt{s_2}}{(s_1s_2)^{1/4}}
        x^{-\mu-1}\,
        e^{-2(\sqrt{s_1}+\sqrt{s_2})\sqrt{2 x}/\nu_0}\,.
     \end{align*}
 \end{enumerate}
 Note that for power functions of the form $s^a$, with $\mathrm{Re\,}s>0$ and $a>0$, the principal value is used.

\subsection{The CIR process} \label{subsect_a12}

{\bf Case: (a)} $\lambda_1 > 0$ ($\kappa > 0$).

 \begin{enumerate}[(i)]
   \item Asymptotic forms of the fundamental solutions defined by (\ref{CIRfund_1_M}) and (\ref{CIRfund_1_U}):
     \begin{align*}
       \varphi^+_s(x)\sim (\kappa x)^{\mu_-}, \qquad
       \varphi^-_s(x)\sim\left\{\begin{array}{ll}
       \frac{\Gamma(\mu)}{\Gamma(\upsilon)}(\kappa x)^{-\mu} & \mbox{ if }\mu>0\\
       \frac{1}{\Gamma(\upsilon)}\ln\left(\frac{1}{x}\right) & \mbox{ if }\mu=0\\
       \frac{\Gamma(-\mu)}{\Gamma(\upsilon-\mu)} & \mbox{ if }\mu<0 \end{array}\right.
       &\mbox{, as } x\to 0\,,  \\
       \varphi^+_s(x)\sim \frac{\Gamma(\vert\mu\vert+1)}{\Gamma(\upsilon + \mu_-)}e^{\kappa x}(\kappa x)^{\upsilon-\mu-1},\qquad
       \varphi^-_s(x)\sim (\kappa x)^{-\upsilon}
       &\mbox{, as } x\to \infty\,. \\
     \end{align*}
   \item Asymptotic forms of the Wronskian functions as $x\to 0$:
     \begin{align*}
       W[\varphi^+_{s_1},\varphi^+_{s_2}](x) &\sim \kappa
       \frac{(\upsilon_2-\upsilon_1)}{1+\vert\mu\vert} (\kappa x)^{2\mu_-}
       \,,\\
       W[\varphi^+_{s_1},\varphi^-_{s_2}](x) &\sim -\kappa
       \frac{\Gamma(1+\vert\mu\vert)}{\Gamma(\upsilon_2 + \mu_-)}(\kappa x)^{-\mu-1}
       \,,\\
       W[\varphi^-_{s_1},\varphi^-_{s_2}](x) &\sim
       \left\{ \begin{array}{ll}
          \frac{\Gamma(|\mu|)\Gamma(|\mu|-1)}{\Gamma(\upsilon_1+\mu_{-})\Gamma(\upsilon_2+\mu_{-})}\,
          \kappa(\upsilon_1-\upsilon_2)\,(\kappa x)^{-2\mu_+} & \mbox{ if }|\mu|>1\,,\\
          \frac{\kappa(\upsilon_1-\upsilon_2)}{\Gamma(\upsilon_1+\mu_{-}) \Gamma(\upsilon_2+\mu_{-})}\,
          (\kappa x)^{-2 + 2\mu_-}\ln\left(\frac{1}{x}\right)& \mbox{ if }|\mu|=1\,,\\
          \frac{\Psi(\upsilon_1)-\Psi(\upsilon_2)}{\Gamma(\upsilon_1)\Gamma(\upsilon_2)}x^{-1} & \mbox{ if }\mu=0\,,\\
          \frac{\kappa\Gamma(\vert\mu\vert)\Gamma(1-\vert\mu\vert)}{\Gamma(\upsilon_1 + \mu_-)\Gamma(\upsilon_2 + \mu_-)}\,
          \left(\frac{\Gamma(\upsilon_1 + \mu_-)}{\Gamma(\upsilon_1-\mu_+)}- \frac{\Gamma(\upsilon_2 + \mu_-)}{\Gamma(\upsilon_2-\mu_+)}
          \right)\, (\kappa x)^{-\mu-1} & \mbox{ if }|\mu|\in(0,1)\,.
       \end{array} \right.
     \end{align*}
   \item Asymptotic forms of the Wronskian functions as $x\to \infty$:
     \begin{align*}
       W[\varphi^+_{s_1},\varphi^+_{s_2}](x) &\sim
       \frac{\kappa \Gamma^2(1+\vert\mu\vert) (\upsilon_2 - \upsilon_1)}
	{\Gamma(\upsilon_1+\mu_-)\Gamma(\upsilon_2 + \mu_-)}e^{2\kappa x}(\kappa x)^{\upsilon_1+\upsilon_2-2\mu-3}
       \,,\\
       W[\varphi^+_{s_1},\varphi^-_{s_2}](x) &\sim
       -\kappa \frac{\Gamma(1 + \vert\mu\vert)}{\Gamma(\upsilon_1 + \mu_-)}e^{\kappa x}(\kappa x)^{\upsilon_1-\upsilon_2-\mu-1}
       \,,\\
       W[\varphi^-_{s_1},\varphi^-_{s_2}](x) &\sim
       \kappa(\upsilon_1-\upsilon_2)(\kappa x)^{-\upsilon_1-\upsilon_2-1}
       \,.
     \end{align*}
 \end{enumerate}

\noindent {\bf Case: (b)} $\lambda_1 < 0$ ($\kappa < 0$). In (i)--(iii) we denote $\upsilon\equiv s/\vert\lambda_1\vert$,
$\upsilon_i\equiv s_i/\vert\lambda_1\vert$, $i=1,2$.

 \begin{enumerate}[(i)]
   \item Asymptotic forms of the fundamental solutions defined by (\ref{CIRfund_2_M}) and (\ref{CIRfund_2_U}):
     \begin{align*}
       \varphi^+_s(x)\sim (\vert\kappa\vert x)^{\mu_-}, \qquad
       \varphi^-_s(x)\sim\left\{\begin{array}{ll}
       \frac{\Gamma(\mu)}{\Gamma(\upsilon + 1 + \mu)}(\vert\kappa\vert x)^{-\mu} & \mbox{ if }\mu>0\\
       \frac{1}{\Gamma(\upsilon + 1)}\ln\left(\frac{1}{x}\right) & \mbox{ if }\mu=0\\
       \frac{\Gamma(-\mu)}{\Gamma(\upsilon + 1)} & \mbox{ if }\mu<0 \end{array}\right.
       &\mbox{, as } x\to 0\,,  \\
       \varphi^+_s(x)\sim \frac{\Gamma(\vert\mu\vert+1)}{\Gamma(\upsilon + 1 + \mu_+)}
	 (\vert\kappa\vert x)^{\upsilon},\qquad
       \varphi^-_s(x) \sim (\vert\kappa\vert x)^{-\upsilon -1 -\mu}e^{-\vert\kappa\vert x}
       &\mbox{, as } x\to \infty\,. \\
     \end{align*}
   \item Asymptotic forms of the Wronskian functions as $x\to 0$:
     \begin{align*}
       W[\varphi^+_{s_1},\varphi^+_{s_2}](x) &\sim \vert\kappa\vert
       \frac{(\upsilon_2-\upsilon_1)}{1+\vert\mu\vert} (\vert\kappa\vert x)^{2\mu_-}
       \,,\\
       W[\varphi^+_{s_1},\varphi^-_{s_2}](x) &\sim -\vert\kappa\vert
       \frac{\Gamma(1+\vert\mu\vert)}{\Gamma(\upsilon_2 + 1 + \mu_+)}(\vert\kappa\vert x)^{-\mu-1}
       \,,\\
       W[\varphi^-_{s_1},\varphi^-_{s_2}](x) &\sim
       \left\{ \begin{array}{ll}
          \frac{\Gamma(|\mu|)\Gamma(|\mu|-1)}{\Gamma(\upsilon_1 + 1 + \mu_{+})\Gamma(\upsilon_2 + 1 + \mu_{+})}\,
          \vert\kappa\vert(\upsilon_1-\upsilon_2)\,(\vert\kappa\vert x)^{-2\mu_+} & \mbox{ if }|\mu|>1\,,\\
          \frac{\vert\kappa\vert(\upsilon_1-\upsilon_2)}{\Gamma(\upsilon_1 + 1 + \mu_{+}) \Gamma(\upsilon_2 + 1 + \mu_{+})}\,
          (\vert\kappa\vert x)^{-2 + 2\mu_-}\ln\left(\frac{1}{x}\right)& \mbox{ if }|\mu|=1\,,\\
          \frac{\Psi(\upsilon_1 + 1)-\Psi(\upsilon_2 + 1)}{\Gamma(\upsilon_1 + 1)\Gamma(\upsilon_2 + 1)}x^{-1} & \mbox{ if }\mu=0\,,\\
          \frac{\vert\kappa\vert\Gamma(\vert\mu\vert)\Gamma(1-\vert\mu\vert)}
	    {\Gamma(\upsilon_1 + 1 + \mu_+)\Gamma(\upsilon_2 + 1 + \mu_+)}\,
          \left(\frac{\Gamma(\upsilon_1 + 1 + \mu_+)}{\Gamma(\upsilon_1 +1 - \mu_-)}
		- \frac{\Gamma(\upsilon_2 + 1 + \mu_+)}{\Gamma(\upsilon_2 +1 -\mu_-)}
          \right)\, (\vert\kappa\vert x)^{-\mu-1} & \mbox{ if }|\mu|\in(0,1)\,.
       \end{array} \right.
     \end{align*}
   \item Asymptotic forms of the Wronskian functions as $x\to \infty$:
     \begin{align*}
       W[\varphi^+_{s_1},\varphi^+_{s_2}](x) &\sim
       \frac{\vert\kappa\vert \Gamma^2(1+\vert\mu\vert) (\upsilon_2 - \upsilon_1)}
	{\Gamma(\upsilon_1 + 1 + \mu_+)\Gamma(\upsilon_2 + 1 + \mu_+)}(\vert\kappa\vert x)^{\upsilon_1+\upsilon_2-1}
       \,,\\
       W[\varphi^+_{s_1},\varphi^-_{s_2}](x) &\sim
       -\vert\kappa\vert \frac{\Gamma(1 + \vert\mu\vert)}{\Gamma(\upsilon_1 + 1 + \mu_+)}
	  e^{-\vert\kappa\vert x}(\vert\kappa\vert x)^{\upsilon_1-\upsilon_2-\mu-1}
       \,,\\
       W[\varphi^-_{s_1},\varphi^-_{s_2}](x) &\sim
       \vert\kappa\vert(\upsilon_1-\upsilon_2)(\vert\kappa\vert x)^{-\upsilon_1-\upsilon_2-2\mu - 3}e^{-2\vert\kappa\vert x}
       \,.
     \end{align*}
 \end{enumerate}

\subsection{The Ornstein-Uhlenbeck process} \label{subsect_a13}
 \begin{enumerate}[(i)]
   \item Asymptotic forms of the fundamental solutions:
     \begin{align*}
       \varphi^\pm_s(x)\sim \frac{\sqrt{2\pi}}{\Gamma(\upsilon)}\,
       (\sqrt{\kappa}|x|)^{\upsilon-1}\,e^{\kappa x^2/2}
       &\mbox{ , as } x\to \pm\infty\,,  \\
       \varphi^\pm_s(x)\sim (\sqrt{\kappa}|x|)^{-\upsilon}
       &\mbox{ , as } x\to \mp\infty\,.
     \end{align*}
   \item Asymptotic forms of the Wronskian functions as $x\to
   -\infty$:
     \begin{align*}
       W[\varphi^+_{s_1},\varphi^+_{s_2}](x) &\sim
       (\upsilon_2-\upsilon_1)\,\sqrt{\kappa}\,(\sqrt{\kappa}|x|)^{-\upsilon_1-\upsilon_2-1}
       \,,\\
       W[\varphi^+_{s_1},\varphi^-_{s_2}](x) &\sim
       -\frac{\sqrt{2\pi\kappa}}{\Gamma(\upsilon_2)}\,(\sqrt{\kappa}|x|)^{\upsilon_2-\upsilon_1}\,e^{\kappa
       x^2/2}
       \,,\\
       W[\varphi^-_{s_1},\varphi^-_{s_2}](x) &\sim
       (\upsilon_1-\upsilon_2)\,\frac{2\pi\sqrt{\kappa}}{\Gamma(\upsilon_1)\Gamma(\upsilon_2)}\,
       (\sqrt{\kappa}|x|)^{\upsilon_1+\upsilon_2-3}\,e^{\kappa x^2}
       \,.
     \end{align*}

   \item Asymptotic forms of the Wronskian functions as $x\to
   \infty$:
     \begin{align*}
       W[\varphi^+_{s_1},\varphi^+_{s_2}](x) &\sim
       (\upsilon_2-\upsilon_1)\,\frac{2\pi\sqrt{\kappa}}{\Gamma(\upsilon_1)\Gamma(\upsilon_2)}\,
       (\sqrt{\kappa}x)^{\upsilon_1+\upsilon_2-3}\,e^{\kappa x^2}
       \,,\\
       W[\varphi^+_{s_1},\varphi^-_{s_2}](x) &\sim
       -\frac{\sqrt{2\pi\kappa}}{\Gamma(\upsilon_1)}\,(\sqrt{\kappa}x)^{\upsilon_1-\upsilon_2}\,e^{\kappa
       x^2/2}
       \,,\\
       W[\varphi^-_{s_1},\varphi^-_{s_2}](x) &\sim
       (\upsilon_1-\upsilon_2)\,\sqrt{\kappa}\,(\sqrt{\kappa}x)^{-\upsilon_1-\upsilon_2-1}
       \,.
     \end{align*}
 \end{enumerate}

\subsection{Asymptotic Properties of Wronskians of the Fundamental Solutions} \label{subsect_a14}

\begin{proposition} \label{aprop1}
 Let $\rho_1$ and $\rho_2$ be positive real parameters.
 For the three underlying diffusions, namely, the SQB, CIR, and OU
 diffusions defined in Sections \ref{subsect2.1}--\ref{subsect2.3},
we have the following limits for the Wronskians of the fundamental solutions:
 \begin{align*}
   & W[\varphi^+_{\rho_1},\varphi^+_{\rho_2}](0+) =
     \left\{\begin{array}{ll} \sgn(\rho_2-\rho_1)\cdot C &\text{ if }\mu\ge 0\,,\\
     \sgn(\rho_2-\rho_1)\cdot 0 &\text{ if }\mu < 0\,, \end{array}\right.\text{ for SQB and CIR,} \\
   & W[\varphi^+_{\rho_1},\varphi^+_{\rho_2}](-\infty) =
     \sgn(\rho_2-\rho_1)\cdot 0\,,\text{ for OU,} \\
   & W[\varphi^+_{\rho_1},\varphi^-_{\rho_2}](0+) =
     \left\{\begin{array}{ll} -\infty &\text{ if } \mu>-1\,,\\-C &\text{ if }\mu=-1\,,\\-0 &\text{ if }\mu<-1\,, \end{array}\right.\text{ for SQB and CIR,} \\
   & W[\varphi^+_{\rho_1},\varphi^-_{\rho_2}](-\infty) =
     -\infty\,,\text{ for OU,} \\
 \end{align*}
 \begin{align*}
   & W[\varphi^-_{\rho_1},\varphi^-_{\rho_2}](0+) =
     \left\{\begin{array}{ll} -\sgn(\rho_2-\rho_1)\cdot\infty &\text{ if } \mu\geq-1\,,\\-\sgn(\rho_2-\rho_1)\cdot C &\text{ if }\mu<-1\,, \end{array}\right.\text{ for SQB and CIR,} \\
   & W[\varphi^-_{\rho_1},\varphi^-_{\rho_2}](-\infty) =
     -\sgn(\rho_2-\rho_1)\infty\,,\text{ for OU,} \\
   & W[\varphi^+_{\rho_1},\varphi^+_{\rho_2}](+\infty) = \left\{\begin{array}{ll} \sgn(\rho_2-\rho_1)\cdot \infty\,, &\text{for OU, SQB and
   CIR if $\kappa > 0$}\\
   & \text{or $(\kappa < 0 \,\,and\,\, \rho_1 + \rho_2 > 1$)}\,,\\
   \sgn(\rho_2-\rho_1)\cdot C & \text{for CIR if $\kappa < 0$ and $\rho_1 + \rho_2 = 1$}\,,\\
   \sgn(\rho_2-\rho_1)\cdot 0 & \text{for CIR if $\kappa < 0$ and $\rho_1 + \rho_2 < 1$} \end{array}\right.
   \\
   & W[\varphi^+_{\rho_1},\varphi^-_{\rho_2}](+\infty) = \left\{\begin{array}{ll} -\infty\,, &\text{ for OU, and CIR if $\kappa > 0$,}
   \\ -0\,, &\text{ for CIR if $\kappa < 0$,} \end{array}\right.
   \\
   & W[\varphi^+_{\rho_1},\varphi^-_{\rho_2}](+\infty) =
     \left\{\begin{array}{ll} -0 &\text{ if  $\rho_2>\rho_1$ or ($\rho_1=\rho_2$ and $\mu>-1$)}\,,\\
      -C &\text{ if $\rho_1 = \rho_2$ and $\mu=-1$}\,,\\
     -\infty &\text{ if $\rho_2 < \rho_1$ or ($\rho_1=\rho_2$ and $\mu<-1$)}\,, \end{array}\right.\text{ for SQB,} \\
   & W[\varphi^-_{\rho_1},\varphi^-_{\rho_2}](+\infty) =
     -\sgn(\rho_2-\rho_1)\cdot 0\,, \text{ for all cases of SQB, CIR and OU},
 \end{align*}
 with constant $C>0$ for the SQB and CIR diffusions. Here $\sgn(x)= +1 (-1)$ for $x>0$ ($x<0$) and we define
$\sgn(0)\equiv 0$ as well as $0\cdot\infty \equiv 0$. The convergences are monotonic, i.e.
$W(x)\to+0$ and $W(x)\to+\infty$ (or $W(x)\to-0$ and $W(x)\to-\infty$) where
$W(x)>0$ (or $W(x)<0$) as $x$ approaches an endpoint.
\end{proposition}
\begin{proof}
The limits follow from the asymptotics of the
Wronskian functions presented in Subsections~\ref{subsect_a11}--\ref{subsect_a13}
where $s_i=\rho_i$; $l=0$ for SQB and CIR, $l=-\infty$
for OU and $r=\infty$ for all three diffusions. For the CIR diffusion, the limiting value of
$W[\varphi^-_{\rho_1},\varphi^-_{\rho_2}](x)$ follows from Proposition~\ref{aprop2} for $|\mu|\in (0,1)$, and for $\mu=0$ we use
the fact that $\Psi(z)$ is strictly increasing for $z>0$. \end{proof}

\begin{proposition} \label{aprop2}
$R(x)=\frac{\Gamma(x)}{\Gamma(x-a)}$ is an increasing function of $x>0$ for any $a\in(0,1)$.
\end{proposition}
\begin{proof}
First, notice that $R(a)=0$ and $R(x)>0$ for $x>a$.
We have that $\frac{\partial \ln R(x)}{\partial x} = \Psi(x) - \Psi(x-a)$ for $x\in(a,\infty)$. Therefore, $R$ is increasing on $(a,\infty)$, since the digamma function $\Psi$ is an increasing function on $(0,\infty)$. Let $0<x<a$, hence $R(x)<0$. Take the logarithmic derivative of $-R(x)=\frac{\Gamma(x)}{\Gamma(x+(1-a))}(a-x)$ to obtain $\frac{\partial \ln (-R(x))}{\partial x} = \Psi(x) - \Psi(x+(1-a))+\frac{1}{x-a}<0$. Thus, $R(x)$ is increasing on $(0,a)$. Since $R$ is a continuous function on $(0,\infty)$, the assertion is proved.
\end{proof}

\section*{Acknowledgements}
The authors acknowledge the support of the Natural Sciences and
Engineering Research Council of Canada (NSERC) for discovery
research grants.

%%%%%%%%%%%%%%%%%%%%%%%%%%%%%%%%%%%%%%%%%%%%%%%%%%%%%%%%%%%%%%%%%%%%%%

%%%%%%%%%%%%%%%%%%%%%%%%%%%%%%%%%%%%%%%%%%%%%%%%%%%%%%%%%%%%%%%%%%%%%%

\end{document}